\documentclass[12pt]{article}

\usepackage{latexsym,amsmath,amssymb,amsthm,amsfonts,graphicx,natbib}
\usepackage{epsfig}
\usepackage{bm}
\usepackage[euler]{textgreek}
\usepackage{xcolor}
\usepackage{algorithm,algpseudocode}
\usepackage{hyperref}
\usepackage{multirow}
\usepackage[normalem]{ulem}
\usepackage[pagewise,mathlines]{lineno}
\usepackage{multirow}
\usepackage[caption=false]{subfig}

\include{def}
\def\dispmuskip{\thinmuskip= 3mu plus 0mu minus 2mu \medmuskip=  4mu plus 2mu minus 2mu \thickmuskip=5mu plus 5mu minus 2mu}
\def\textmuskip{\thinmuskip= 0mu                    \medmuskip=  1mu plus 1mu minus 1mu \thickmuskip=2mu plus 3mu minus 1mu}
\def\beq{\dispmuskip\begin{equation}}    \def\eeq{\end{equation}\textmuskip}
\def\beqn{\dispmuskip\begin{displaymath}}\def\eeqn{\end{displaymath}\textmuskip}
\def\bea{\dispmuskip\begin{eqnarray}}    \def\eea{\end{eqnarray}\textmuskip}
\def\bean{\dispmuskip\begin{eqnarray*}}  \def\eean{\end{eqnarray*}\textmuskip}

\newtheorem{theorem}{Theorem}

\newcommand{\eps}{\epsilon}
\newcommand{\wh}{\widehat}

\def\E{{\mathbb E}}                         
\def\V{{\mathbb V}}

\def\I{1\!\!1} 				

\def\F{{\cal F}}

\def\N{{\cal N}}

\newcommand{\bol}[1]{\textbf{#1}}

\def\E{{\rm E}}                         

\def\I{1\!\!1} 				

\def\F{{\cal F}}

\def\N{{\cal N}}

\def\Var{\text{\rm Var}}

\begin{document}
	
\title{Recurrent Conditional Heteroskedasticity\thanks{\textit{We would like to thank the editor and the three anonymous referees for their constructive comments and suggestions.}}}
\author{T.-N. Nguyen\thanks{\textit{Discipline of Business Analytics, The University of Sydney Business School.}}
\and M.-N. Tran\footnotemark[1]
\and R. Kohn\thanks{\textit{School of Economics, UNSW Business School. The research of Nguyen and Kohn was partially supported by the Australian Research Council's Centre of Excellence for Mathematical and Statistical Frontiers (ACEMS).
Tran is partially supported by the ARC Discovery Project DP200103015.}}
}

\maketitle
\begin{abstract}
We propose a new class of financial volatility models, called the REcurrent Conditional Heteroskedastic (RECH) models, to improve both in-sample analysis and out-of-sample forecasting of the traditional conditional heteroskedastic models. In particular, we incorporate auxiliary deterministic processes, governed by recurrent neural networks, into the conditional variance of the traditional conditional heteroskedastic models, e.g. GARCH-type models, to flexibly capture the dynamics of the underlying volatility.
RECH models can detect interesting effects in financial volatility overlooked by the existing conditional heteroskedastic models such as the GARCH, GJR and EGARCH. The new models often have good out-of-sample forecasts while still explaining well the stylized facts of financial volatility by retaining the well-established features of econometric GARCH-type models.
These properties are illustrated through simulation studies and applications to thirty-one stock indices and exchange rate data.
An user-friendly software package, together with the examples reported in the paper, is available at \url{https://github.com/vbayeslab}.


\end{abstract}

\section{Introduction}
Financial time series, e.g. currency exchange rates or stock returns, exhibit stylized facts such as volatility clustering and leverage effects. The volatility clustering phenomenon of financial time series refers to the observation that ``large changes tend to be followed by large changes, of either sign, and small changes tend to be followed by small changes" \citep{Benoit1967}. This behavior implies that the volatilities, i.e. the conditional standard deviations, of financial returns are observed to be highly autocorrelated and exhibit periods of both low and high volatility. The leverage effects exhibited in financial time series, on the other hand, relate to the observation that the negative and positive past returns have asymmetric effects on the volatility \citep{Black:1976}. More specifically, the current volatility tends to be larger following a previous negative shock, i.e. a return below its expected value, than a positive one of the same absolute value. The volatility clustering and leverage imply that the volatility of financial assets changes over time, i.e., being heteroskedastic.

Time-varying volatility is a key assumption in the volatility modeling literature. A large number of volatility models have been developed since \cite{Engle:1982} proposed the Autoregressive Conditional Heteroskedastic (ARCH) model, which allows the conditional variance, i.e, the squared volatility, to change over time as a {\it deterministic} function of the historical shocks while leaving the unconditional variance unchanged. The most successful extension of the ARCH model is the Generalized Autoregressive Conditional Heteroskedasticity (GARCH) model of \cite{Bollerslev1986}; it models the current conditional variance as a linear function of the past conditional variances and squared returns. The GARCH model together with the ARCH model and their variants define a class of models, often referred to as the conditional heteroskedastic or GARCH-type models, which use deterministic functions of historical information, e.g. the past returns and past conditional variances, to model the current conditional variance and are able to capture the stylized facts of financial time series. Another notable line of research in the volatility modeling literature focuses on the stochastic volatility (SV) model \citep{Taylor:1986} and its variants, which formulate the conditional variance using latent {\it stochastic} processes that do not directly involve the past returns. Our article is, however, interested in GARCH-type models which are probably more popular in the volatility modeling literature because it is much easier to estimate GARCH-type models than SV-type alternatives. See \cite{Koopman:2016} for a comprehensive comparison between the GARCH-type and SV-type models.

Recurrent neural networks (RNNs) in the Deep Learning literature are successfully used in a large number of industrial-level applications; e.g. language translation, image captioning, speech synthesis.
RNNs are well-known for their ability to efficiently capture the long-range memory  and non-linear serial dependence existing within various types of sequential data, and are considered as the state-of-the-art models for many sequence learning problems \citep{Lipton:2018}. See \cite{Goodfellow2016} for a comprehensive discussion of various types of neural network models (NNs) and their broad range of applications.
The recent success of RNNs has motivated econometricians to incorporate RNNs and other deep learning models into their econometric models.
There is a large amount of research along this line, with a focus on modelling the mean rather than the variance of financial asset returns; see, e.g., \cite{ZHANG:1998,ZHANG:2003}.
Leveraging the power of deep learning in volatility modelling is still somewhat overlooked in the econometrics literature.
\cite{Donaldson:1997} are one of the first to propose a NN-GARCH model that adds a feedforward NN component into the conditional variance of the GJR model \citep{Glosten:1993}.
Their in-sample and out-of-sample results on four stock markets suggest that the NN-GARCH model is preferred to several benchmark GARCH-type models.
\cite{Roh:2007} proposes NN-based volatility models
that first estimate the conditional variance by an econometric volatility model, then use these estimates as inputs to a feedforward neural network, which then non-linearly transforms these inputs to output the final estimate of the conditional variance.
\cite{Kim:2018} extend this idea by using the outputs from several GARCH-type models rather than a single one as the inputs to a recurrent neural network; see also \cite{Luo:2018}.
\cite{LIU201999} uses the Long Short-term Memory, a sophisticated RNN technique, for volatility modelling and reports its improved prediction over GARCH in two datasets.
In general, these hybrid models that combine neural networks and financial econometric models are empirically superior to several econometric models and plain neural network models, in terms of predictive performance.
However, these models are often engineering-oriented and ignore the interpretable aspects of econometrics volatility models and the important stylized facts of financial time series.
Some of these existing models use feedforward NNs rather than RNNs and thus might ignore the time effects in time series data.
Their design is also rather inflexible in the sense that they combine a particular econometric model with a particular NN.
It is important to design a more flexible framework that is easy to adapt to advances in both the deep learning and volatility modelling literature.

GARCH-type models are generally simple yet highly interpretable in the sense that they are designed to explain the distinct behaviors of financial time series. Any new volatility model should not overlook this interpretability of the traditional econometric models. This paper proposes a new class of models, called the REcurrent Conditional Heteroskedastic models (RECH), that not only improve the forecast performance of GARCH-type models, by leveraging the capability of learning non-linearity and long-range dependence of RNNs, but also place significant emphasis on the interpretation of the estimated volatility, by inheriting the well-established features of GARCH-type models.
We now briefly explain why RECH models fit well within the volatility modeling literature. First, similarly to GARCH-type models, the conditional variances in RECH models are a deterministic function of the past values; hence it is easy to estimate RECH models as their likelihood functions can be evaluated analytically.
Second, RECH models are still able to explain the stylized facts of the underlying volatility dynamics.
Third, by inheriting the predictive power from deep learning techniques, RECH models often forecast better. Fourth, the highly flexible design of RECH models makes it easy to adopt advances in both the deep learning and volatility modeling literatures, allowing it to be used in a wide range of applications in financial time series analysis.
A Matlab software package implementing Bayesian estimation and inference for RECH models together with the examples reported in this paper is available at \url{https://github.com/vbayeslab}.

The rest of the article is organized as follows. Section \ref{sec:convol} briefly reviews the GARCH model and its variants. Section \ref{sec:rech} briefly reviews different types of neural networks and proposes RECH models. Section \ref{Sec:Bayes} discusses Bayesian estimation and inference for RECH models. Section \ref{sec:simulation and applications} presents the simulation study and applies RECH models to analyze four benchmark financial datasets. Section \ref{sec:conclusion} concludes. The Appendix gives implementation details and further empirical results.

\section{Conditional heteroscedastic models}\label{sec:convol}
Let $y=\{y_t,\ t=1,...,T\}$ be a time series of demeaned returns and $\mathcal F_t$ be the $\sigma$-field of the information up to time $t$. Conditional heteroskedastic models represent the conditional variance $\sigma_t^2:=\Var(y_t|\mathcal F_{t-1})$ of the observation $y_t$ as a deterministic function of the observations and the conditional variances in the previous time steps. Mathematically, these models are expressed as:
\begin{subequations}
	\begin{align}
	y_t &= \sigma_t\eps_t,\;\;\eps_t\sim i.i.d \;\;\; \text{with} \;\;t=1,2,...,T, \label{eqn:CH_constant_1}\\
	\sigma_t^2 &= \omega + f(\sigma_{t-1}^2,...,\sigma_{t-p}^2,y_{t-1},...,y_{t-q},\theta); \label{eqn:CH_constant_2}
	\end{align}
\end{subequations}
$f(\cdot)$ is a positive deterministic function parameterized by the vector of unknown parameters $\theta$; $p,q \ge 0$ are the number of lags of $\sigma_t^2$ and $y_t$ respectively;
 and $\omega$ is a non-negative constant ensuring that the conditional variance $\sigma^2_t$ is positive.
The shocks $\eps_t$ are i.i.d. with zero mean and unit variance.

The GARCH model \citep{Bollerslev1986} formulates the conditional variance $\sigma^2_t$ as a linear combination of the previous returns and conditional variances in a ARMA($p,q$) form as:
\begin{subequations}
\begin{align}
y_t &= \sigma_t\eps_t,\;\;t=1,2,...,T,\label{eqn:GARCH1}\\
\sigma_t^2 &= \omega + \sum_{i=1}^{p}\alpha_i y_{t-i}^2 +\sum^{q}_{j=1}\beta_j \sigma_{t-j}^2,\;\;t=2,...,T;\label{eqn:GARCH2}
\end{align}
\end{subequations}
$\omega>0, \alpha_i,\beta_j \ge 0$, $i=1,...,p$, $j=1,...,q$ and $\sum_{i=1}^{p}\alpha_i+\sum^{q}_{j=1}\beta_j <1$ to ensure the stationarity of the GARCH process.

The structure of the conditional variance in \eqref{eqn:GARCH2} is symmetric in the sense that the conditional variance $\sigma^2_t$ does  not depend on the sign of the $y_t$, implying that the conventional GARCH($p,q$) model cannot capture the important leverage effect, i.e., $\sigma_t^2$ depends asymmetrically on the previous returns, in financial time series. \cite{Glosten:1993} propose a variant of the GARCH model, often called the GJR model, of the form
\begin{subequations}
\begin{align}
y_t &= \sigma_t\eps_t,\;\;t=1,2,...,T,\label{eqn:GJR1}\\
\sigma_t^2 &= \omega + \sum_{i=1}^{p}\alpha_i y_{t-i}^2  +\sum^{q}_{j=1}\beta_j  \sigma_{t-j}^2+ \sum^{p}_{i=1} \gamma_i \I[y_{t-i}<0]y_{t-i}^2;\label{eqn:GJR2}
\end{align}
\end{subequations}
$\omega>0, \alpha_i,\beta_j \ge 0, \alpha_i+\gamma_i \ge 0$, $i=1,...,p$, $j=1,...,q$ and $\sum_{i=1}^{p}\alpha_i+\sum^{q}_{j=1}\beta_j +\sum^{p}_{i=1}\gamma_i <1$ to ensure the stationarity of the $y_t$ process and the positivity of the conditional variance. The indicator function $\I[y_t<0]$ in \eqref{eqn:GJR2} equals $1$ if $y_t<0$ and is $0$ otherwise. In the GJR model, if $\gamma_i>0$, negative returns are more influential than positive returns.

Another popular GARCH-type model is the Exponential Garch (EGARCH) model of
\cite{Nelson:1991}
\begin{subequations}
	\begin{align}
	y_t &= \sigma_t\eps_t,\;\;t=1,2,...,T,\label{eqn:EGARCH1}\\
	\log \sigma_t^2 &= \omega + \sum_{i=1}^{p}\beta_i \log\sigma_{t-i}^2  +\sum^{q}_{j=1}\alpha_j \left[\frac{|y_{t-j}|}{\sigma_{t-j}}-\E\left\{\frac{|y_{t-j}|}{\sigma_{t-j}}\right\}\right]  + \sum^{q}_{j=1} \gamma_j \frac{y_{t-j}}{\sigma_{t-j}},\label{eqn:EGARCH2}
	\end{align}
\end{subequations}
where the roots of the polynomial $(1-\beta_1 L - ... -\beta_p L^p)$ must lie outside the unit circle to ensure the stationarity of the $y_t$ process.
By working on the log-scale,
the EGARCH model removes the positivity constraints on the model parameters
and states the leverage terms in Eq. \eqref{eqn:EGARCH2} to capture the asymmetry in volatility clustering.
See \cite{Poon:2003} and \cite{Bollerslev:2008} for a comprehensive discussion of the family of GARCH models and their
properties. We will use GARCH, GJR and EGARCH as the benchmark econometric models to compare against RECH models, because they are widely used in the volatility modelling literature.

\section{Recurrent conditional heteroskedastic models}\label{sec:rech}

\subsection{Recurrent neural network models}
This section denotes the time series data as $\{D_t=(x_t,z_t),t=1,2,..\}$, where $x_t=(x_{t,1},...,x_{t,K})^\top$ is the vector of inputs and $z_t$ the scalar output.
For the sequence $\{x_t\}$, $x_{i:j}$ denotes $(x_i,...,x_j)$ for $i\leq j$.
The goal of recurrent neural network models is to model the conditional distribution $p(z_t|x_t,D_{1:t-1})$.

There are several standard time series models.
One approach is to represent time effects {\it explicitly} via some simple functions, often a linear function, of the lagged values of the time series.
This is the mainstream time series data analysis approach with the well-known ARIMA method \citep{Box:1976}.
This section considers an alternative approach representing time effects {\it implicitly} via latent variables that are designed to store the memory of the dynamics in the data.
These latent variables, also called hidden states, are updated recurrently using the information carried over by their values from the past and the information from the data at the current time.
Recurrent neural networks (RNNs), belonging to the second category, were first developed in cognitive science \citep{Elman:1990} and successfully used in machine learning.

If the serial dependence structure is ignored, then a feedforward neural network (FNN) can be used to
transform the raw input data $x_t$ into a set of hidden units $h_t$, often called {\it learned features}, for the purpose of explaining or predicting $z_t$. Figure \ref{fig:FNN} is a graphical representation of a FNN model with one hidden layer containing $L$ hidden units.

\begin{figure}[ht]
	\centering
	\includegraphics[scale=0.5]{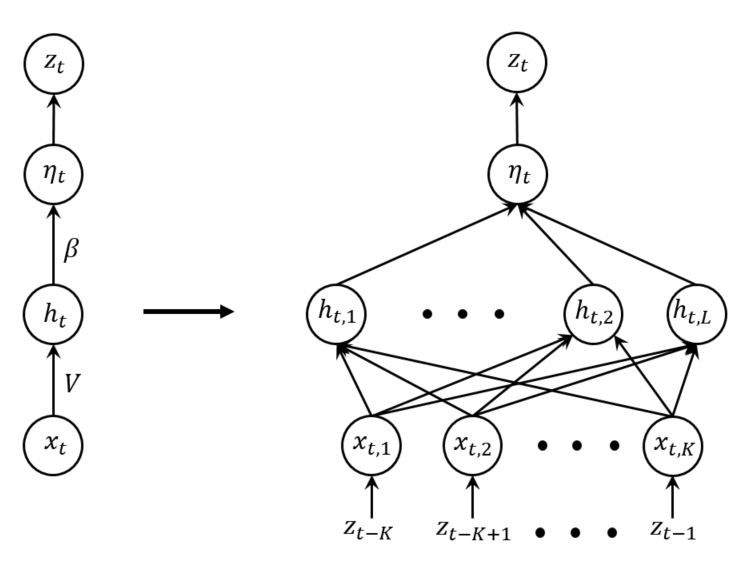}
	\caption{A FNN model with one hidden layer in compact {\it (left)} and explicit {\it {(right)}} styles.}.
	\label{fig:FNN}
\end{figure}

Given the FNN model in Figure \ref{fig:FNN}, the output $z_t$ is calculated as:
\begin{subequations}
	\begin{align}
	h_t     &=\phi(V x_t  + b),\label{eq:FNN1}\\
	\eta_t  &=\beta^\top h_t + \beta_0,\label{eq:FNN2}\\
	z_t|\eta_t&\sim p(z_t|\eta_t); \label{eq:FNN3}
	\end{align}
\end{subequations}
$V$ is a $L \times K$ matrix of weights connecting the input layer to the hidden layer; $\beta =(\beta_1,...,\beta_L)^\top$ is a vector of weights connecting the hidden layer to the output layer; $\beta_0$ is a scalar; $b=(b_1,...,b_L)^\top$ is a bias vector and $\phi(\cdot)$ is a non-linear scalar function, called the activation function. The scalar function $\phi$ is applied component-wise to a vector. In modern neural network modelling, the default recommendation for $\phi(\cdot)$ is to use the rectified linear unit \citep{Nair:2010,Le:2015}, or ReLU, having the form $\phi(z) = max\{0,z\}$.
The density $p(z_t|\eta_t)$ depends on the learning task.
For example, if $z_t$ is continuous, then typically $p(z_t|\eta_t)$ is a normal distribution with mean $\eta_t$ and variance $\sigma^2$;
if $z_t$ is binary, then $z_t|\eta_t$ follows a Bernoulli distribution with probability $\text{logit}^{-1}(\eta_t)$.

FNNs provide a powerful way to approximate the true function that maps the input $x_t$ to the mean $\E(z_t|x_t)$ or to transform the raw data $x_t$ into summary statistics $h_t$ having some desirable properties. However, FNNs are unsuitable for time series data analysis as the time effects and the serial correlations are ignored.
The main idea behind RNNs is to let the set of hidden units $h_t$ feed itself on its lagged value $h_{t-1}$.
Hence, an RNN can be best thought of as a FNN that allows a connection of the hidden units to their value from the previous time step, enabling the network to possess memory. This basic RNN model \citep{Elman:1990} can be written as:
\begin{subequations}
	\begin{align}
	h_t&=\phi(V x_t  + Wh_{t-1}+b),\label{eq:simpleRNN1}\\
	\eta_t&=\beta^\top h_t + \beta_0,\label{eq:simpleRNN2}\\
	z_t|\eta_t& \sim p(z_t|\eta_t); \label{eq:simpleRNN3}
	\end{align}
\end{subequations}
the parameters are the bias vector $b$, the bias scalar $\beta_0$, the weight matrices $V$, $W$, and $\beta$ for input-to-hidden, hidden-to-hidden and hidden-to-output connections, respectively. Similarly to FNNs, $\phi(\cdot)$ is a non-linear activation function; common choices are the ReLU or the sigmoid $\phi(z)=1/(1+e^{-z})$. Usually, we can set $h_1=0$, i.e. the neural network initially memoryless.

Figure \ref{f:RNN} \citep{Nguyen:2019} graphically illustrates the RNN model \eqref{eq:simpleRNN1}-\eqref{eq:simpleRNN3}. The circuit diagram (\textit{Left}) can be interpreted as an unfolded computational graph (\textit{Right}), where each node is associated with a particular time step.
The calculation of $h_t$ can be represented as a Simple Recurrent Neuron (SRN) unit, as Figure \ref{f:rnn_lstm_sru} shows, and we refer to \eqref{eq:simpleRNN1} as $h_t=\text{SRN}(x_t,h_{t-1})$, taking data $x_t$ at time $t$ and the previous state $h_{t-1}$ as the inputs. Using the SRN structure, the unfolded graph of the RNN model of \cite{Elman:1990}, which is normally referred to as the Simple RNN model, can be reinterpreted as the unfolded graph in Figure \ref{f:rnn_lstm_sru}(\textit{right}).

\begin{figure}[ht]
	\centering
	\includegraphics[width=0.6\columnwidth]{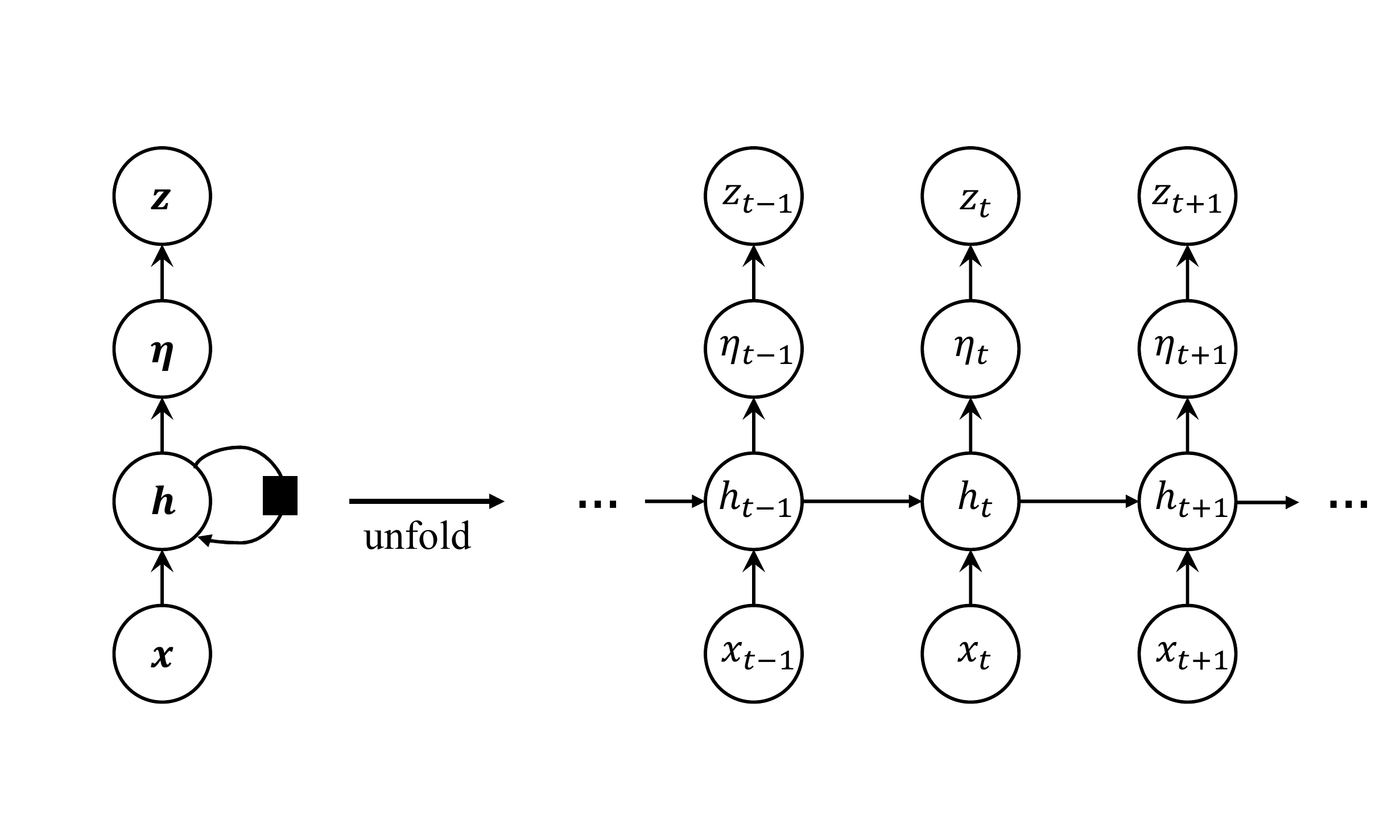}
	\caption{Graphical representation of the RNN model in \eqref{eq:simpleRNN1}-\eqref{eq:simpleRNN3}. }
	\label{f:RNN}
\end{figure}

\begin{figure}[h]
	\centering
	\includegraphics[width=0.6\columnwidth]{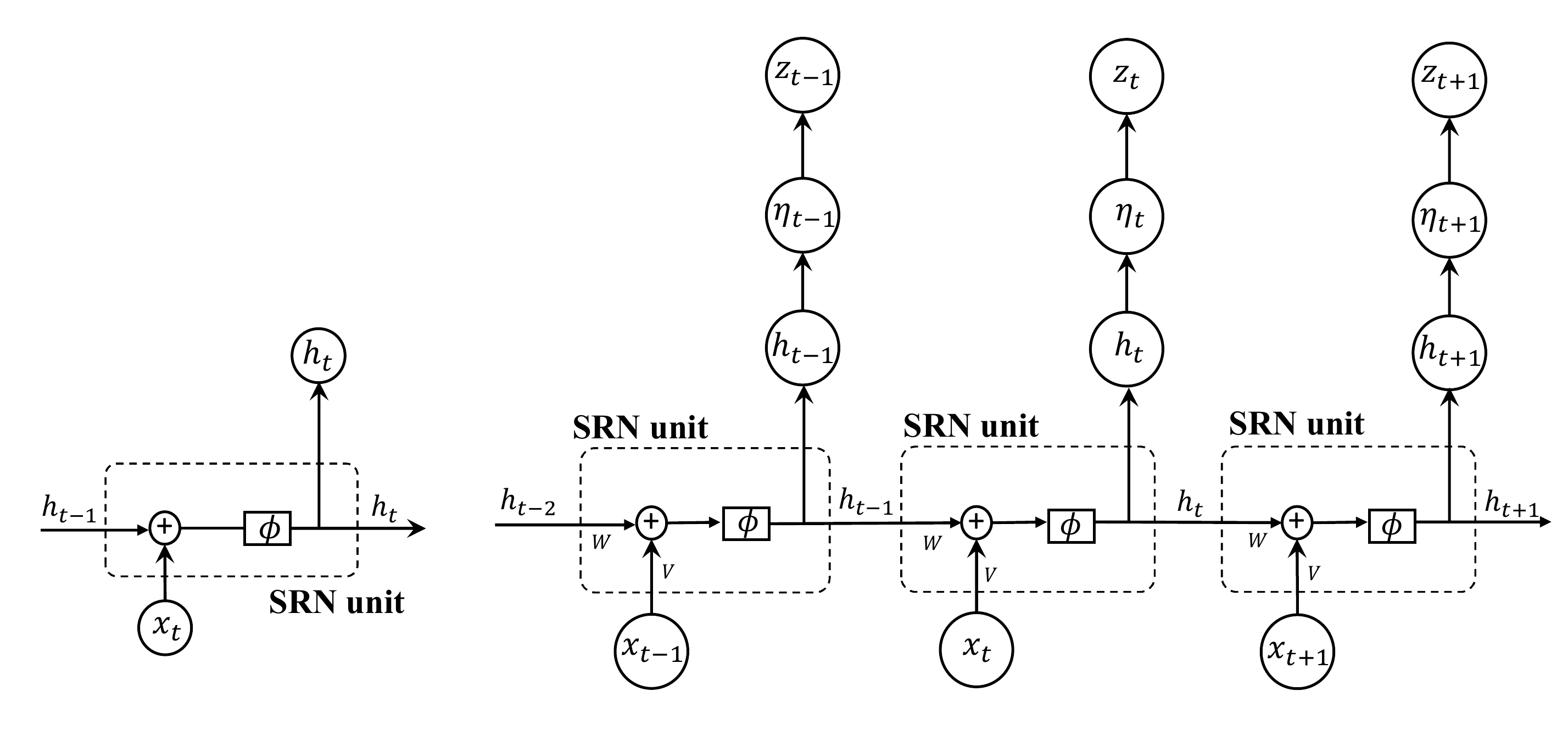}
	\caption{The structures of the SRN unit (\textit{left}) and the graphical representation of the Simple RNN model (\textit{right}), which uses the SRN unit to compute the latent state $h_t$.  The $\oplus$ symbols represents the addition operation.}
	\label{f:rnn_lstm_sru}
\end{figure}

There are more sophisticated recurrent neuron unit structures to compute $h_t$, such as the memory cell in the Long Short-Term Memory (LSTM) model of \cite{Hochreiter1997} and
the Statistical Recurrent Unit (SRU) of \cite{Oliva:2017}.
Below, we sometimes write $h_t=\text{RNN}(x_t,h_{t-1})$ to represent the way the hidden state $h_t$ is computed in an RNN model. Our article, however, only considers the SRN.

\subsection{Recurrent conditional heteroskedastic models}
\subsubsection{The general formulation}
The key motivation of RECH models is to allow the constant term $\omega$ in the general formulation of GARCH-type models in \eqref{eqn:CH_constant_1}-\eqref{eqn:CH_constant_2} to be driven
by an auxiliary deterministic process governed by an RNN, in order to capture complex dynamics such as non-linearity and long-term dependence that might not be captured efficiently by the GARCH-type component $f(\sigma_{t-1}^2,....,y_{t-q})$.
Our general RECH model is written as:
\begin{subequations}
\begin{align}
y_t &= \sigma_t\eps_t,\;\;\eps_t \sim i.i.d ,\;\;t=1,2,...,T  \label{eqn:RECH1}\\
\sigma_t^2 &= g(\omega_t)+ f(\sigma_{t-1}^2,...,\sigma_{t-p}^2,y_{t-1},...,y_{t-q}),\;\;t=2,...,T,\;\; \sigma^2_1=\sigma^2_0             \label{eqn:RECH2}\\
\omega_t&= \beta^\top h_t + \beta_0  ,\;\; t=2,...,T,\label{eqn:RECH3}\\
h_t&=\text{RNN}(x_t,h_{t-1}),\;\;t=2,...,T,\;\; \text{with} \; h_1 \equiv 0;\label{eqn:RECH4}
\end{align}
\end{subequations}
$g(\cdot)$ is a non-negative activation function; $p$ and $q$ are lag orders of $\sigma^2_t$ and $y_t$ respectively; $\beta_0$ is a scalar; $\beta=(\beta_1,...,\beta_L)$ is the weight vector with $L$ the number of hidden states. The reason an activation function is applied to $\omega_t$ is to ensure the conditional variance $\sigma^2_t$ is positive.
We refer to $\omega_t$ in \eqref{eqn:RECH2} as the {\it recurrent} component, as it is driven by a RNN, and $f(\cdot)$ as the {\it GARCH} component as this is formed based on the GARCH-type structures without the constant term.
Hence, we shall refer to the parameters of the recurrent component as the recurrent parameters, and refer to the parameters in $f(\cdot)$ as the GARCH parameters.
The recurrent state $h_t=\text{RNN}(x_t,h_{t-1})$ takes as its inputs the previous state $h_{t-1}$ and a vector of additional information $x_t$ whose choice is discussed shortly.

The conditional variance in \eqref{eqn:RECH2} is a sum of the recurrent and the GARCH components. This flexible design allows RECH models to enjoy many advances from both worlds of deep learning and volatility modelling.
Similarly to deep learning models, RECH models can use highly sophisticated neural network structures to capture complicated dynamics, e.g., long-range dependence and non-linearity, of the volatility dynamics and hence improve the forecasting of traditional GARCH-type models in applications where the underlying volatility dynamics exhibits long memory and nonlinearity.
Similarly to GARCH-type models, RECH models use simple yet interpretable structures to simulate important stylized facts in financial time series such as volatility clustering and leverage effects.
RECH models are well suited to modeling volatility because they inherit many properties from the GARCH-type models and distinguish themselves from the existing NN-based volatility models that often overlook the interpretability of the mainstream econometric models. As the recurrent and GARCH components are additive, increasing the complexity of the recurrent component $\omega_t$ will not decrease the interpretability of the GARCH component, and hence of RECH models.
The general formulation in \eqref{eqn:RECH1}-\eqref{eqn:RECH4} implies that most (if not all) variants of the GARCH models are nested in the RECH framework, because RECH models reduce to the corresponding GARCH-type models if $\beta=0$.

In previous work that combined an NN-based component with a GARCH-type model, \cite{Donaldson:1997} added a FNN-based component to a GJR model and reported some improvement of their FNN-GJR model compared to several benchmark GARCH-type models.
However, as discussed above, it might be inefficient to use a FNN to analyze time series data as FNNs are typically designed for cross-sectional data. Also, the estimation method of \cite{Donaldson:1997} uses randomized weights
for the FNN component rather than optimizing them; this technique is not recommended in the modern deep learning literature \citep{Goodfellow2016}.
\cite{Nguyen:2019} incorporate LSTM into the stochastic volatility models and name their model LSTM-SV.
LSTM-SV belongs to the class of parameter-driven models while RECH is an observation-driven model;
see \cite{Koopman:2016} for a comprehensive comparison between these two classes of models.
More specifically, the volatility dynamics in RECH is a {\it deterministic} process rather than a stochastic latent process as in LSTM-SV,
making it much easier to estimate RECH models than the LSTM-SV model as the likelihood of RECH models can be calculated analytically.
The RNN component of the LSTM-SV model takes only its past values as inputs while the recurrent component $\omega_t$ of RECH models allows any information including past observations as inputs.
We show below that this flexibility enables RECH models with SRN to capture complicated dynamics and the leverage effect in financial time series without needing to use complicated RNN structures such as LSTM.
Finally, like GARCH and SV, RECH and LSTM-SV complement each other and offer different perspectives towards the volatility modelling problem. 
The Appendix compares the performance of the RECH and LSTM-SV models.

\subsubsection{Specifications for RECH models}
\label{sec:RECH_specifications}
The RECH framework is highly flexible because it can easily incorporate advances from both the deep learning and volatility modeling literatures to design the recurrent and GARCH components, respectively.
For example, by using the SRN structure for the recurrent component $\omega_t$ and the conditional variance structure of the GARCH(1,1) model for the GARCH component, we obtain the SRN-GARCH specification of RECH model as:
\begin{subequations}
\begin{align}
y_t &= \sigma_t\eps_t,\;\;\eps_t\stackrel{iid}{\sim}\N(0,1),\;\;t=1,2,...,T  \label{RNN-GARCH1}\\
\sigma_t^2 &= \omega_t+ \alpha y_{t-1}^2 + \beta \sigma_{t-1}^2,\;\;t=2,...,T,\;\; \sigma^2_1=\sigma^2_0             \label{RNN-GARCH2}\\
\omega_t&=\beta_0+\beta_1 h_t,\;\; t=2,...,T,\label{RNN-GARCH3}\\
h_t&=\text{SRN}(x_{t},h_{t-1}),\;\;t=2,...,T,\;\; \text{with} \;\; h_1 \equiv 0;\label{RNN-GARCH4}
\end{align}
\end{subequations}
$x_t$ is the input vector of the RNN at time $t$. Figure \ref{f:SRN-GARCH graph} graphically represents the SRN-GARCH model.
For simplicity, we consider the standard normal distribution $\N(0,1)$ for the errors $\eps_t$. 
Here, we have used a linear activation function for $g(\cdot)$ in Eq. \eqref{eqn:RECH2}, and set $\beta_0,\beta_1 \geq  0$ to ensure the positivity of the conditional variance.
Alternatively, one can use a positive activation function for $g(\cdot)$, such as the ReLu or sigmoid, and relax the positivity constraints for $\beta_0$ and $\beta_1$.
We follow the GARCH literature and put the stationarity and positivity constraints on the GARCH parameters $\alpha$ and $\beta$, i.e., $\alpha,\beta \ge 0$ and $\alpha+\beta<1$.
These constraints do not imply that the SRN-GARCH model is stationary, but might improve numerical stability in estimation.

The recurrent function \text{SRN} is
\[\text{SRN}(x_t,h_{t-1}):=\phi(v^\top x_t + w h_{t-1}+b),\]
with $\phi(\cdot)$ a non-linear activation function.
We use the bounded ReLU activation for $\phi(\cdot)$ as it is easier to train and often performs better than other alternatives in the deep learning literature \citep{LIEW2016718}.
The bounded activation also guarantees a finite unconditional volatility; see Theorem \ref{the:Theorem}.
The recurrent weight $w$ and offset term $b$ are scalars. By default, we use only one recurrent state, i.e. $h_t$ is a scalar; however, it is possible to extend the specification in \eqref{RNN-GARCH1}-\eqref{RNN-GARCH4} to the case where $h_t$ is a vector of hidden states.
There is no restriction on the choice of the input vector $x_t$; typically, $x_t$ should include variables such as covariates and past returns that are deemed useful for predicting volatility $\sigma_t^2$.
If there are no covariates, as in the applications below, our choice for $x_t$ is the vector of the past return $y_{t-1}$ and the past conditional variance $\sigma_{t-1}^2$.
We also found it useful to include in $x_t$ the past recurrent component $\omega_{t-1}$.
Hence, $x_t=(w_{t-1},y_{t-1},\sigma_{t-1}^2)^\top$ and the input weights are $v=(v_0,v_1,v_2)^\top$.

While a non-zero $\beta$ in \eqref{RNN-GARCH2} quantifies the linear dependence of the current conditional variance $\sigma_t^2$ on its past value $\sigma_{t-1}^2$, a non-zero weight $v_2$ quantifies the non-linear dependence of $\sigma_t^2$ on $\sigma_{t-1}^2$.
That is, the recurrent component allows non-linear dependence of $\sigma_t^2$ on $\sigma_{t-1}^2$.
It is also well perceived in the deep learning literature that RNNs are able to capture long-range dependence \citep{pmlr-v97-greaves-tunnell19a},
therefore the recurrent component can allow long-range dependence that $\{\sigma_{s}^2,s<t\}$ have on $\sigma_t^2$.

We note that the SRN-GARCH specification uses the conditional variance structure of the GARCH model, but is still able to capture the leverage effects of the volatility dynamics as the input vector $x_t$ includes the asymmetric leverage term $y_{t-1}$ itself, not $y_{t-1}^2$.
Section \ref{sec:simulation and applications} shows that this choice of the input vector $x_t$ makes the volatility estimated by RECH models less sensitive to the choice of the structure for the GARCH component.
For selecting the lags $p$ and $q$ in the GARCH component, we find that SRN-GARCH(1,1) often works well in almost all cases.
This is probably because the recurrent component $\omega_t$ is able to capture the long-range dependence \citep{pmlr-v97-greaves-tunnell19a}, hence larger lags are unnecessary.
This is also consistent with the observation in the financial econometrics literature that the GARCH(1,1) model often works the best among other GARCH models \citep{Hansen:2005}.
Below, by SRN-GARCH without mentioning the lags we mean  SRN-GARCH(1,1).
We note that the SRN-GARCH specification simplifies to the GARCH(1,1) model if $\beta_0 >0$ and $\beta_1=0$.

\begin{figure}[ht]
	\centering
	\includegraphics[width=0.7\columnwidth]{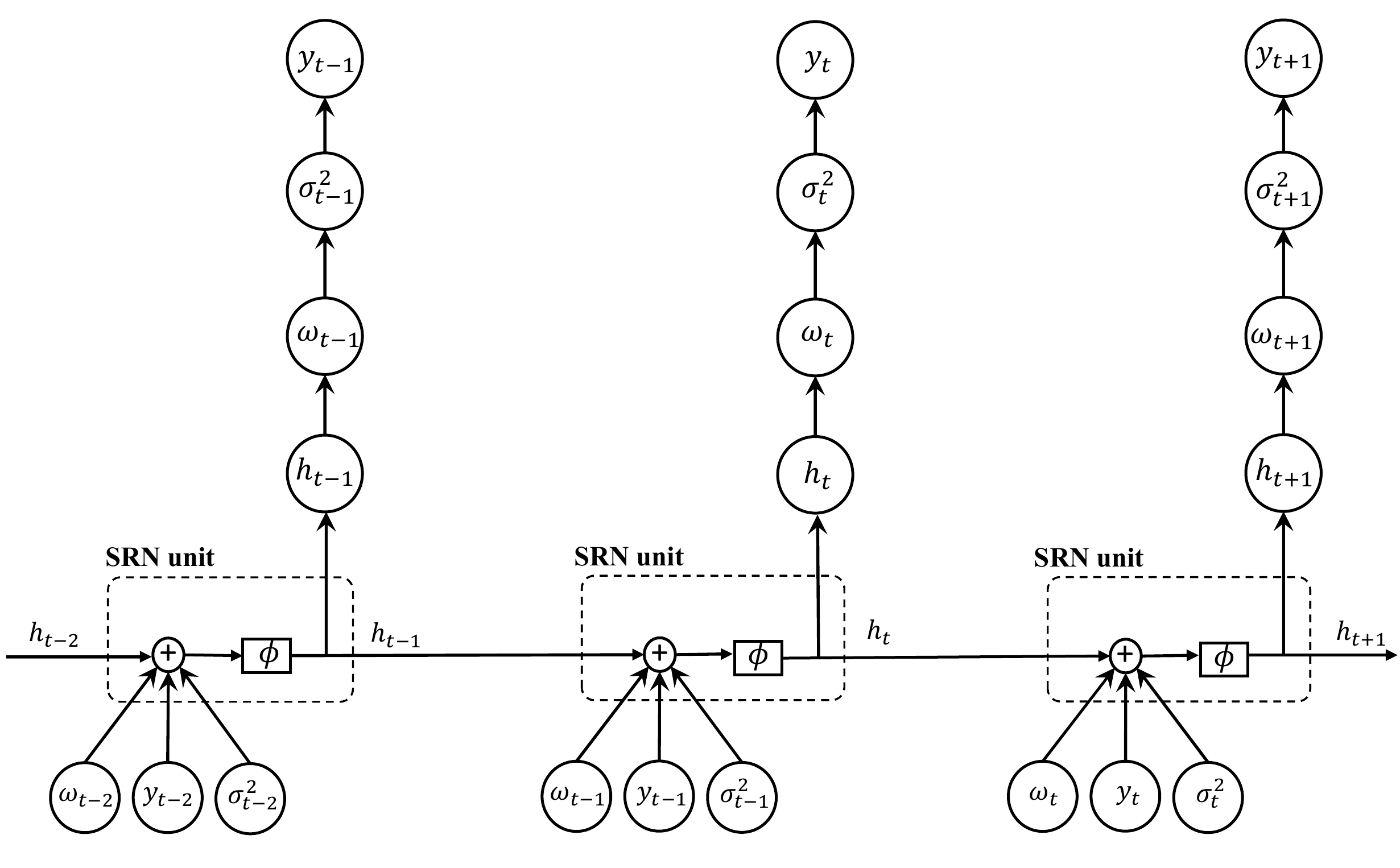}
	\caption{Graphical representation of SRN-GARCH.}
	\label{f:SRN-GARCH graph}
\end{figure}

Many other specifications of RECH models can be constructed. For example, the SRN-GJR specification is obtained by using the SRN structure for the recurrent component and the conditional variance structure of the GJR(1,1) for the GARCH component. Table \ref{tab:RECH_GARCH_models} presents several RECH model specifications. It is also possible to use other RNN structures such as LSTM \citep{Hochreiter1997} or SRU \citep{Oliva:2017} for the recurrent component.
It is worth noting that the GJR and EGARCH models accommodate the leverage effects as linear terms in the conditional variance equation and hence can only capture the linear dependence of the leverage effects.
RECH models, on the other hand, are able to capture other leverage dependence rather than the linearity, e.g. non-linearity or temporal dependence, of the leverage effects by allowing the leverage term $y_{t-1}$ to be an input of the RNN.

\begin{table}[h]
	\begin{center}
		\small
		\begin{tabular}{c|c|c}
			\hline\hline
			\rule{0pt}{3ex}
			\textbf{Models}          &\textbf{Conditional variance}     &\textbf{Constraints} \\
			\hline
			\rule{0pt}{3ex}
			\multirow{4}{*}{SRN-GJR}      &\multirow{3}{*}{$
			\begin{aligned}
			\sigma_t^2 &= \omega_t + \alpha  y_{t-1}^2 +\beta \sigma_{t-1}^2 + \gamma \text{I}_{[y_{t-1}<0]}y^2_{t-1}\\
			\omega_t&=\beta_0+\beta_1 h_t\\
			h_t&=\text{SRN}(x_{t},h_{t-1}), \; \; \text{with} \; \; x_t=(\omega_{t-1},y_{t-1},\sigma^2_{t-1})
			\end{aligned}$}                    &$\alpha,\beta \geq 0$\\
			&&$\alpha + \gamma \geq 0$ \\
			&&$\alpha + \beta + \gamma < 1$\\
			&&$\beta_0,\beta_1 \ge 0$\\
			\hline
			\rule{0pt}{3ex}
			\multirow{4}{*}{SRN-EGARCH}       &\multirow{2}{*}{$
			\begin{aligned}
			\sigma_t^2 &= \omega_t + \text{exp}\left\{\omega+\beta \text{log}\sigma_{t-1}^2 +\alpha \left [\frac{|y_{t-1}|}{\sigma_{t-1}} - \sqrt{\frac{2}{\pi}}\right]+\gamma\frac{y_{t-1}}{\sigma_{t-1}}\right\}\\
			\omega_t&=\beta_0+\beta_1 h_t\\
			h_t&=\text{SRN}(x_{t},h_{t-1}) \; \; \text{with} \; \; x_t=(\omega_{t-1},\text{e}^{y_{t-1}},\sigma^2_{t-1})
			\end{aligned}$  }                  &$0 \le \beta <1$\\
			&& \\	
			&& \\
			&& \\
			&& \\
			\hline\hline
		\end{tabular}
	\end{center}
	\caption{Several specifications of RECH models.}
	\label{tab:RECH_GARCH_models}
\end{table}

Given the general formulation of RECH models in \eqref{eqn:RECH1}-\eqref{eqn:RECH4}, its $\sigma^2_t$ process, and thus its $y_t$ process, is not guaranteed to be stationary unless $\beta_1=0$ and the GARCH parameters satisfy the stationary constraints of the corresponding GARCH components $f(\cdot)$. For example, the SRN-GARCH specification in \eqref{RNN-GARCH1}-\eqref{RNN-GARCH4} is stationary if $\beta_0, \alpha, \beta >0, \beta_1 = 0, \alpha+\beta<1$.
Although non-stationarity for volatility may be mathematically less appealing,
it is often argued to be more realistic in practice, e.g. \cite{Bellegem:2012}.
Theorem \ref{the:Theorem} below guarantees that the variance of $y_t$ is bounded.
\begin{theorem}\label{the:Theorem}
Consider the SRN-GARCH specification in \eqref{RNN-GARCH1}-\eqref{RNN-GARCH4}. Assume that
\begin{itemize}
\item $\alpha,\beta>0$, $\alpha+\beta<1$;
\item the recurrent component is bounded, i.e. $\omega_t\leq M$, almost surely for some $M<\infty$.
\end{itemize}
Then,
\[\V(y_t|\sigma_0^2)\leq\frac{M}{1-\alpha-\beta}+\sigma_0^2,\;\;\text{for all }\; t\geq 0.\]
That is, if the initial volatility $\sigma_0^2$ is finite almost surely, then all the subsequent $y_t$ have a finite variance almost surely.
\end{theorem}
The proof can be found in the Appendix.
The conditions in the theorem prevent the volatility $\sigma_t^2$ from exploding and not cause
numerical issues in training.
The condition that $\alpha>0,\beta>0$, $\alpha+\beta<1$ is standard in the GARCH literature and easy to impose. The second condition, i.e. the finite recurrent component condition, is imposed by using the bounded ReLU activation function (to bound $h_t$), and assuming a bounded support for $\beta_0$ and $\beta_1$ (we used uniform U(0,0.5) for these two parameters).

\section{Bayesian inference}\label{Sec:Bayes}
This section discusses Bayesian estimation and inference for RECH models.
We are interested in sampling from the posterior distribution
\bea
\label{eqn:Bayes_SSM}
\pi(\theta) = p(\theta|y_{1:T}) = \frac{p(y_{1:T}|\theta)p(\theta)}{p(y_{1:T})},
\eea
where $p(y_{1:T}|\theta)$ is the likelihood function, $p(\theta)$ is the prior and $p(y_{1:T}) = \int_{\Theta}p(y_{1:T}|\theta)p(\theta)d\theta$ is the marginal likelihood. Recall that the vector of model parameters $\theta$ consists of the recurrent and GARCH parameters. For example, the SRN-GARCH specification in \eqref{RNN-GARCH1}-\eqref{RNN-GARCH4} has the nine parameters $\theta=(\beta_0,\beta_1,\alpha,\beta,v_0,v_1,v_2,w,b)$.

\subsection{Sequential Monte Carlo (SMC)}
The SMC method is an attractive approach for Bayesian inference and forecasting in volatility modelling \citep{li:2019}.
SMC can sample efficiently from non-standard posteriors, provides the marginal likelihood estimate as a by-product, and is a convenient way for computing one-step-ahead forecasts.
In order to sample from the posterior $\pi(\theta)$, the SMC method \citep{Neal:2001,DelMoral:2006,Chopin:2002} first samples a set of $M$ weighted particles $\{W^j_0,\theta_0^j\}^M_{j=1}$ from an easy-to-sample distribution $\pi_0(\theta)$, such as the prior $p(\theta)$, and then traverses these particles through intermediate distributions $\pi_t(\theta), \;\; t=1,...,K$, which become the posterior distribution $\pi(\theta)$, i.e. $\pi_K(\theta)=\pi(\theta)$.
In our article, we set $\pi_0(\theta)=p(\theta)$ as the prior $p(\theta)$ if it is possible to sample from $p(\theta)$.
There are two common ways to design such a sequence of intermediate distributions: likelihood annealing \citep{Neal:2001} and data annealing \citep{Chopin:2002}.
The SMC with likelihood annealing uses the following intermediate distributions
\bea
\pi_t(\theta):=\pi_t(\theta|y_{1:T}) \propto p(y_{1:T}|\theta)^{\gamma_t}p(\theta),
\label{eqn:lik_anneal}
\eea
where $\gamma_t$ is referred to as the temperature level and $0 = \gamma_0 < \gamma_1 < \gamma_2 < ... < \gamma_K=1$.

\begin{algorithm}
	\caption{SMC with likelihood annealing for RECH models}
	\label{alg:lik_annealing}
		1. Sample  $\theta^j_0 \sim p(\theta)$ and set  $W_0^j=1/M$ for $j=1...M$ \\
		2. \textbf{For} $t=1,...,K$,
		\begin{itemize}
			\item[] \textbf{Step 1: }  \bol{Resampling}: Compute the unnormalized weights
			\beq\label{eq: SMC llh re-weight}
			w_t^j = W_{t-1}^j\frac{p(y_{1:T}|\theta_{t-1}^j)^{\gamma_t}p(\theta_{t-1}^j)}{p(y_{1:T}|\theta_{t-1}^j)^{\gamma_{t-1}}p(\theta_{t-1}^j)} = W_{t-1}^jp(y_{1:T}|\theta_{t-1}^j)^{\gamma_t - \gamma_{t-1}}, \; \; j=1,...,M
			\eeq
			and set the new normalized weights
			\bea
				W^j_t = \frac{w^j_t}{\sum_{s=1}^{M}w^s_t}, \; \; j=1,...,M.
			\eea			
			\item[] \textbf{Step 2: } Compute the effective sample size (ESS)
			\bea
				\text{ESS} = \frac{1}{\sum_{j=1}^{M} \left(W_t^j\right)^2}.
				\label{eq:ESS}
			\eea
			\begin{itemize}
			\item[] \textbf{if} $\text{ESS} < c M$ for some $0<c<1$, \textbf{then}
			\begin{itemize}
				\item[(i)] \bol{Resampling}: Resample from $\{\theta_{t-1}^j\}_{j=1}^M$ using the weights $\{W_{t}^j\}^M_{j=1}$, and then set $W_t^j=1/M$ for $j=1...M$, to obtain the new equally-weighted particles $\{\theta_{t}^j,W_{t}^j\}^M_{j=1}$.
				\item[(ii)] \bol{Markov move}: For each $j=1,...,M$, move the sample $\theta_t^j$ according to $N_{\text{lik}}$ random walk Metropolis-Hasting steps:
				\begin{itemize}
					\item[(a)] Generate a proposal $\theta_t^{j\prime}$ from a multivariate normal distribution $\N(\theta_t^j,\Sigma_t)$ with $\Sigma_t$ the covariance matrix.
					\item[(b)] Set $\theta_t^j = \theta_t^{j \prime}$ with the probability
					\bea
					\text{min}\left(1,\frac{p(y_{1:T}|\theta_t^{j \prime})^{\gamma_t}p(\theta_t^{j \prime})}{p(y_{1:T}|\theta_t^{j})^{\gamma_t}p(\theta_t^{j})}\right);
					\eea
					otherwise keep $\theta_t^j $ unchanged.
				\end{itemize}
			\end{itemize}
			\textbf{end}
			\end{itemize}
		\end{itemize}
		3. The log of the estimated marginal likelihood is
		\bea
		\text{log}\widehat{p(y_{1:T})} = \sum^K_{t=1} \text{log}\left(\sum_{j=1}^M w_t^j\right).
		\label{eq:marllh}
		\eea
\end{algorithm}

The SMC method consists of three main steps: reweighting, resampling and a Markov move.
There are various ways to implement SMC in practice; here we briefly present one of these.
At the begining of iteration $t$, the set of weighted particles $\{W_{t-1}^j,\theta_{t-1}^j\}^M_{j=1}$ that approximate the intermediate distribution $\pi_{t-1}(\theta)$ is reweighted to approximate the target $\pi_{t}(\theta)$.
The efficiency of these weighted particles as a representation of $\pi_{t}(\theta)$ is often measured by the effective sample size (ESS) \citep{Robert:1998,LiuChen:1998} defined in \eqref{eq:ESS}.
If the ESS is below a prespecified threshold, the particles are resampled; the resulting equally-weighted samples are then refreshed by a Markov kernel whose invariant distribution is $\pi_{t}(\theta)$.
Algorithm \ref{alg:lik_annealing} summarizes this SMC using the likelihood annealing method.
We follow \cite{Dang:2018} and choose the tempering sequence $\gamma_t$ adaptively to ensure a sufficient level of particle efficiency by selecting the next value of $\gamma_t$ such that ESS stays above a threshold.

SMC with likelihood annealing sampler is suitable for in-sample analysis, as it uses the sequence of distributions in \eqref{eqn:lik_anneal} which requires the full training data $y_{1:T}$ to be available. For out-of-sample rolling forecasts where the model parameters $\theta$ are updated once new data arrive, it is necessary to use SMC with the data annealing \citep{Chopin:2002}. This SMC sampler generates weighted particles from the following sequence of distributions
\bea
\pi_t(\theta):=\pi_t(\theta|y_{1:t}) \propto p(y_{1:t}|\theta)p(\theta)\propto \pi_{t-1}(\theta)p(y_t|\theta,y_{1:t-1}),
\label{eqn:data_anneal}
\eea
with $y_{1:t}$ the data available up to time $t$.
The unnormalized weights at the SMC step $t$ in \eqref{eq: SMC llh re-weight} become
\begin{equation}
w_t^j = W^j_{t-1}\frac{p(y_{1:t}|\theta^j_{t-1})p(\theta^j_{t-1})}{p(y_{1:t-1}|\theta^j_{t-1})p(\theta^j_{t-1})} = W^j_{t-1}p(y_t|y_{1:t-1},\theta^j_{t-1}), \; \; j=1,...,M.
\end{equation}
Algorithm \ref{alg:data_annealing} in the Appendix summarizes SMC with data annealing.
For RECH models, we use SMC with likelihood annealing for in-sample Bayesian analysis, and SMC with data annealing for out-of-sample analysis and forecasting.

\subsection{Model choice by marginal likelihood}\label{sec:marllh}
The marginal likelihood is often used to choose between models using the Bayes factor \citep{jeffreys:1935,Kass:1995}. In order to compare the relative performance between two models $M_1$ and $M_2$ on data $y_{1:T}$, we can use the Bayes factor
\beq
BF_{M_1,M2}=\frac{p(y_{1:T}|M_1)}{p(y_{1:T}|M_2)}.
\label{eq:bayesfactor}
\eeq
The larger the Bayes factor $BF_{M_1,M2}$, the stronger evidence that $M_1$ is more strongly supported by the data than $M_2$.
We note that the SMC with likelihood annealing sampler in the previous section provides an efficient way to compute the marginal likelihood.

\subsection{Runtime of the SMC with likelihood annealing sampler}
As SMC is parallelizable, the running time depends on how the algorithm is parallelized.
For example, we can run the algorithm on a single multi-core machine or multiple multi-core cluster machines. Table \ref{tab:cpu} shows the runtime of the SMC sampler, with $M=1000$ and $M=10000$ particles, when sampling the GARCH and SRN-GARCH models using only one core and six cores (numbers in parentheses). We run all examples on a standard laptop with moderate specification: Intel Core i7, 16GB RAM, 2.2GHz and 6 cores. We use
$M=1000$ in this paper as this value is sufficient to obtain consistent estimation results.

For comparison, the table also shows the running time of the MCMC sampler for the GARCH model by the R package \texttt{bayesGARCH} of \cite{Ardia:2010} with two different numbers of iterations $N=20,000$ and $N=200,000$, and the runtime of the MCMC sampler for the SV model using the R package \texttt{stochvol} of \cite{Hosszejni:2020} with $N=50,000$ and $N=500,000$.
These values of $M$ and $N$ are selected such that the standard errors of the SMC estimators (characterized by the effective sample size) are similar to that of the MCMC estimators (measured by the Integrated Autocorrelation Time, IACT). Here, we use the CODA R package of \cite{Plummer:2006} to compute the IACT of the MCMC chains obtained from the \texttt{bayesGARCH} and \texttt{stochvol} packages.

\begin{table}[h!]
	\small
	\begin{center}
		\begin{tabular}{c|cccccc}
			\hline\hline			
			\rule{0pt}{3ex}
			&\multicolumn{2}{c}{SMC}&\multicolumn{2}{c}{bayesGARCH}&\multicolumn{2}{c}{stochvol}\\
			&$M=1000$&$M=10000$&$N=20000$&$N=200000$&$N=50000$&$N=500000$\\
			\hline
			\rule{0pt}{3ex}
			SRN-GARCH	&150&1568 &&&&\\
			          	&(73)&(475)&&&&\\
			\rule{0pt}{3ex}
			GARCH			&24&154&38&405&&\\
			     			&(34)&(400)&&&&\\
			\rule{0pt}{3ex}
			SV					&&&&&133&1398\\
			\hline\hline
		\end{tabular}
	\end{center}
	\caption{The running time (in seconds) for the GARCH, SRN-GARCH and SV models for analysing the SP500 dataset, using various sampling methods. For the SMC sampler, the numbers in parentheses show the corresponding runtime when using all 6 cores.}
	\label{tab:cpu}
\end{table}

It is important to note that, unlike the MCMC sampler, SMC is parallelizable and hence its runtime can be reduced significantly when running on a computing cluster.
As shown in Table \ref{tab:cpu}, the SMC sampler for SRN-GARCH is quite computationally expensive when using a single CPU core. However, its runtime is significantly reduced when running in parallel with six cores. In all of our examples in Section \ref{sec:simulation and applications}, SMC was run on a high performance computing cluster.

\section{Simulation study and applications}\label{sec:simulation and applications}
This section evaluates the in-sample and out-of-sample performance of RECH models on simulation and stock return datasets.
We use the SMC with likelihood annealing to perform in-sample Bayesian inference and the SMC with data annealing to obtain one-step-ahead forecasts.
Table \ref{tab:implementation} lists our implementation settings of the SMC samplers.
\begin{table}[h]
	\begin{center}
		\begin{tabular}{clc}
			\hline\hline
			\rule{0pt}{3ex}
			\bol{Variable}      & \bol{Description}   &\bol{Value}\\
			\hline
			\rule{0pt}{3ex}
			$K$& Number of annealing levels    & 10000 \\
			$M$& Number of particles & 10000 \\
			$c$& Constant of the ESS threshold & 0.8 \\
			$N_{\text{lik}}$ & Number of Markov moves in the SMC with likelihood annealing& 30\\
			$N_{\text{data}}$ & Number of Markov moves in the SMC with data annealing& 30\\
			\hline\hline
		\end{tabular}
	\end{center}
	\caption{Implementation settings of the SMC samplers.}
	\label{tab:implementation}
\end{table}

We now discuss the selection of the prior distributions for RECH models. We use a normal prior with a zero mean and variance $0.1$ for the recurrence parameters $(v_0,v_1,v_2,w,b)$; as empirical results from the deep learning literature show that the values of the weights of neural networks are often small. As a linear activation function for $g(\cdot)$ is used, we put the uniform prior $U(0,0.5)$ on $\beta_0$ and $\beta_1$ to impose the positivity of $\omega_t$.
For the GARCH parameters, we choose the same priors as we use for the corresponding GARCH-type models, as summarized in Table \ref{table:prior}.

\begin{table}[ht]
	\begin{center}
		\begin{tabular}{c c c c c c c c}
			\hline\hline
			\multicolumn{3}{c}{\bol{GARCH(1,1)}} &\multicolumn{3}{c}{\bol{GJR(1,1)}} &\multicolumn{2}{c}{\bol{EGARCH(1,1)}}\\
			\cline{1-2}
			\cline{4-5}
			\cline{7-8}
			\rule{0pt}{3ex}
			Parameter       &Prior                    &   &Parameter       &Prior &         &Parameter       &Prior \\
			\hline
			\rule{0pt}{3ex}
			$\omega$  & $U(0,10)$ && $\omega$ & $U(0,10)$ & &$\omega$ & $\N(0,1)$ \\
			$\alpha$  & $U(0,1)$  && $\alpha$ & $U(0,1)$  & &$\alpha$ & $\N(0,1)$ \\
			$\beta$   & $U(0,1)$  && $\beta$  & $U(0,1)$  & &$\beta$  & $U(0,1)$ \\
			&           && $\gamma$ & $\N(0,0.1)$  & & $\gamma$    & $\N(0,0.1)$ \\
			\hline\hline
		\end{tabular}
	\end{center}
	\caption{Prior distributions for the parameters in the GARCH(1,1), GJR(1,1) and EGARCH(1,1) models. The notations $U$ and $\N$ denote the Uniform and Gaussian distributions, respectively.}
	\label{table:prior}
\end{table}

Next, we discuss the score metrics used to evaluate the out-of-sample performance. Denote by $D_{\text{test}}$ a test dataset, $T_{\text{test}}$ the number of observations in $D_{\text{test}}$ and $ \wh\theta$ the posterior mean estimate of $\theta$; we use four predictive scores to measure out of sample performance: the partial predictive score (PPS), the number of violations (\#Vio.), the quantile score (QS) and the hit percentage (\%Hit) to measure the out-of-sample performance. The PPS \citep{Gneiting2007} is evaluated on the test dataset $D_{test}$ as
\[PPS:=-\frac{1}{T_{test}}\sum_{D_{test}} \log \; p(y_t|y_{1:t-1},\wh\theta).\]
The model with smallest PPS is preferred.
The \#Vio. is defined as the number of times over the test data $D_{test}$ that the observation $y_t$ is outside its 99\% one-step-ahead forecast interval.
One of the main applications of volatility modelling is to forecast the Value at Risk (VaR). The $\alpha$-VaR is defined as the $\alpha$-quantile of the one-step-ahead forecast distribution $p(y_t|y_{1:t-1},\wh\theta)$.
The performance of a method producing VaR forecasts is often measured by the quantile score \citep{Taylor:2017} defined as
\[QS:=\frac{1}{T_{test}}\sum_{D_{test}} (\alpha-I_{y_t\leq q_{t,\alpha}})(y_t-q_{t,\alpha}),\]
where $q_{t,\alpha}$ is the $\alpha$-VaR forecast of $y_t$, conditional on $y_{1:t-1}$.
The smaller the quantile score, the better the VaR forecast. The \%Hit \citep{Taylor:2017} is defined as the percentage of the $y_t$ in the test data that is below its $\alpha$-VaR forecast.
The \%Hit is expected to be close to $\alpha$, if the model predicts well.
We note that these predictive performance measures complement each other. For example, it is possible to make the number of violations small by increasing the forecast volatility, but the PPS and QS scores then increase. A volatility model minimizing all three predictive scores, and having a hit percentage close to $\alpha$, is arguably the preferred one.

For the simulation data, given the true volatility $\sigma_t$, we follow \cite{Hansen:2005} and use six additional predictive scores as summarized in Table \ref{tab:forecast_score}.
\begin{table}[h]
	\begin{center}
		\footnotesize
		\begin{tabular}{ccc}
			\hline\hline
			\rule{0pt}{3ex}
			\bol{Score}         &&\bol{Definition}           \\
			\hline
			\rule{0pt}{3ex}
			$\text{MSE}_1$   &&$T_{test}^{-1}\sum_{D_{test}} (\sigma_t -\widehat{\sigma}_t)^2$  \\
			&&\\
			$\text{MSE}_2$   &&$T_{test}^{-1}\sum_{D_{test}} (\sigma_t^2 -\widehat{\sigma}_t^2)^2$  \\
			&&\\
			$\text{MAE}_1$   &&$T_{test}^{-1}\sum_{D_{test}} |\sigma_t -\widehat{\sigma}_t|$  \\
			&&\\
			$\text{MAE}_2$   &&$T_{test}^{-1}\sum_{D_{test}} |\sigma_t^2 -\widehat{\sigma}_t^2|$  \\
			&&\\
			$\text{QLIKE}$   &&$T_{test}^{-1}\sum_{D_{test}} \left(\text{log}(\widehat{\sigma}_t^2) + \sigma_t^2 \widehat{\sigma}_t^{-2}\right)$  \\
			&&\\
			$\text{R}^2\text{LOG}$   &&$T_{test}^{-1}\sum_{D_{test}} \left[\text{log}(\sigma_t^2 \widehat{\sigma}_t^{-2})\right]^2$  \\
			\hline\hline
		\end{tabular}
	\end{center}
	\caption{Definition of the predictive scores to measure the out-of-sample performance on simulation and index data. Here, $\widehat{\sigma}_t$ is an estimate of the volatility $\sigma_t$.}
	\label{tab:forecast_score}
\end{table}

\subsection{Simulation studies}\label{sec:Simulation}
\subsubsection{Simulation study I (SIM I)}
We generated a time series of $2000$ observations from the GARCH(1,1) model
\begin{subequations}
\begin{align}
y_t &= \sigma_t\eps_t,\;\;\eps_t\sim\N(0,1),\;\;\;\;t=1,...,T,\label{eq:Sim1GARCH1}\\
\sigma_t^2 &= 0.05 + 0.18 y_{t-1}^2 + 0.8 \sigma_{t-1}^2 ,\;\;\;\;t=2,...,T, \;\;\;\;\sigma^2_1 = 0.1.\label{eq:Sim1GARCH2}
\end{align}
\end{subequations}
The first $1000$ observations are used for model estimation and the last $1000$ observations for out-of-sample analysis. Table \ref{tab:Sim_params} shows the posterior means and the posterior standard deviations of the GARCH(1,1) and SRN-GARCH model parameters, obtained from the SMC using likelihood annealing.
Figure \ref{f:Sim_volatility_density_1} plots the estimated volatility together with the true volatility of the simulated data, i.e. the true values $\sigma_t^2$ generated from equation \eqref{eq:Sim1GARCH2}.

\begin{table}[h]
	\centering
	\footnotesize
	\begin{tabular}{rcccccccc}
		\hline\hline
		\rule{0pt}{3ex}
		& $\omega$ & $\alpha$ & $\beta$ &$\beta_0$&$\beta_1$&$v_1$&$v_2$&log ML\\
		\hline
		\rule{0pt}{3ex}
		\rule{0pt}{3ex}
		\multirow{2}{*}{}GARCH   &\bol{0.048} &\bol{0.178}  &\bol{0.791}   & & &&& $-1507.2$   \\
		&(0.022) &(0.029) &(0.035) &&& && (0.118)\\
		
		&&&&&&&&\\
		\rule{0pt}{3ex}
		\multirow{2}{*}{}SRN-GARCH &  &\bol{0.150}  &\bol{0.806} & \bol{0.056} &\bol{0.232}&\bol{0.154}&$-\bol{0.248}$&$-1507.3$\\
		& &(0.028) &(0.029)  &(0.022) &(0.152)  &(0.314)&(0.378)&(0.152) \\
		\hline\hline
	\end{tabular}
	\caption{SIM I: Posterior means (in bold) of the GARCH and SRN-GARCH model parameters with the posterior standard deviations in brackets. The last column shows the natural logarithms of the estimated marginal likelihood with the Monte Carlo standard errors in brackets, averaged over 10 different runs of the likelihood annealing SMC.}
	\label{tab:Sim_params}
\end{table}

\begin{figure}[h]
	\centering
	\includegraphics[width=1\columnwidth]{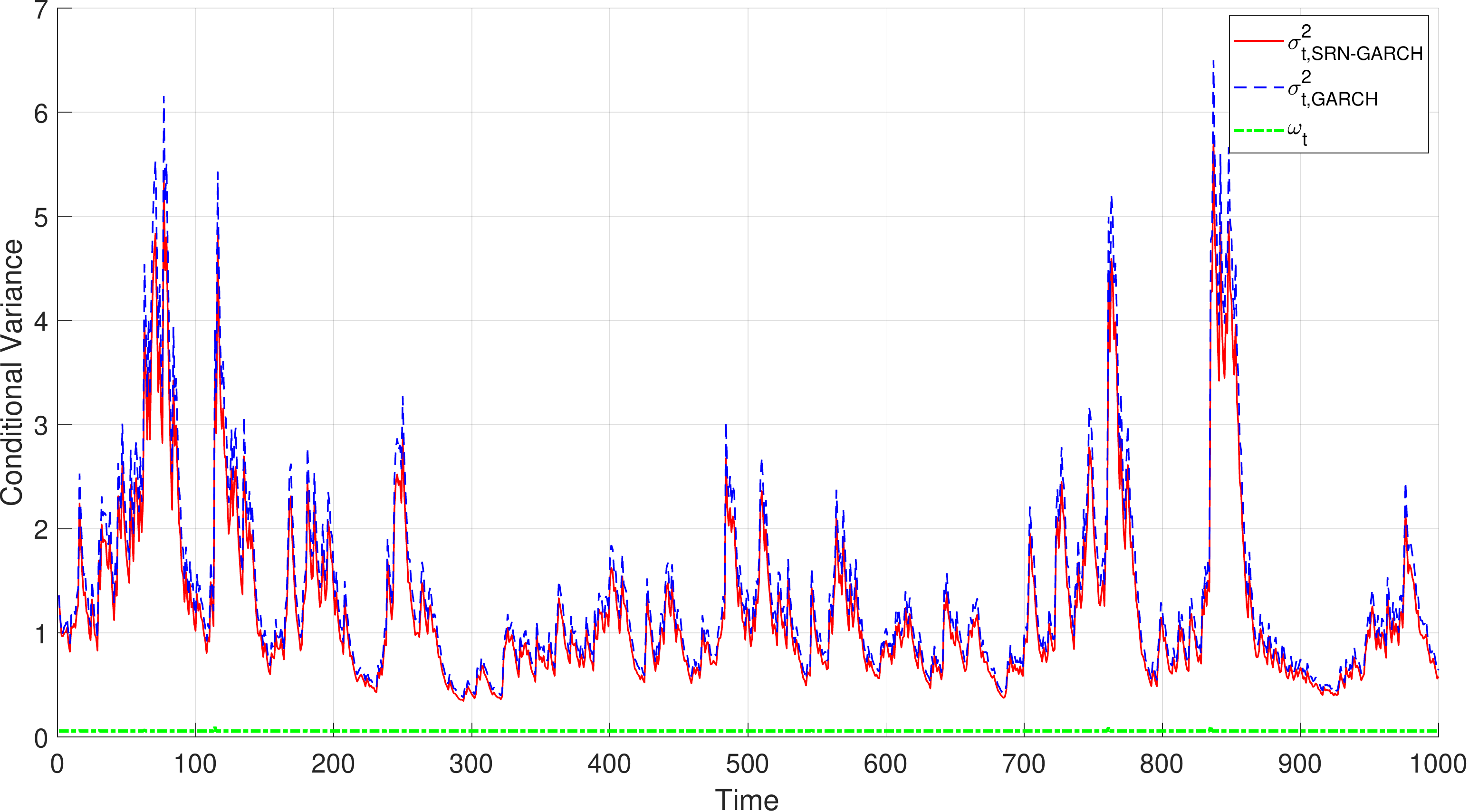}
	\caption{SIM I: The true conditional variance (dashed line) and the estimated conditional variance (solid line) using the SRN-GARCH model. The bottom line shows the values of the recurrent component $\omega_t$ of the SRN-GARCH specification at all time points. (The figure is better viewed in colour).}
	\label{f:Sim_volatility_density_1}
\end{figure}

The estimation results in Table \ref{tab:Sim_params} and the volatility plots in Figure \ref{f:Sim_volatility_density_1} suggest some important implications.
First, the GARCH parameters $\alpha$ and $\beta$ of the SRN-GARCH model in Table \ref{tab:Sim_params} are close to those of the GARCH model and the volatility estimated from the RECH model in Figure \ref{f:Sim_volatility_density_1} are close to the true volatility, implying that the estimated RECH model is close to the GARCH(1,1) model if the GARCH(1,1) is the true model. The almost identical estimates of the marginal likelihood in the last column of Table \ref{tab:Sim_params} also indicate that the GARCH(1,1) and RECH models provide an equally good fit to the data.
Second, the coefficient $\beta_1$ is statistically insignificant, i.e. the posterior mean is less than two standard deviations from zero, and the values of the recurrent component $\omega_t$ are consistently small at all time points as Figure \ref{f:Sim_volatility_density_1} shows, implying that there are no volatility effects other than linearity that are captured by the recurrent component $\omega_t$ and that the recurrent component $\omega_t$ contributes very little to the conditional variance at all time steps.
Additionally, the weights $v_1$ and $v_2$ of the inputs of the recurrent component are statistically insignificant, suggesting that there is no evidence of non-linearity, long range dependence and leverage effects within the data generating process GARCH(1,1).

\subsubsection{Simulation study II (SIM II)}
We generated a time series of $2000$ observations from the following non-linear GARCH-type model
\begin{subequations}
	\begin{align}
	y_t &= \sigma_t\eps_t,\;\;\eps_t\sim\N(0,1),\;\;t=1,2,...,T,\label{eq:Sim2GARCH1}\\
	\sigma_t^2 &= 0.05 + 0.10 y_{t-1}^2 + 0.21 \frac{y_{t-1}^2}{1+y_{t-1}^2} + 0.8 \sigma_{t-1}^2 +0.11 \frac{\sigma_{t-1}^2}{1+\sigma_{t-1}^2}+0.21 I_{[{y}_{t-1}<0]}y_{t-1}^2+0.1 \frac{I_{[{y}_{t-1}<0]}}{1+e^{-y_{t-1}^2}} .\label{eq:Sim2GARCH2}
	\end{align}
\end{subequations}
The model in \eqref{eq:Sim2GARCH1}-\eqref{eq:Sim2GARCH2} modifies the GJR(1,1) model by adding non-linear transformations of the past observation, conditional variance and leverage term to the equation of the conditional variance. The volatility evolution in \eqref{eq:Sim2GARCH2} suggests that the simulated volatility exhibits highly non-linear effects. The parameters of the model in \eqref{eq:Sim2GARCH1}-\eqref{eq:Sim2GARCH2} are set so that the simulated data somewhat resembles real financial time series data exhibiting both volatility clustering and leverage effects. The first $T = 1000$ observations are used for model estimation and the last $1000$ for out-of-sample analysis. Table \ref{tab:Sim2_params} shows the posterior means and standard deviations of the parameters from the GARCH(1,1), GJR(1,1), EGARCH(1,1) and three RECH counterparts.

\begin{table}[ht!]
	\centering
	\footnotesize
	\begin{tabular}{rccccccccc}
		\hline\hline
		\rule{0pt}{3ex}
		& $\omega$ & $\alpha$ & $\beta$ &$\gamma$&$\beta_0$ &$\beta_1$&$v_1$&$v_2$&Log ML\\
		\hline
		\rule{0pt}{3ex}
		\multirow{2}{*}{}GARCH   &\bol{0.854} &\bol{0.188}  &\bol{0.807}  & & &&&& $-3617.9$   \\
		&(0.124) &(0.018) &(0.018) &&& && & (0.164)\\
		\rule{0pt}{3ex}
		\multirow{2}{*}{}SRN-GARCH & &\bol{0.246}  &\bol{0.556}&&\bol{0.328}  &\bol{0.382}&-\bol{0.525}&\bol{0.396}&$-3614.9^\ast$\\
		&  &(0.048) &(0.158)  &&(0.120)&(0.100)  &(0.253)&(0.232)&(0.232) \\
		&&&&&&&&&\\
		\multirow{2}{*}{}GJR   &\bol{0.846} &\bol{0.204}  &\bol{0.801}   &\bol{-0.012} && &&& $-3620.1$   \\
		&(0.130) &(0.035) &(0.021) &(0.024)&& && & (0.141)\\
		\rule{0pt}{3ex}
		\multirow{2}{*}{}SRN-GJR &  &\bol{0.141}  &\bol{0.557}&\bol{0.179}&\bol{0.354} &\bol{0.373}&-\bol{0.180}&-\bol{0.418}&$-3611.6^\ast$\\
		& &(0.045) &(0.138)  &(0.064)&(0.110) &(0.092)  &(0.294)&(0.172)&(0.211) \\
		&&&&&&&&&\\
		\multirow{2}{*}{}EGARCH   &\bol{0.106} &\bol{0.373}  &\bol{0.975}   &\bol{-0.101} && &&& $-3616.3$   \\
		&(0.027) &(0.041) &(0.005) &(0.026)&& && & (0.147)\\
		\rule{0pt}{3ex}
		\multirow{2}{*}{}SRN-EGARCH &\bol{-0.058}  &\bol{0.450}  &\bol{0.976}&\bol{-0.114} &\bol{0.227}&\bol{0.270}&\bol{-0.028}&\bol{0.211}&$-3613.1^\ast$\\
		&(0.173)  &(0.145) &(0.017)  &(0.037)&(0.140) &(0.139)&(0.361)&(0.360)&(0.202) \\
		\hline\hline
	\end{tabular}
	\caption{SIM II: Posterior means (in bold) of the GARCH and RECH model parameters with the posterior standard deviations in brackets. The last column shows the natural logarithms of the estimated marginal likelihood with the Monte Carlo standard errors in brackets, averaged over 10 different runs of the SMC with likelihood annealing sampler. The asterisks indicate when the Bayes factors strongly support RECH models over their GARCH-type counterparts.}
	
	\label{tab:Sim2_params}
\end{table}

The estimation results from Table \ref{tab:Sim2_params} suggest the following. First, the posterior means of $\alpha$ and $\beta$ from the GARCH and GJR models are close to their true values, suggesting that the GARCH and GJR models can capture the linear serial dependence within the volatility dynamics of the data generating process. The constants $\omega$ of both the GARCH and GJR model are significantly inflated compared to the true value, possibly caused by the non-linear effects that cannot be captured by the GARCH and GJR models. The leverage parameter $\gamma$ of the GJR model is close to zero and statistically insignificant, implying that the GJR model cannot capture the leverage effect in the data generating process. The leverage parameter $\gamma$ of the EGARCH model is more than three standard deviations from zero, implying that the EGARCH model is the only benchmark GARCH-type model that can capture the simulated leverage effect.

Second, the coefficients $\beta_1$ of the SRN-GARCH and SRN-GJR models are more than three standard deviations from zero, implying that there is strong evidence of non-linearity in the volatility dynamics, and that the RNN structure within the recurrent component $\omega_t$ of the SRN-GARCH and SRN-GJR model is able to capture such dependence.
The weight $v_1$ with respect to the leverage input $y_{t-1}$ of the RNN in the SRN-GARCH model is more than two standard deviations from zero and the leverage parameters $\gamma$ of the SRN-GJR and SRN-EGARCH are more than three standard deviation from zero, all suggesting that the three RECH specifications can capture the leverage effects exhibited within the simulated volatility. The recurrent component $\omega_t$ of RECH models is useful in capturing the leverage effects overlooked by the GARCH and GJR models. Interestingly, the coefficient $v_2$ with respect to the input $\sigma_{t-1}^2$ in the SRN-GJR is more than two standard deviations from zero, indicating that the SRN-GJR is able to detect the non-linearity dependence of the past conditional variance on the current conditional variance $\sigma_t^2$, exhibited  within the simulated data generating process.
The analysis above suggests that observing the recurrent parameters of RECH models, e.g. $\beta_1, v_1$ and $v_2$, helps to detect the possible non-linearity effects within the underlying volatility.

Third, the estimates of marginal likelihood in the last column of Table \ref{tab:Sim2_params} show that RECH models consistently have higher marginal likelihood than their GARCH-type counterparts and that the SRN-GJR model provides the best fit to the data. The difference between the log marginal likelihood estimates is equivalent to the Bayes factors of the SRN-GARCH, SRN-GJR and SRN-EGARCH models compared to the GARCH, GJR and EGARCH models of roughly $e^3\approx 20.1$, $e^{8}\approx 3000$ and $e^{3}$, respectively, strongly supporting the RECH models. We note that among the benchmark GARCH-type models, the EGARCH model best fits to the SIM II data.

\begin{figure}[h]
	\centering
	\includegraphics[width=1\columnwidth]{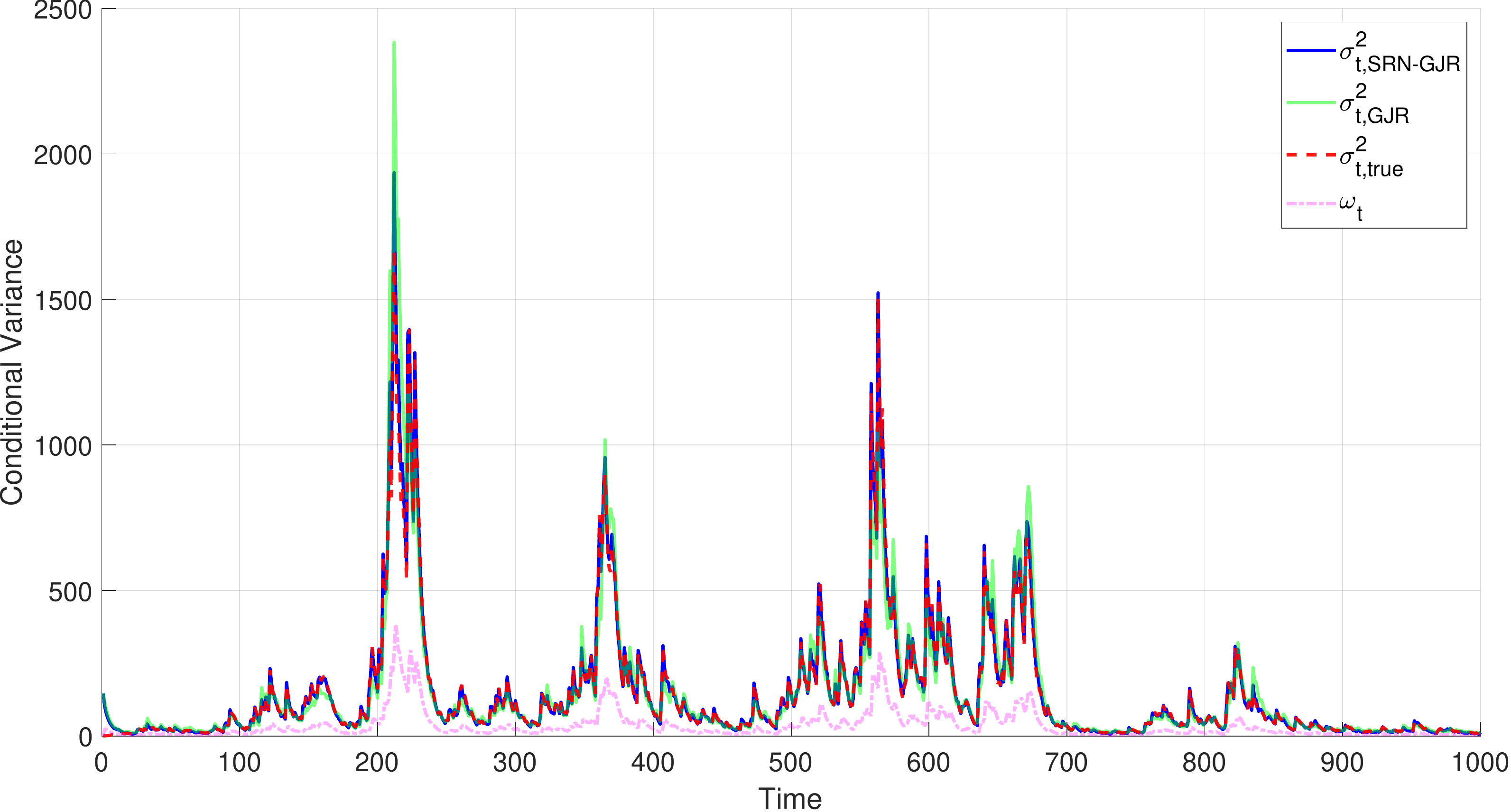}
	\caption{SIM II: The true conditional variance (dashed line) and the estimated conditional variance (solid lines) by the GJR and SRN-GJR models. The bottom line shows the values of the recurrent component $\omega_t$ of the SRN-GJR specification. The SRN-GJR plot appears to trace the true volatility plot better than the GJR plot.
 (The figure is better viewed in colour).}
	\label{f:Sim2_volatility_density_1}
\end{figure}

Figure \ref{f:Sim2_volatility_density_1} plots the values of the recurrent component and estimated volatility of the SRN-GJR model, together with the true volatility, at all time points. Figure \ref{f:Sim2_volatility_density_2} in the Appendix plots the volatility estimated by the GJR model and the true volatility. During the low volatility periods, i.e. time $t$ is between 0-100  or 900-1000, the volatility of the SRN-GJR model are very close to the true volatility, which is also the case for the GJR model in Figure \ref{f:Sim2_volatility_density_2}. During the periods when the volatility changed dramatically and oscillated highly, i.e. $t$ is between 200-250, 350-400 or 500-700, the volatility produced by the SRN-GJR model still tracks the true volatility well, but this is not the case for the GJR model. The GJR model often produces overly-large and overly-small volatility during these periods of abruptly changed volatility. The SRN-GJR model, on the other hand, appears to be able to capture well these changes. The plot of the recurrent component $\omega_t$ at all time points in Figure \ref{f:Sim2_volatility_density_1} shows that it is highly responsive to the changes in the true volatility.

Table \ref{tab:Sim_forecast} reports the forecast performance of the benchmark GARCH-type and the RECH models.
For the SIM I data, the forecast performance of the SRN-GARCH model is very close to that of the GARCH model, in all predictive scores, which again strongly supports the earlier in-sample conclusion that the SRN-GARCH model closely approximates the GARCH model if the GARCH is the true model.
For the SIM II data, Table \ref{tab:Sim_forecast} suggests some important results.
First, the RECH models outperform their GARCH-type counterparts in most of the predictive scores, which is consistent with the in-sample analysis showing that the RECH models fit better than their benchmark GARCH-type counterparts.
Second, the SRN-GJR model has the best forecast performance for all predictive measures and the SRN-EGARCH model also performs well on the SIM II data, which is consistent with the in-sample analysis showing that these two RECH specifications have the highest marginal likelihood estimates. Third, amongst the benchmark GARCH-type models, the EGARCH model has the best predictive performance, compared to the GARCH and GJR models.

\begin{table}[h]
	\begin{center}
		\footnotesize
		\begin{tabular}{rccccccccccc}
			\hline\hline
			\rule{0pt}{3ex}
			&PPS   &\#Vio&QS &\%Hit     &$\text{MSE}_1$   &$\text{MSE}_2$  &$\text{MAE}_1$   &$\text{MAE}_2$ &QLike& $\text{R}^2\text{Log}$ \\
			\hline
			\rule{0pt}{3ex}
			\textbf{SIM I} \\
			GARCH*     &1.794&09   &0.044&0.011  &0.001  &0.020&0.016&0.069&1.689&0.001\\
			\rule{0pt}{3ex}
			SRN-GARCH &1.795&09   &0.044&0.011  &0.001  &0.057 &0.022 &0.101&1.690&0.001\\
			\hline
			\rule{0pt}{3ex}
			\textbf{SIM II} \\
			\rule{0pt}{3ex}
			GARCH      &3.013&14 &0.145&0.008  &0.390  &88.970&0.485 &6.018&4.207&0.046\\
			SRN-GARCH*  &3.009&12 &0.142&0.008  &0.256  &85.151&0.349 &4.739&4.195&0.022\\
			\rule{0pt}{3ex}
			GJR        &3.011&16   &0.144& 0.008 &0.331  &77.430&0.445 &5.503&4.204&0.039\\
			SRN-GJR*    &\bol{3.003}&\bol{10}   &\bol{0.139}&\bol{0.009}  &\bol{0.023}  &\bol{5.071}&\bol{0.114} &\bol{1.339}&\bol{4.186}&\bol{0.003}\\
			\rule{0pt}{3ex}
			EGARCH  &3.005&\bol{10}&0.142&0.006&0.109&22.250&0.251&3.078&4.191&0.013\\
			SRN-EGARCH*&3.004&\bol{10}&0.140&0.008&0.054 &12.876&0.173 &2.153&4.188&0.007\\
			\hline\hline
		\end{tabular}
	\end{center}
	\caption{Simulation: one-step-ahead forecast comparison. For the QS and \%Hit measures, the results are calculated at the $1\%$-quantile. For the SIM II data, the bold numbers denote the best scores. For each pair of the RECH and GARCH-type models, the asterisk indicates the models having better forecast performance.}
	\label{tab:Sim_forecast}
\end{table}

\subsubsection{Simulation study III (SIM III)}
This simulation study examines if RECH models are able to simulate the long-memory volatility,
by fitting the FIGARCH$(1,d,1)$ model of \cite{Baillie:1996} to data generated from RECH models.
The FIGARCH$(1,d,1)$ model is defined as:
\begin{align}
		y_t &= \sigma_t\eps_t,\;\;\eps_t\sim \N(0,1), \;\;t=1,2,...,T, \nonumber\\
   	\sigma_t^2 &= \omega + \beta \sigma_{t-1}^2 + \left[ 1-\beta L - (1-\psi L)(1-L)^d\right]y_t^2, \nonumber
\end{align}
where $ d \in (0,1)$ is the fractional integrated parameter and $L$ is the backshift operator. The parameter $\Theta = (\omega,\psi,d,\beta)$.
When $d=0$, the FIGARCH becomes a GARCH model.
When $d>0$ and is close to 1, the persistence of the past shocks in the FIGARCH process decays at a slow hyperbolic rate \citep{Baillie:1996};
hence the FIGARCH process exhibits long-memory effects in its volatility dynamics.
		
\begin{table}[h]
	\centering
	\small
	\begin{tabular}{rccccccccc}
		\hline\hline
		\rule{0pt}{3ex}
		               & $\alpha$ & $\beta$ &$\beta_0$ &$\beta_1$& $v_0$ & $v_1$    & $v_2$ & $w$&$b$ \\
		\hline
		\rule{0pt}{3ex}
		$\theta_1$   & 0.058    & 0.681   & 0.068    & 0.418   & $-0.018$ & $-0.430$ & 0.524 & 0.161    & $-0.173$\\
		\rule{0pt}{3ex}
		$\theta_2$   & 0.071    & 0.690   & 0.075    & 0.362   & 0.062    & $-0.422$ & 0.538 & 0.087    & $-0.130$\\
		\rule{0pt}{3ex}
		$\theta_3$  & 0.076    & 0.744   & 0.016    & 0.388   & $-0.075$ & $-0.574$ & 0.400 & $-0.040$ & $-0.023$\\
		\rule{0pt}{3ex}
		$\theta_4$ & 0.057    & 0.562   & 0.101    & 0.413   & 0.015    & $-0.380$ & 0.652 & 0.270    & $-0.170$\\
		\hline\hline
	\end{tabular}
	\caption{SIM III: The parameters used in the DGP.}
	\label{tab:sim3_true_values}
\end{table}

We use the SRN-GARCH model as the true data generating process (DGP)
with the four different parameter sets  $\theta_i, \; i=1,..,4$, listed in Table \ref{tab:sim3_true_values}.
These are the estimated parameters obtained in Section \ref{sec:applications} when the SRN-GARCH model is fitted to four real datasets.
For each parameter set $\theta_i$, $500$ datasets of $T=3000$ observations are generated from each of two different specifications of the SRN-GARCH model: $\beta_1 = 0$ and $\beta_1$ equals to the true values in Table \ref{tab:sim3_true_values}, i.e. $\beta_1 \neq 0$. We note that if $\beta_1 = 0$, the DGP is the GARCH$(1,1)$ model.
To generate each time series, we generate $10,000$ observations and use the last $3,000$ for the simulation data.
We then use the Matlab MFE toolbox \footnote{https://github.com/bashtage/mfe-toolbox/}, with the default settings, to produce the Quasi-Maximum Likelihood Estimate (QMLE) of the parameter $\widehat\Theta_i, \;i=1,..,4$, of the FIGARCH$(1,d,1)$ model.
Table \ref{tab:sim3_params} shows the means and standard deviations averaged over $500$ QMLE estimates of the FIGARCH$(1,d,1)$ parameters.
\begin{table}[H]
	\centering
	\small
	\begin{tabular}{rccccccccc}
		\hline\hline
		\rule{0pt}{3ex}
		               & \multicolumn{4}{c}{DGP is GARCH$(1,1)$ $(\beta_1 = 0)$}              && \multicolumn{4}{c}{DGP is SRN-GARCH $(\beta_1 \neq0)$} \\
		               \cline{2-5}
		               \cline{7-10}
		\rule{0pt}{3ex}
		
		               & $\widehat\omega$ & $\widehat\psi$  &$\widehat{d}$       &$\widehat\beta$  && $\widehat\omega$ & $\widehat\psi$ & $\widehat{d}$     &$\widehat\beta$ \\
		\hline
		\rule{0pt}{3ex}
		$\widehat\Theta_1$  & 0.211    & 0.423   & 0.017    & 0.373   && 0.088    & 0.107 & 0.728       & 0.708   \\
		               & (0.043)  & (0.159) & (0.039)  & (0.127) && (0.014)  & (0.052) & (0.109)   & (0.078) \\
		\rule{0pt}{3ex}
		$\widehat\Theta_2$   & 0.225    & 0.366   & 0.036    & 0.323   && 0.126    & 0.152   & 0.608     & 0.647   \\
	                   & (0.064)  & (0.206) & (0.053)  & (0.155) && (0.020)  & (0.058) & (0.12)    & (0.100) \\
		\rule{0pt}{3ex}
		$\widehat\Theta_3$  & 0.045    & 0.210   & 0.093    & 0.222   && 0.095    & 0.105   & 0.711     & 0.653   \\
					   & (0.018)  & (0.214) & (0.059)  & (0.162) && (0.016)  & (0.053) & (0.110)   & (0.089) \\
		\rule{0pt}{3ex}
		$\widehat\Theta_4$ & 0.235    & 0.437   & 0.002    & 0.376   && 0.096    & 0.124   & 0.706     & 0.711   \\
					   & (0.019)  & (0.127) & (0.014)  & (0.125) && (0.017)  & (0.054) & (0.116)   & (0.084) \\
		\hline\hline
	\end{tabular}
	\caption[]{SIM III: Means and standard deviations (in brackets) of $500$ QMLE estimates of the FIGARCH$(1,d,1)$ model parameters.}
	\label{tab:sim3_params}
\end{table}

The important conclusion from Table \ref{tab:sim3_params} is that the short-memory and long-memory properties of the GARCH$(1,1)$ and SRN-GARCH models, respectively, are distinguishable. When the DGP is the GARCH$(1,1)$, the estimates of the fractional integrated parameter $\widehat{d}$ are insignificant in all cases, suggesting that there is no evidence of the long-memory effects in the volatility of the GARCH$(1,1)$ model.
When the DGP is the SRN-GARCH model, i.e. $\beta_1 \neq 0$, the estimates of the fractional integrated parameter $\widehat{d}$ are close to $1$ in all cases, implying the existence of long-memory in the volatility dynamics of simulated time series, which are generated from SRN-GARCH models. The difference in the QMLE estimates of the parameter $d$
between the two DGPs in Table \ref{tab:sim3_params} implies that the SRN-GARCH model is able to simulate long-memory volatility effects. We observe similar results for the SRN-GJR and SRN-EGARCH models.

\subsection{Applications to stock market returns}\label{sec:applications}
We demonstrate the performance of RECH models using four stock index datasets: the Standard and Poor's $500$ Index (SP500), the Japanese Nikkei 225 Index (N225), the Russell 2000 Index (RUT) and the German stock index (DAX). The datasets were downloaded from the Realized Library of The Oxford-Man Institute\footnote{https://realized.oxford-man.ox.ac.uk/}.
We used the daily closing prices $\{P_t, \ t=1,...,T_P\}$ and calculated the demeaned return process as
\bea
\label{eqn:log-return}
y_{t}=100\left(\log\frac{P_{t+1}}{P_{t}}-\frac1{T_P-1}\sum_{i=1}^{T_P-1}\log\frac{P_{i+1}}{P_{i}}\right),\;\;\;t=1,2,...,T_P-1.
\eea
The length of the four return series is fixed to be $T=4000$, with $T=T_P-1$, and each series is divided into an in-sample period of the first $T_{\text{in}}=2000$ observations and an out-of-sample period of the last $T_{\text{out}}=2000$ observations. Table \ref{tab:data summaries} summarizes the datasets.
\begin{table}[h]
	\begin{center}
		\footnotesize
		\begin{tabular}{ccccc}
			\hline\hline
			\rule{0pt}{3ex}
			&\bol{In-sample Period}         &\bol{Out-of-sample Period}           &$T_{\text{in}}$   & $T_{\text{out}}$\\
			\hline
			\rule{0pt}{3ex}
			SP500 &27 Feb 2004 -- 06 Feb 2012    &06 Feb 2012 -- 24 Jan 2020  &2000    &2000 \\
			N225    &16 Sep 2003 -- 16 Nov 2011    &17 Nov 2011 -- 24 Jan 2020  &2000    &2000                 \\
			RUT     &24 Feb 2004 -- 01 Feb 2012    &02 Feb 2012 -- 24 Jan 2020  &2000    &2000                  \\
			DAX    &06 Aug 2003 -- 02 Nov 2011    &02 Nov 2011 -- 24 Jan 2020   &2000    &2000                 \\
			\hline\hline
		\end{tabular}
	\end{center}
	\caption{Description of the four index datasets. }
	\label{tab:data summaries}
\end{table}

Table \ref{tab:data statistics} reports some descriptive statistics for these four datasets together with the modified R/S test \citep{Lo:1991} for long-range memory in the logarithm of the  squared returns. Lo's modified R/S test is widely used in the financial time series literature; see, e.g., \cite{Lo:1991}, \cite{Giraitis:2003}, \cite{Breidt:1998}.
All the index data exhibit some negative skewness, a high excess kurtosis and high variation. The N225 returns are more skewed and leptokurtic than those of the SP500, RUT and DAX data.
The result of Lo's modified R/S test for long-memory dependence with several different lags $q$ indicates that there is significant evidence of long-memory dependence in the SP500, RUT and DAX stock indices. For the N225 data, however, the evidence of long memory is less clear as the null hypothesis of short memory for the squared returns is not rejected at the 5\% level of significance when $q=20$ and $q=30$.

\begin{table}[H]
	\begin{center}
		\footnotesize
		\begin{tabular}{ccccccclll}
			\hline\hline
			\rule{0pt}{3ex}
			&Min         &Max      &Std         & Skew    & Kurtosis &$V_n(10)$& $V_n(20)$&$V_n(30)$\\
			\hline
			\rule{0pt}{3ex}
			\multirow{2}{*}{}SP500     &$-9.351$   &10.220   &1.307       &$-0.256$     & 12.502&3.188*&2.412*&2.047*\\
			&&&&&&2.664*&2.040*&1.748*\\
			\rule{0pt}{3ex}
			\multirow{2}{*}{}N225    &$-10.563$   &11.658    &1.224       &$-0.585$     & 18.171 &2.768*&2.171*&1.905*\\
			&&&&&&1.956*&1.566&1.415\\
			\rule{0pt}{3ex}
			\multirow{2}{*}{}RUT    &$-8.391$   &8.056    &1.364       &$-0.130$     & 8.764 &3.065*&2.385*&2.055*\\
			&&&&&&2.459*&1.943*&1.691\\
			\rule{0pt}{3ex}
			\multirow{2}{*}{}DAX    &$-7.437$   &9.993    &1.267       &0.115     & 10.960 &3.226*&2.501*&2.146*\\
			&&&&&&2.456*&1.926*&1.670\\
			\hline\hline
		\end{tabular}
	\end{center}
	\caption{Descriptive statistics for the demeaned returns of the SP500, N225, RUT and DAX datasets. $V_n(q), \ q=10, \ 20 \ \text{and} \ 30$, are the test statistics of Lo's modified R/S test of long memory with lag $q$. Upper and lower values of the 3 last columns are Lo's test statistics for absolute and squared returns, respectively. The asterisks indicate significance at the 5\% level.}
	\label{tab:data statistics}
\end{table}

The Realized Library provides different realized measures\footnote{See https://realized.oxford-man.ox.ac.uk/documentation/estimators for the list of the available realized measures} that can be used in financial econometrics as a proxy to the latent $\sigma_t^2$. We use the following six common realized measures including Realized Variance (RV) \citep{Andersen:1998}, Bipower Variation (BV) \citep{Nielsen:2004}, Median Realized Volatility (MedRV) \citep{Andersen:2012}, Realized Kernel Variance \citep{Barndorff&Nielsen:2008} with the Non-Flat Parzen kernal ($\text{RKV}_1$), the Tukey-Hanning kernal ($\text{RKV}_2$) and the Two-Scale/Bartlett kernal ($\text{RKV}_3$), to evaluate the forecast performance of the volatility models using the predictive scores in Table \ref{tab:forecast_score}. \cite{Shephard:2010} give more details about the Realized Library.

Denote by $RV_t$ the realized measure of $\sigma_t^2$ at time $t$. As the realized measures ignore the variation of the prices overnight and sometimes the variation in the first few minutes of the trading day when recorded prices may contain large errors \citep{Shephard:2010}, we follow \cite{Hansen:2005} to scale the realized measure $RV_t$ as
\beq
\label{eqn:realized_measure_scale}
\widetilde\sigma^2_t = \widehat{c} \cdot RV_t \;\; \text{where} \;\;\ \widehat{c} = \dfrac {T_\text{out}^{-1} \sum_{t=T_\text{in}+1}^{T}\left[y_t-\E(y_t|\mathcal F_{t-1})\right]^2}{T_\text{out}^{-1}\sum_{t=T_{in}+1}^{T} RV_t},\;\;\;t=T_\text{in}+1,2,...,T,
\eeq
with $\E(y_t|\mathcal F_{t-1}) = 0$, and use $\widetilde\sigma^2_t$ as the estimate of the latent conditional variance $\sigma_t^2$; see Table \ref{tab:data summaries} for a definition of $T_\text{in}$ and $T_\text{out}$ used in our datasets. See \cite{Martens:2002} and \cite{Fleming:2003} for a similar scaling estimator of the daily volatility.

\subsubsection{In-sample analysis}
Table \ref{tab:Real_params} and \ref{tab:Real_params2} summarize the estimation results of fitting the benchmark GARCH-type models and their RECH counterparts to the SP500, N225, RUT and DAX datasets. The posterior mean estimates and posterior standard deviation estimates are obtained using the SMC with likelihood annealing sampler. We draw the following conclusions from the estimation results.
\begin{table}[H]
	\centering
	\footnotesize
	\begin{tabular}{rccccccccc}
		\hline\hline
		\rule{0pt}{3ex}
		& $\omega$& $\alpha$ & $\beta$ &$\gamma$ &$\beta_0$&  $\beta_1$&$v_1$&$v_2$&Mar.llh\\
		\hline
		\rule{0pt}{3ex}
		\textbf{SP500} & & & & &&&&&\\
		\rule{0pt}{3ex}
		\multirow{2}{*}{} GARCH   & \bol{0.016}  & \bol{0.093}  & \bol{0.894}    &   &&&& &$-2778.3$\\
		&(0.004)&(0.011)&(0.012)&&&&&&(0.113)\\
		\rule{0pt}{3ex}
		\multirow{2}{*}{} SRN-GARCH &   &\bol{0.057} & \bol{0.562} &&\bol{0.101} &\bol{0.413}&\bol{-0.380}&\bol{0.652}& $-2742.3$*\\
		& &(0.011)  & (0.082) &&(0.026) &(0.063) &(0.076)&(0.167)&(0.284)\\
		
		 & & & & & & &&&\\
		\rule{0pt}{3ex}
		\multirow{2}{*}{} GJR &\bol{0.024} &\bol{0.040} &\bol{0.891} &\bol{0.065}& &&&&$-2764.4$\\
		&(0.004) &(0.011) & (0.009)&(0.011) &&&& &(0.121)\\
		
		\rule{0pt}{3ex}
		\multirow{2}{*}{}SRN-GJR  & &\bol{0.011} &\bol{0.685} &\bol{0.109}&\bol{0.078}&\bol{0.349} &\bol{-0.222}&\bol{0.581}&$-2736.6$*\\
	    & &(0.008) &(0.109) &(0.026) &(0.036)&(0.102)&(0.093)&(0.196)&(0.323)\\
	
		 & & & & & &&&&\\
		\rule{0pt}{3ex}
		\multirow{2}{*}{}EGARCH &\bol{0.004}&\bol{0.136} &\bol{0.978} &$\bol{-0.126}$ &&&&&$-2754.8$\\
		&(0.002)  &(0.016) &(0.003) &(0.013) & &&&&(0.133)\\

		\rule{0pt}{3ex}
		\multirow{3}{*}{}SRN-EGARCH  &\bol{-0.196} &\bol{0.106} &\bol{0.972} &-\bol{0.236} &\bol{0.066} &\bol{0.343}&\bol{0.104}&\bol{0.538}&$-2744.0$*\\
		&(0.092)  &(0.030) &(0.016) &(0.044)&(0.027) &(0.099) &(0.103)&(0.198)&(0.225)\\
		\hline
		\rule{0pt}{3ex}
		\textbf{N225} & & & & &&&&&\\
		\rule{0pt}{3ex}
		\multirow{2}{*}{} GARCH   & \bol{0.033}  & \bol{0.138}  & \bol{0.839}    &   &&&& &$-2800.1$\\
		&(0.008)&(0.016)&(0.018)&&&&&&(0.119)\\
		\rule{0pt}{3ex}
		\multirow{2}{*}{} SRN-GARCH &   &\bol{0.076} & \bol{0.744} &&\bol{0.016} &\bol{0.388}&\bol{-0.574}&\bol{0.400}& $-2766.8$*\\
		& &(0.016)  & (0.051) &&(0.013) &(0.076) &(0.156)&(0.147)&(0.279)\\
		
		& & & & & & &&&\\
		\rule{0pt}{3ex}
		\multirow{2}{*}{} GJR &\bol{0.048} &\bol{0.060} &\bol{0.835} &\bol{0.100}& &&&&$-2787.5$\\
		&(0.009) &(0.015) & (0.020)&(0.018) &&&& &(0.121)\\
		
		\rule{0pt}{3ex}
		\multirow{2}{*}{}SRN-GJR  & &\bol{0.066} &\bol{0.734} &\bol{0.029}&\bol{0.027}&\bol{0.386} &\bol{-0.531}&\bol{0.429}&$-2769.2$*\\
		& &(0.021) &(0.061) &(0.036) &(0.015)&(0.081)&(0.146)&(0.183)&(0.313)\\
		
		& & & & & &&&&\\
		\rule{0pt}{3ex}
		\multirow{2}{*}{}EGARCH &\bol{-0.003}&\bol{0.215} &\bol{0.967} &$\bol{-0.134}$ &&&&&$-2772.0$\\
		&(0.004)  &(0.022) &(0.006) &(0.015) & &&&&(0.133)\\
		
		\rule{0pt}{3ex}
		\multirow{3}{*}{}SRN-EGARCH  &\bol{-0.088} &\bol{0.255} &\bol{0.990} &\bol{-0.164} &\bol{0.033} &\bol{0.344}&\bol{-0.355}&\bol{0.434}&$-2772.5$\\
		&(0.048)  &(0.046) &(0.009) &(0.019)&(0.020) &(0.117) &(0.148)&(0.325)&(0.225)\\
		\hline\hline
	\end{tabular}
	\caption{SP500 and N225 data: Posterior means (in bold) of the parameters with the posterior standard deviations (in brackets). The last column shows the estimated log marginal likelihood with the Monte Carlo standard errors in brackets, averaged over 10 different runs of the SMC using the likelihood annealing algorithm. The asterisks indicate the cases when the Bayes factors strongly support the RECH models over their corresponding GARCH-type models.}
	\label{tab:Real_params}
\end{table}

\begin{table}[H]
	\centering
	\footnotesize
	\begin{tabular}{rccccccccc}
		\hline\hline
		\rule{0pt}{3ex}
		& $\omega$& $\alpha$ & $\beta$ &$\gamma$ &$\beta_0$&  $\beta_1$&$v_1$&$v_2$&Mar.llh\\
		\hline
		\rule{0pt}{3ex}
		\textbf{RUT} & & & & &&&&&\\
		\rule{0pt}{3ex}
		\multirow{2}{*}{} GARCH   & \bol{0.036}  & \bol{0.112}  & \bol{0.863}    &   &&&& &$-3037.6$\\
		&(0.009)&(0.015)&(0.019)&&&&&&(0.117)\\
		\rule{0pt}{3ex}
		\multirow{2}{*}{} SRN-GARCH &   &\bol{0.071} & \bol{0.690} &&\bol{0.075} &\bol{0.362}&\bol{-0.422}&\bol{0.538}& $-3014.3$*\\
		& &(0.017)  & (0.078) &&(0.040) &(0.082) &(0.118)&(0.203)&(0.284)\\
		
		& & & & & & &&&\\
		\rule{0pt}{3ex}
		\multirow{2}{*}{} GJR &\bol{0.049} &\bol{0.048} &\bol{0.864} &\bol{0.084}& &&&&$-3025.4$\\
		&(0.009) &(0.013) & (0.017)&(0.015) &&&& &(0.121)\\
		
		\rule{0pt}{3ex}
		\multirow{2}{*}{}SRN-GJR  & &\bol{0.019} &\bol{0.780} &\bol{0.109}&\bol{0.078}&\bol{0.311} &\bol{-0.255}&\bol{0.428}&$-3010.6$*\\
		& &(0.008) &(0.109) &(0.026) &(0.036)&(0.102)&(0.093)&(0.196)&(0.323)\\
		
		& & & & & &&&&\\
		\rule{0pt}{3ex}
		\multirow{2}{*}{}EGARCH &\bol{0.013}&\bol{0.162} &\bol{0.969} &$\bol{-0.111}$ &&&&&$-3023.1$\\
		&(0.004)  &(0.021) &(0.006) &(0.015) & &&&&(0.130)\\
		
		\rule{0pt}{3ex}
		\multirow{3}{*}{}SRN-EGARCH  &\bol{-0.090} &\bol{0.156} &\bol{0.990} &\bol{-0.161} &\bol{0.050} &\bol{0.307}&\bol{-0.282}&\bol{0.238}&$-3015.5$*\\
		&(0.049)  &(0.027) &(0.010) &(0.019)&(0.018) &(0.115) &(0.205)&(0.254)&(0.225)\\
		\hline
		\rule{0pt}{3ex}
		\textbf{DAX} & & & & &&&&&\\
		\rule{0pt}{3ex}
		\multirow{2}{*}{} GARCH   & \bol{0.019}  & \bol{0.096}  & \bol{0.890}    &   &&&& &$-2902.0$\\
		&(0.005)&(0.013)&(0.014)&&&&&&(0.119)\\
		\rule{0pt}{3ex}
		\multirow{2}{*}{} SRN-GARCH &   &\bol{0.058} & \bol{0.681} &&\bol{0.068} &\bol{0.418}&\bol{-0.430}&\bol{0.524}& $-2867.2$*\\
		& &(0.019)  & (0.126) &&(0.037) &(0.059) &(0.131)&(0.281)&(0.279)\\
		
		& & & & & & &&&\\
		\rule{0pt}{3ex}
		\multirow{2}{*}{} GJR &\bol{0.031} &\bol{0.046} &\bol{0.884} &\bol{0.066}& &&&&$-2889.1$\\
		&(0.006) &(0.010) & (0.013)&(0.013) &&&& &(0.120)\\
		
		\rule{0pt}{3ex}
		\multirow{2}{*}{}SRN-GJR  & &\bol{0.035} &\bol{0.711} &\bol{0.072}&\bol{0.067}&\bol{0.369} &\bol{-0.301}&\bol{0.575}&$-2866.4$*\\
		& &(0.016) &(0.088) &(0.032) &(0.038)&(0.095)&(0.117)&(0.213)&(0.313)\\
		
		& & & & & &&&&\\
		\rule{0pt}{3ex}
		\multirow{2}{*}{}EGARCH &\bol{0.006}&\bol{0.158} &\bol{0.975} &$\bol{-0.116}$ &&&&&$-2872.5$\\
		&(0.002)  &(0.018) &(0.004) &(0.014) & &&&&(0.128)\\
		
		\rule{0pt}{3ex}
		\multirow{3}{*}{}SRN-EGARCH  &\bol{-0.063} &\bol{0.169} &\bol{0.989} &\bol{-0.127} &\bol{0.034} &\bol{0.265}&\bol{-0.185}&\bol{0.193}&$-2869.0$*\\
		&(0.063)  &(0.026) &(0.011) &(0.020)&(0.024) &(0.132) &(0.160)&(0.394)&(0.232)\\
		\hline\hline
	\end{tabular}
	\caption{RUT and DAX data: Posterior means (in bold) of the parameters with the posterior standard deviations (in brackets). The last column shows the estimated log marginal likelihood with the Monte Carlo standard errors in brackets, averaged over 10 different runs of the SMC using the likelihood annealing algorithm. The asterisks indicate the cases when the Bayes factors strongly support the RECH models over their corresponding GARCH-type models.}
	\label{tab:Real_params2}
\end{table}

First, the marginal likelihood estimates show that the RECH models fit the index datasets better than the GARCH-type models, except for the SRN-EGARCH model for the N225 data.
For example, for the SP500 data, the Bayes factors of the SRN-GARCH, SRN-GJR and SRN-EGARCH models compared to the GARCH, GJR and EGARCH models are roughly $e^{36}$, $e^{28}$ and $e^{10}$, respectively, which, according to Jeffrey’s scale for interpreting the Bayes factor \citep{Jeffreys:1961}, decisively support the RECH models. Among the benchmark GARCH-type models, the EGARCH model constantly has the highest marginal likelihood.

Second, the estimated posterior means of the parameter $\beta_1$ of the RECH models are more than two standard deviations from zero in all cases, providing evidence of the volatility effects rather than linearity, e.g. probably non-linearity and long-memory effects, in the volatility dynamics and also suggesting that the recurrent component of the RECH models is able to effectively detect these effects. Additionally, the coefficients $v_2$ of the RECH models are statistically significant, indicating that the RECH models are able to detect the serial dependence rather than linearity that the previous conditional variance $\sigma_{t-1}^2$ has on $\sigma_t^2$.

Third, the existence of the leverage effects in the volatility is clear across all four stock markets. The leverage parameters $\gamma$ of the GJR and EGARCH models are statistically significant, implying that these models can detect the asymmetric volatility. All the leverage effect-related parameters $\gamma$ and $v_1$ in the RECH models are statistically significant, except the parameter $v_1$ of the SRN-EGARCH model. In particular, the {\it linear} leverage coefficient $\gamma$ of the SRN-GRJ and SRN-EGARCH models are significant, similarly to those of the GRJ and EGARCH models.
Interestingly, the {\it non-linear} leverage coefficient $v_1$ of the RECH models is significant in almost all cases, suggesting that the spillover effect of asymmetric volatility can be non-linear.
In particular, the leverage coefficient $v_1$ of the SRN-GARCH model is also statistically significant across all markets; i.e., unlike the conventional GARCH model, SRN-GARCH can detect the leverage effects in volatility.

Finally, as pointed out by a reviewer, in the EGARCH case the $\alpha$ and $\beta$ appear to be less affected by adding the recurrent component.
This is because the EGARCH structure is very different from GARCH and GJR (see Table \ref{tab:RECH_GARCH_models});
the EGARCH part in SRN-EGARCH is already non-linear in $\sigma_{t-1}$ and hence it can accommodate non-linearity.
This doesn't mean that SRN-EGARCH cannot improve EGARCH - $\beta_1$ is still far away from zero and the improvement is evident in better marginal likelihood and better out-of-sample prediction as confirmed in Section \ref{sec:out of sample analysis}.

\begin{figure}[h]
	\centering
	\includegraphics[width=1\columnwidth]{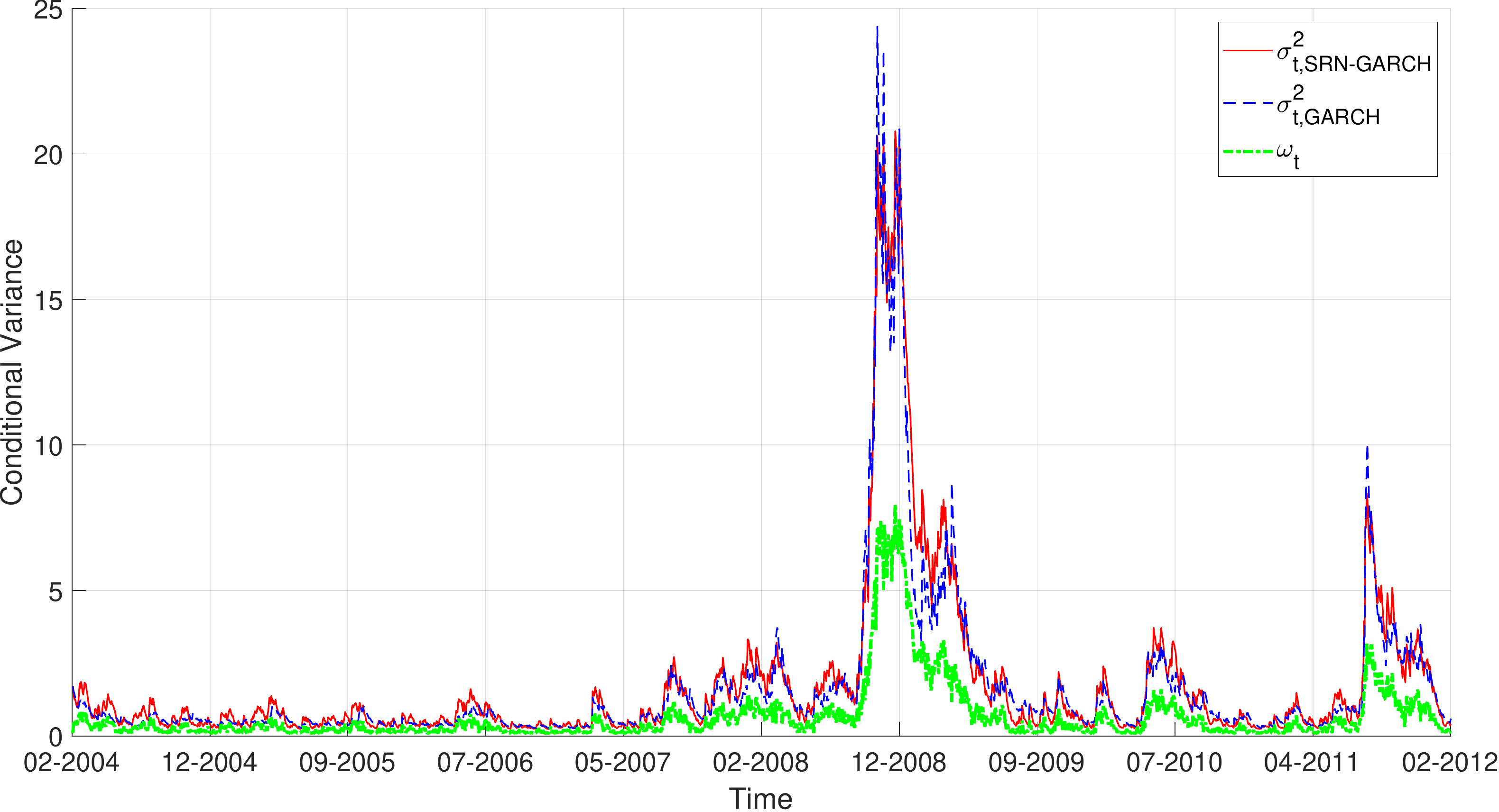}
	\caption{SP500: The in-sample conditional variance of the GARCH (dashed line) and SRN-GARCH (solid line) at all time points. The bottom line shows the values of the recurrent component $\omega_t$ of the SRN-GARCH specification. (The figure is better viewed in colour).}
	\label{f:SP500_volatility_density}
\end{figure}

Figure \ref{f:SP500_volatility_density} shows the volatility estimated by the GARCH and SRN-GARCH models for the SP500 index data, together with the values of the recurrent component $\omega_t$ at all time points.
Figure \ref{f:SP500_volatility_density_SRN_GJR} and \ref{f:SP500_volatility_density_SRN_EGARCH} in the Appendix are similar plots for the SRN-GJR and SRN-EGARCH models.
Clearly, the recurrent component $\omega_t$ is responsive to the changes in the volatility dynamics: it is small during the low volatility periods and large in the high volatility periods.
This distinct behavior of financial volatility is well-captured by the recurrent neural network structure of the recurrent component $\omega_t$.

Figure \ref{f:SP500_residual_plot} plots the standardized residuals $\wh\eps_t$ from the GARCH and SRN-GARCH models together with their QQ-plots.
\begin{figure}[h!]
	\centering
	\includegraphics[width=1\columnwidth]{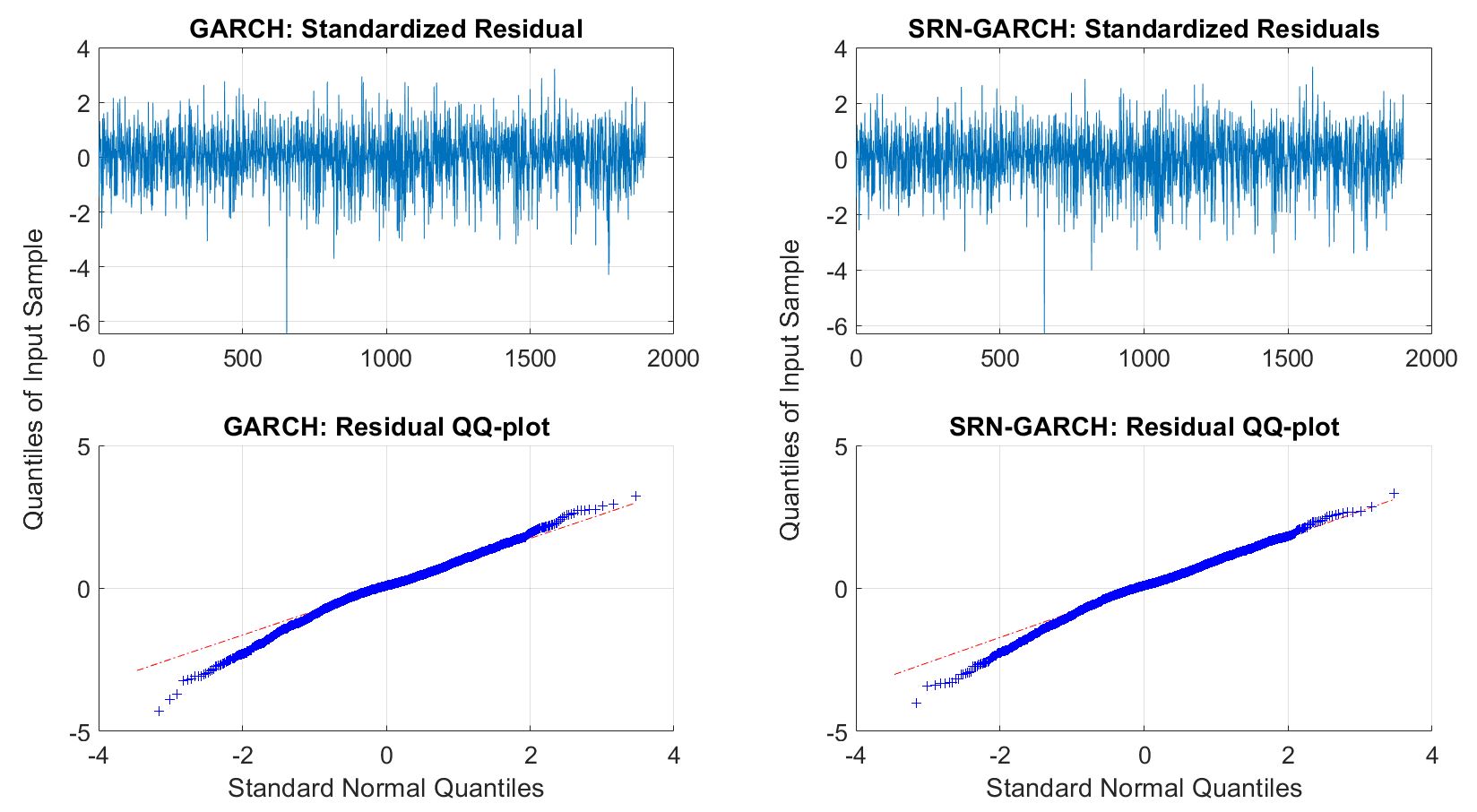}
	\caption{SP500: Estimated residuals $\wh\eps_t$ of the SRN-GARCH and GARCH models and their Q-Q plots.}
	\label{f:SP500_residual_plot}
\end{figure}

We observe similar results for the SRN-GRJ and GRJ models, SRN-EGARCH and EGARCH models. Generally the RECH residuals appear to lie closer to the expected straight line than those of the counterpart GARCH-type models.

\begin{table}[h!]
	\begin{center}
		\footnotesize
		\begin{tabular}{rccccccccc}
			\hline\hline
			\rule{0pt}{3ex}
			&\multicolumn{4}{l}{Fitted Conditional Variance } &&\multicolumn{4}{l}{Residual $\wh\eps_t$}\\
			\cline{2-5}\cline{7-10}
			\rule{0pt}{3ex}
			&Mean &Std &Skew& Kurtosis& & Std      & Skew           &Kurtosis      &LB-$\wh\eps_t$  \\
			\hline
			\rule{0pt}{3ex}
			\textbf{S\&P500}      &&&&&&&&&\\
			\rule{0pt}{3ex}
			GARCH      &1.650&2.845&4.543&26.522&& 1.002  &$-0.505$   &4.535 &0.054 \\
			SRN-GARCH  &1.755&2.910&4.045&21.251&& 1.002 &$-0.509$   &4.219 &0.119  \\
			\rule{0pt}{3ex}
			GJR      &1.423&2.342&4.585&27.169&&  1.001 &$-0.485$   &4.202 &0.058 \\
			SRN-GJR    &1.600&2.877&4.445&25.585&& 1.001&$-0.515$   &4.200 &0.110\\
			\rule{0pt}{3ex}
			EGARCH        &1.534&2.318&4.291&25.389&& 0.999  &$-0.587$   &4.616 &0.051 \\
			SRN-EGARCH    &1.603&2.887&4.661&29.080&& 0.999&$-0.524$   &4.125 &0.107 \\
			\hline
			\rule{0pt}{3ex}
			\textbf{N225}      &&&&&&&&&\\
			\rule{0pt}{3ex}
			GARCH      &1.493&3.133&7.792&72.005&& 0.999  &$-0.528$   &5.037 &0.138 \\
			SRN-GARCH  &1.390&2.521&6.950&58.482&& 0.999 &$-0.407$   &4.275 &0.108 \\
			\rule{0pt}{3ex}
			GJR      &1.352&2.692&7.780&72.782&&  1.003 &$-0.455$   &4.509 &0.104  \\
			SRN-GJR    &1.434&2.634&7.054&60.510&& 1.003&$-0.404$   &4.235 &0.085  \\
			\rule{0pt}{3ex}
			EGARCH        &1.370&2.282&7.252&68.428&& 1.001  &$-0.404$   &4.197 &0.073 \\
			SRN-EGARCH    &1.242&2.430&7.605&73.471&& 1.001&$-0.378$   &4.158 &0.083 \\
			\hline
			\rule{0pt}{3ex}
			\textbf{RUT}      &&&&&&&&&\\
			\rule{0pt}{3ex}
			GARCH      &1.803&2.524&4.177&23.108&& 1.000  &$-0.318$   &3.717 &0.046 \\
			SRN-GARCH  &1.826&2.407&3.746&18.897&& 1.000 &$-0.379$   &3.728 &0.084 \\
			\rule{0pt}{3ex}
			GJR        &1.619&2.122&4.343&25.208&&  1.002 &$-0.322$   &3.663 &0.034  \\
			SRN-GJR    &1.831&2.635&4.000&21.316&& 1.002&$-0.364$   &3.725 &0.085  \\
			\rule{0pt}{3ex}
			EGARCH     &1.681&1.973&3.775&20.933&& 0.999  &$-0.397$   &3.843 &0.061 \\
			SRN-EGARCH &1.634&2.246&4.245&24.624&& 0.999&$-0.373$   &3.690 &0.057 \\
			\hline
			\rule{0pt}{3ex}
			\textbf{DAX}      &&&&&&&&&\\
			\rule{0pt}{3ex}
			GARCH      &1.598&2.203&4.332&26.122&& 1.002  &$-0.462$   &4.870 &0.957 \\
			SRN-GARCH  &1.656&2.378&3.484&16.904&& 1.002 &$-0.468$   &4.323 &0.893 \\
			\rule{0pt}{3ex}
			GJR        &1.400&1.739&3.815&19.772&&  1.001 &$-0.399$   &4.478&0.957  \\
			SRN-GJR    &1.510&2.013&3.517&17.005&& 1.001&$-0.420$   &4.228&0.911  \\
			\rule{0pt}{3ex}
			EGARCH      &1.504&1.766&3.438&17.766&& 1.000  &$-0.414$   &4.116 &0.905 \\
			SRN-EGARCH  &1.262&1.730&4.132&23.497&& 1.000&$-0.385$   &4.066 &0.916 \\
			\hline\hline
		\end{tabular}
	\end{center}
	\caption{SP500: Model diagnostics of the fitted conditional variance and residual $\wh\eps_t$. The LB p-values denote the p-value from the Ljung-Box test with 10 lags.}
	\label{tab:Real Model diagnostics}
\end{table}

Table \ref{tab:Real Model diagnostics} provides the skewness and kurtosis statistics together with the $p$-values of the Ljung-Box (LB) autocorrelation test of the residuals and squared residuals estimated by the RECH and the benchmark GARCH-type models.
The $p$-value of the LB test, together with the sample ACF plots, of the standardized and squared standardized residuals suggest that there is no evidence of autocorrelation.
The residuals produced by all models in Table \ref{tab:Real Model diagnostics} exhibit some negative skewness and have kurtosis values higher than 3 (the kurtosis of the standard normal distribution). In general, the residuals of the RECH models seem closer to normality than those of the corresponding GARCH-type models.
Similarly to the GARCH-type models, it is straightforward to use Student's $t$ distribution for the innovation in the RECH models to improve the residual diagnostics; however, this extension is not considered here.

\subsubsection{Out-of-sample analysis}\label{sec:out of sample analysis}
\begin{figure}[h!]
	\begin{center}
		\includegraphics[width=1\columnwidth]{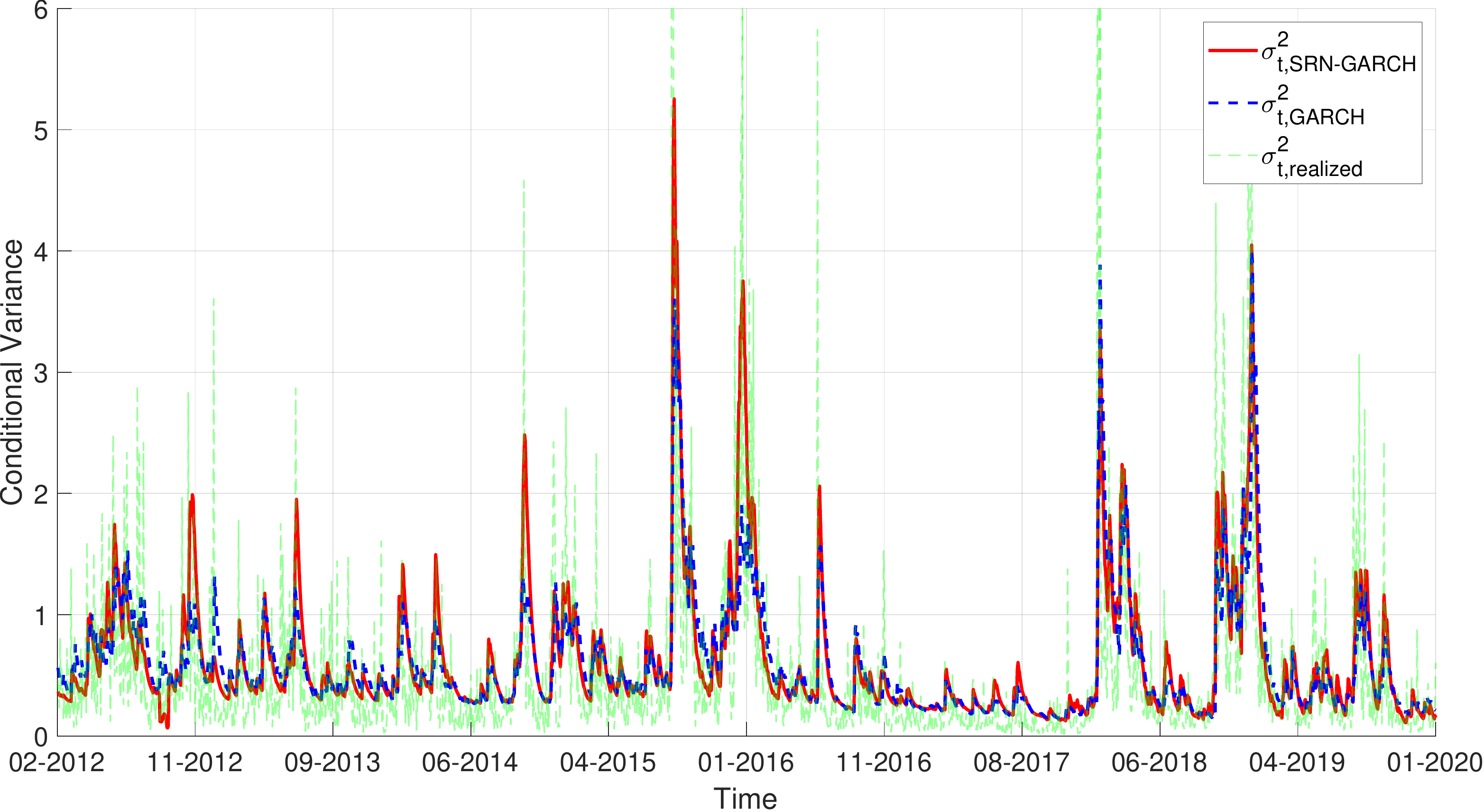}
		\caption{SP500 data: Forecast conditional variance by the GARCH (dashed) and SRN-GARCH (solid) models, together with the realized variance (dotted). (The figure is better viewed in colour).}
		\label{f:SP500_forecast_variance_GARCH}
	\end{center}
\end{figure}

Figure \ref{f:SP500_forecast_variance_GARCH} plots the one-step-ahead forecast conditional variance of the GARCH and SRN-GARCH models, together with the out-of-sample realized variance of the S\&P500 data, obtained by data annealing SMC.
Figure \ref{f:SP500_forecast_variance_GJR} and \ref{f:SP500_forecast_variance_EGARCH} in the Appendix are similar plots for the case of the SRN-GJR and SRN-EGARCH models, respectively. To save space, we do not report the plots for the N225, RUT and DAX datasets as similar behaviors of the forecast variance are observed for these datasets.
We note from Figures \ref{f:SP500_forecast_variance_GARCH}, \ref{f:SP500_forecast_variance_GJR} and \ref{f:SP500_forecast_variance_EGARCH} that in general, the RECH models and their GARCH-type counterparts produce forecast variances that adequately track the movement of the realized variance.
The forecast variance of the RECH and benchmark models are similar in the low volatility regions while the RECH models have higher variance forecasts during high volatility periods. The variance forecasts of the RECH models seem to track the realized variance better than those of the benchmark GARCH-type models.

\begin{table}[h!]
	\begin{center}
		\footnotesize
		\begin{tabular}{rccccccccc}
			\hline\hline
			\rule{0pt}{3ex}
			&\multicolumn{4}{l}{Forecast Conditional Variance } &&\multicolumn{4}{l}{Forecast Residual $\wh\eps_t$}\\
			\cline{2-5}\cline{7-10}
			\rule{0pt}{3ex}
			&Mean &Std &Skew& Kurtosis& & Std      & Skew           &Kurtosis      &LB-$\wh\eps_t$  \\
			\hline
			\rule{0pt}{3ex}
			\textbf{SP500}      &&&&&&&&&\\
			\rule{0pt}{3ex}
			GARCH      &0.597&0.476&2.805&14.391&& 0.932  &$-0.818$  &6.438 &0.471 \\
			SRN-GARCH  &0.620&0.546&2.309&9.009&& 0.921 &$-0.697$   &5.243 &0.592  \\
			\rule{0pt}{3ex}
			GJR        &0.566&0.417&2.856&14.351&&  0.924 &$-0.738$   &5.775 &0.484 \\
			SRN-GJR    &0.626&0.565&2.443&9.868&&  0.932&$-0.730$   &5.477 &0.537\\
			\rule{0pt}{3ex}
			EGARCH        &0.632&0.566&2.404&10.518&& 0.914&$-0.815$   &5.939 &0.438 \\
			SRN-EGARCH    &0.634&0.615&2.945&13.554&& 0.910&$-0.749$   &5.592 &0.598 \\
			\hline
			\rule{0pt}{3ex}
			\textbf{N225}      &&&&&&&&&\\
			\rule{0pt}{3ex}
			GARCH      &0.930&1.061&4.553&30.093&& 0.977  &$-1.327$   &13.908 &0.184 \\
			SRN-GARCH  &0.884&0.921&3.752&21.485&& 1.004 &$-1.640$    &14.399 &0.089 \\
			\rule{0pt}{3ex}
			GJR        &0.903&1.046&5.145&38.668&&  0.986 &$-1.495$   &15.828&0.128  \\
			SRN-GJR    &0.889&0.944&3.761&21.200&& 1.006&$-1.661$     &18.781&0.094  \\
			\rule{0pt}{3ex}
			EGARCH       &0.917&1.245&6.677&69.000&& 1.007&$-1.599$   &17.672&0.146 \\
			SRN-EGARCH   &0.897&1.231&6.858&74.418&& 1.017&$-1.668$   &18.652&0.116 \\
			\hline
			\rule{0pt}{3ex}
			\textbf{RUT}      &&&&&&&&&\\
			\rule{0pt}{3ex}
			GARCH      &0.928&0.459&1.961&9.011&& 0.959  &$-0.471$   &4.101 &0.721 \\
			SRN-GARCH  &0.953&0.531&2.139&9.850&& 0.947 &$-0.547$   &4.378 &0.562 \\
			\rule{0pt}{3ex}
			GJR        &0.919&0.448&2.034&8.473&& 0.953&$-0.462$   &4.084 &0.689  \\
			SRN-GJR    &0.958&0.548&2.258&10.044&& 0.940&$-0.465$   &4.080 &0.572  \\
			\rule{0pt}{3ex}
			EGARCH     &0.990&0.545&1.828&8.450&& 0.932&$-0.492$   &4.273 &0.631 \\
			SRN-EGARCH &0.964&0.594&3.273&12.528&& 0.937&$-0.462$   &4.106 &0.618 \\
			\hline
			\rule{0pt}{3ex}
			\textbf{DAX}      &&&&&&&&&\\
			\rule{0pt}{3ex}
			GARCH      &0.867&0.474&1.238&4.394&& 0.973  &$-0.245$   &4.923 &0.152 \\
			SRN-GARCH  &0.904&0.585&1.566&5.807&& 0.962 &$-0.235$   &4.992 &0.125 \\
			\rule{0pt}{3ex}
			GJR        &0.829&0.440&1.661&6.735&& 0.977&$-0.235$   &4.733&0.097  \\
			SRN-GJR    &0.907&0.594&1.646&6.226&& 0.964&$-0.250$   &4.969&0.107  \\
			\rule{0pt}{3ex}
			EGARCH      &0.924&0.617&1.584&6.046&& 0.967&$-0.296$   &5.205 &0.124 \\
			SRN-EGARCH  &0.902&0.634&1.987&7.982&& 0.972&$-0.228$   &4.956 &0.105 \\
			\hline\hline
		\end{tabular}
	\end{center}
	\caption{Application: Summary statistics on the one-step-ahead out-of-sample forecast conditional variances $\widehat\sigma^2_t$ and residual $\wh\eps_t$. The LB p-values denote the p-value from the Ljung-Box test with 10 lags.}
	\label{tab:Real Model diagnostics_Out_sample}
\end{table}

The in-sample analysis suggests that the RECH models fit the in-sample data of the four index datasets better than the counterpart GARCH-type models. However, it is possible that the superior in-sample performance is the result of overfitting \citep{Pagan:1990,Donaldson:1997}. Table \ref{tab:Real Model diagnostics_Out_sample} provides summary statistics on the one-step-ahead forecasts of conditional variance and standardized residuals.
The most important conclusion from Table \ref{tab:Real Model diagnostics_Out_sample} is that the RECH models do not overfit the data, as the forecast conditional variance of the RECH models are not excessively variable and the forecast residuals of the RECH models are very close to those of the GARCH-type benchmark models. The RECH models occasionally produce one-step-ahead forecast residuals with lower kurtosis than the counterpart GARCH-type models.

Table \ref{tab:forecast_score_compare_pps} shows the forecast performance of the models using the four predictive scores PPS, \#Vio, QS and \%Hit.
As these four predictive scores complement each other, for each pair of the RECH and GARCH-type models, we compare their forecast performance by counting the number of times one model has a better predictive score than the other and report this count in the last column of Table \ref{tab:forecast_score_compare_pps}. The model with the higher count is preferred. Table \ref{tab:forecast_score_compare_pps} shows that the RECH models consistently outperform their counterpart GARCH-type models for the S\&P500, RUT and DAX data. For the N225 index, the predictive improvement of the RECH models over the benchmark counterparts is less clear, especially for the SRN-GARCH and GARCH models.

\begin{table}[h!]
	\centering
	\footnotesize
	\begin{tabular}{rcccccc}
		\hline\hline
		\rule{0pt}{3ex}
	    && PPS & \# Violation &QS &Hit Per. &Count\\
		\hline
		\rule{0pt}{3ex}
		\bol{SP500}&&   & & & &\\
		\rule{0pt}{3ex}
		GARCH     &&0.993&32&0.026&0.018 &0\\
		SRN-GARCH* &&0.955&25&0.024&0.017 &4\\
		\rule{0pt}{3ex}
		GJR       &&0.981&27&0.025&0.018 &0\\
		SRN-GJR*  & &0.959&25&0.024&0.017 &4\\
		\rule{0pt}{3ex}
		EGARCH    &&0.963&29&0.025&0.018 &0\\
		SRN-EGARCH*&&0.963&23&0.025&0.015 &2\\
		\hline
		\rule{0pt}{3ex}
		\bol{N225}&&   & & & &\\
		\rule{0pt}{3ex}
		GARCH*     &&1.214&32&0.036&0.016 &2\\
		SRN-GARCH &&1.216&33&0.035&0.016 &1\\
		\rule{0pt}{3ex}
		GJR       &&1.217&31&0.036&0.017 &1\\
		SRN-GJR*  &&1.216&32&0.035&0.017 &2\\
		\rule{0pt}{3ex}
		EGARCH    &&1.212&32&0.035&0.017 & 0\\
		SRN-EGARCH*&&1.212&30&0.035&0.017 & 1\\
		\hline
		\rule{0pt}{3ex}
		\bol{RUT}&&   & & & &\\
		\rule{0pt}{3ex}
		GARCH     &&1.298&34&0.032&0.018 &0\\
		SRN-GARCH* &&1.286&26&0.030&0.016&4\\
		\rule{0pt}{3ex}
		GJR &&1.290&31&0.031&0.017&0\\
		SRN-GJR*  &&1.283&24&0.030&0.015&4\\
		\rule{0pt}{3ex}
		EGARCH    &&1.285&23&0.030&0.014 &0\\
		SRN-EGARCH*&&1.281&23&0.030&0.014 &1\\
		\hline
		\rule{0pt}{3ex}
		\bol{DAX}&&   & & & &\\
		\rule{0pt}{3ex}
		GARCH     &&1.257&47&0.031&0.020 & 0\\
		SRN-GARCH* &&1.247&38&0.029&0.020 & 3\\
		\rule{0pt}{3ex}
		GJR       &&1.249&50&0.030&0.022&0\\
		SRN-GJR*   &&1.246&41&0.029&0.022&3\\
		\rule{0pt}{3ex}
		EGARCH    &&1.246&39&0.028&0.018 &1\\
		SRN-EGARCH*&&1.243&37&0.029&0.020 &3\\
		\hline\hline
	\end{tabular}
	\caption{Applications: Forecast performance of the RECH and benchmark GARCH-type models. For the QS and \%Hit measures, the results are calculated at the $1\%$-quantile. For each pair of the RECH and GARCH-type models, the asterisks  indicate the model with the higher count. }
	\label{tab:forecast_score_compare_pps}
\end{table}

Tables \ref{tab:SP500_forecast_score_compare} to \ref{tab:DAX_forecast_score_compare} summarize the forecast performance measured by the predictive scores defined in Table \ref{tab:forecast_score}. Each table contains six panels, corresponding to the six realized measures mentioned earlier. For each pair of the RECH and GARCH-type models, their forecast performance are also compared in the same way as in Table \ref{tab:forecast_score_compare_pps}. Additionally, in each panel, bold numbers are used to indicate the lowest forecast errors. For each type of realized measure, the model with the highest number of lowest forecast errors is preferred. 
The table shows that the RECH models in general outperform their counterpart GARCH-type models.

\begin{table}[h!]
	\centering
	\footnotesize
	\begin{tabular}{crccccccc}
		\hline\hline
		\rule{0pt}{3ex}
		Estimator&& $\text{MSE}_1$ &$\text{MSE}_2$ & $\text{MAE}_1$&$\text{MAE}_2$&$\text{Qlike}$&$\text{R}^2\text{LOG}$&Count\\
		\hline
		\rule{0pt}{3ex}
		&GARCH     &0.088&0.780&0.220&0.351&0.115&0.715&0 \\
		&SRN-GARCH* &\bol{0.075}&0.703&\bol{0.203}&0.325&\bol{0.078}&\bol{0.631} &6\\
		\rule{0pt}{3ex}
		BV&GJR       &0.077&0.721&0.206&\bol{0.322}&0.097&0.678 &1\\
		&SRN-GJR*   &0.076&\bol{0.691}&0.205&0.328&0.080&0.642 &5\\
		\rule{0pt}{3ex}
		&EGARCH    &0.078&0.695&0.207&0.332&0.091&0.644 &1\\
		&SRN-EGARCH*&\bol{0.075}&0.694&0.204&0.326&0.081&0.646 &5\\
		\hline
		\rule{0pt}{3ex}
		&GARCH     &0.097&0.600&0.237&0.367&0.127&0.911 &0\\
		&SRN-GARCH* &\bol{0.084}&0.538&\bol{0.222}&0.346&\bol{0.093}&0.815 &6\\
		\rule{0pt}{3ex}
		$\text{RKV}_1$&GJR       &0.085&0.545&0.225&\bol{0.339}&0.111&0.873 &1\\
		&SRN-GJR*   &0.084&0.531&0.224&0.348&0.094&0.827 &5\\
		\rule{0pt}{3ex}
		&EGARCH    &0.086&\bol{0.528}&0.224&0.349&0.098&\bol{0.810} &2\\
		&SRN-EGARCH&0.085&0.538&0.224&0.349&0.097&0.834 &2\\
		\hline
		\rule{0pt}{3ex}
		&GARCH     &0.071&0.422&0.200&0.315&0.135&0.575&0 \\
		&SRN-GARCH* &\bol{0.061}&0.371&\bol{0.187}&0.296&\bol{0.109}&\bol{0.507} &6\\
		\rule{0pt}{3ex}
		$\text{RKV}_2$&GJR &0.062&0.374&0.188&\bol{0.288}&0.121&0.542 &2\\
		&SRN-GJR*   &\bol{0.061}&\bol{0.366}&0.189&0.299&0.110&0.517 &3\\
		\rule{0pt}{3ex}
		&EGARCH    &0.063&0.368&0.190&0.301&0.117&0.514 &2\\
		&SRN-EGARCH*&0.062&0.376&0.188&0.298&0.111&0.521 &4\\
		\hline
		\rule{0pt}{3ex}
		&GARCH     &0.070&0.401&0.199&0.314&0.137&0.569&0 \\
		&SRN-GARCH* &\bol{0.060}&0.351&\bol{0.187}&\bol{0.295}&\bol{0.110}&\bol{0.501} &6\\
		\rule{0pt}{3ex}
		$\text{RKV}_3$&GJR  & \bol{0.060}&0.353&\bol{0.187}&0.287&0.123&0.536 &2\\
		&SRN-GJR*   &\bol{0.060}&\bol{0.347}&0.188&0.297&0.112&0.511&3\\
		\rule{0pt}{3ex}
		&EGARCH    &0.063&0.348&0.189&0.300&0.120&0.509 &2\\
		&SRN-EGARCH*&0.061&0.358&0.188&0.298&0.114&0.515 &4\\
		\hline
		\rule{0pt}{3ex}
		&GARCH     &0.102&0.637&0.246&0.386&0.083&0.946&0 \\
		&SRN-GARCH* &\bol{0.087}&0.555&\bol{0.227}&\bol{0.357}&\bol{0.039}&\bol{0.849} &6\\
		\rule{0pt}{3ex}
		MedRV&GJR       &0.091&0.586&0.234&0.360&0.067&0.909 &1\\
		&SRN-GJR*   &0.088&0.544&0.230&0.361&0.043&0.863 &5\\
		\rule{0pt}{3ex}
		&EGARCH    &0.091&0.550&0.232&0.366&0.051&0.863 &1\\
		&SRN-EGARCH*&\bol{0.087}&\bol{0.540}&0.229&0.359&0.045&0.867&5\\
		\hline
		\rule{0pt}{3ex}
		&GARCH     &0.096&1.124&0.226&0.361&0.105&0.782 &0\\
		&SRN-GARCH* &\bol{0.083}&1.054&\bol{0.212}&0.340&\bol{0.073}&0.701 &6\\
		\rule{0pt}{3ex}
		RV&GJR       &0.084&1.061&0.214&\bol{0.333}&0.091&0.751 &1\\
		&SRN-GJR*   &0.084&1.039&0.214&0.343&0.075&0.713&3\\
		\rule{0pt}{3ex}
		&EGARCH*    &0.085&\bol{1.037}&0.214&0.345&0.075&\bol{0.695} &3\\
		&SRN-EGARCH&0.084&1.045&0.214&0.344&0.078&0.720 &2\\
		\hline\hline
	\end{tabular}
	\caption{SP500 data: Forecast performance of the RECH and benchmark GARCH-type models using different realized measures. For each pair of the RECH and GARCH-type models, the asterisks indicate the model with the higher count. In each panel, the bold numbers indicate the best predictive scores.}
	\label{tab:SP500_forecast_score_compare}
\end{table}

\begin{table}[h!]
	\centering
	\footnotesize
	\begin{tabular}{crccccccc}
		\hline\hline
		\rule{0pt}{3ex}
		Estimator&& $\text{MSE}_1$ &$\text{MSE}_2$ & $\text{MAE}_1$&$\text{MAE}_2$&$\text{QLIKE}$&$\text{R}^2\text{LOG}$&Count\\
		\hline
		\rule{0pt}{3ex}
		&GARCH     &0.138&1.924&0.242&0.526&0.562&0.537&0 \\
		&SRN-GARCH* &\bol{0.124}&1.795&0.227&\bol{0.487}&0.557&0.491 &6\\
		\rule{0pt}{3ex}
		BV&GJR       &0.133&1.896&0.234&0.505&0.561&0.526 &0\\
		&SRN-GJR*   &0.125&\bol{1.792}&0.227&0.491&0.553&0.485 &6\\
		\rule{0pt}{3ex}
		&EGARCH    &0.130&1.870&0.226&0.499&0.553&0.478 &0\\
		&SRN-EGARCH*&\bol{0.124}&1.794&\bol{0.224}&\bol{0.487}&\bol{0.552}&\bol{0.472} &6\\
		\hline
		\rule{0pt}{3ex}
		&GARCH     &0.216&3.792&0.311&0.661&0.573&1.028 &0\\
		&SRN-GARCH* &\bol{0.200}&\bol{3.635}&0.297&0.624&0.568&0.956 &6\\
		\rule{0pt}{3ex}
		$\text{RKV}_1$&GJR  &0.210&3.758&0.304&0.641&0.574&1.012 &0\\
		&SRN-GJR*   &0.201&3.648&0.298&0.628&0.565&0.952 &6\\
		\rule{0pt}{3ex}
		&EGARCH    &0.206&3.761&0.295&0.634&0.567&0.939 &0\\
		&SRN-EGARCH*&\bol{0.200}&3.668&\bol{0.293}&\bol{0.623}&\bol{0.564}&\bol{0.930} &6\\
		\hline
		\rule{0pt}{3ex}
		&GARCH     &0.121&1.421&0.228&0.492&0.580&0.460&0 \\
		&SRN-GARCH* &\bol{0.108}&\bol{1.268}&0.215&0.455&0.575&0.420 &6\\
		\rule{0pt}{3ex}
		$\text{RKV}_2$&GJR &0.116&1.404&0.221&0.473&0.577&0.449 &0\\
		&SRN-GJR*   &\bol{0.108}&1.277&0.214&0.457&\bol{0.569}&0.411 &6\\
		\rule{0pt}{3ex}
		&EGARCH    &0.114&1.438&0.213&0.466&\bol{0.569}&0.405 &0\\
		&SRN-EGARCH*&\bol{0.108}&1.304&\bol{0.211}&\bol{0.454}&\bol{0.569}&\bol{0.400} &5\\
		\hline
		\rule{0pt}{3ex}
		&GARCH     &0.120&1.405&0.228&0.491&0.581&0.456&0 \\
		&SRN-GARCH* &\bol{0.107}&\bol{1.252}&0.214&0.453&0.576&0.415 &6\\
		\rule{0pt}{3ex}
		$\text{RKV}_3$&GJR  & 0.115&1.387&0.220&0.471&0.578&0.444 &0\\
		&SRN-GJR*   &\bol{0.107}&1.261&0.213&0.455&\bol{0.569}&0.407&6\\
		\rule{0pt}{3ex}
		&EGARCH    &0.113&1.421&0.212&0.465&0.571&0.401 &0\\
		&SRN-EGARCH*&\bol{0.107}&1.287&\bol{0.210}&\bol{0.452}&\bol{0.569}&\bol{0.396} &6\\
		\hline
		\rule{0pt}{3ex}
		&GARCH     &0.132&1.575&0.252&0.518&0.541&0.655&0 \\
		&SRN-GARCH* &\bol{0.118}&\bol{1.415}&0.238&0.484&0.530&0.600 &6\\
		\rule{0pt}{3ex}
		MedRV&GJR   &0.129&1.574&0.248&0.507&0.540&0.651 &0\\
		&SRN-GJR*   &\bol{0.118}&1.421&0.238&0.486&0.525&0.593 &6\\
		\rule{0pt}{3ex}
		&EGARCH    &0.125&1.636&0.238&0.502&0.524&0.580 &0\\
		&SRN-EGARCH*&\bol{0.118}&1.458&\bol{0.234}&\bol{0.482}&\bol{0.523}&\bol{0.576}&6\\
		\hline
		\rule{0pt}{3ex}
		&GARCH     &0.145&2.068&0.246&0.535&0.577&0.557 &0\\
		&SRN-GARCH* &\bol{0.132}&\bol{1.918}&0.233&0.500&0.572&0.510 &6\\
		\rule{0pt}{3ex}
		RV&GJR       &0.141&2.046&0.239&0.515&0.576&0.546 &0\\
		&SRN-GJR*   &0.133&1.927&0.232&0.501&\bol{0.567}&0.503&6\\
		\rule{0pt}{3ex}
		&EGARCH    &0.138&2.067&0.231&0.508&0.569&0.496 &0\\
		&SRN-EGARCH*&0.133&1.948&\bol{0.229}&\bol{0.498}&\bol{0.567}&\bol{0.490} &6\\
		\hline\hline
	\end{tabular}
	\caption{N225 data: Forecast performance of the RECH and benchmark GARCH-type models using different realized measures. For each pair of the RECH and GARCH-type models, the asterisks indicate the model with the higher count. In each panel, the bold numbers indicate the best predictive scores.}
	\label{tab:N225_forecast_score_compare}
\end{table}

\begin{table}[h!]
	\centering
	\footnotesize
	\begin{tabular}{crccccccc}
	\hline\hline
	\rule{0pt}{3ex}
	Estimator&& $\text{MSE}_1$ &$\text{MSE}_2$ & $\text{MAE}_1$&$\text{MAE}_2$&$\text{QLIKE}$&$\text{R}^2\text{LOG}$&Count\\
	\hline
	\rule{0pt}{3ex}
	&GARCH      &0.099&0.593&0.246&0.468&0.708&0.534&0 \\
	&SRN-GARCH* &0.087&0.509&0.232&0.440&0.687&0.482 &6\\
	\rule{0pt}{3ex}
	BV&GJR      &0.089&0.538&0.236&0.443&0.694&0.499 &0\\
	&SRN-GJR*   &0.087&0.509&0.233&0.442&0.686&0.686 &6\\
	\rule{0pt}{3ex}
	&EGARCH     &0.092&0.527&0.242&0.460&0.693&0.502 &0\\
	&SRN-EGARCH*&\bol{0.085}&\bol{0.505}&\bol{0.231}&\bol{0.436}&\bol{0.684}&\bol{0.475} &6\\
	\hline
	\rule{0pt}{3ex}
	&GARCH     &0.121&0.620&0.278&0.518&0.740&0.713 &0\\
	&SRN-GARCH* &0.108&0.543&0.266&0.493&0.720&0.651&6\\
	\rule{0pt}{3ex}
	$\text{RKV}_1$&GJR  &0.110&0.569&0.268&0.497&0.725&0.669 &0\\
	&SRN-GJR*   &0.109&0.549&0.266&0.497&0.718&0.652 &5\\
	\rule{0pt}{3ex}
	&EGARCH    &0.113&0.562&0.273&0.510&0.725&0.670 &0\\
	&SRN-EGARCH*&\bol{0.106}&\bol{0.540}&\bol{0.263}&\bol{0.491}&\bol{0.714}&\bol{0.640} &6\\
	\hline
	\rule{0pt}{3ex}
	&GARCH     &0.097&0.524&0.247&0.464&0.708&0.565&0 \\
	&SRN-GARCH* &\bol{0.087}&\bol{0.459}&0.234&0.437&0.690&0.514 &6\\
	\rule{0pt}{3ex}
	$\text{RKV}_2$&GJR &0.090&0.483&0.238&0.441&0.697&0.537 &0\\
	&SRN-GJR*   &0.088&0.464&0.235&0.440&0.690&0.520 &6\\
	\rule{0pt}{3ex}
	&EGARCH    &0.092&0.474&0.241&0.451&0.694& 0.534 &0\\
	&SRN-EGARCH*&\bol{0.087}&0.462&\bol{0.233}&\bol{0.435}&\bol{0.688}&\bol{0.512} &6\\
	\hline
	\rule{0pt}{3ex}
	&GARCH     &0.095&0.520&0.244&0.459&0.709&0.538&0 \\
	&SRN-GARCH* &0.085&\bol{0.456}&0.231&0.433&0.691&0.489 &6\\
	\rule{0pt}{3ex}
	$\text{RKV}_3$&GJR  & 0.088&0.478&0.234&0.436&0.698&0.509 &0\\
	&SRN-GJR*   &0.086&0.460&0.232&0.436&0.691&0.494&5\\
	\rule{0pt}{3ex}
	&EGARCH    &0.090&0.472&0.238&0.448&0.696&0.508 &0\\
	&SRN-EGARCH*&\bol{0.084}&0.459&\bol{0.230}&\bol{0.431}&\bol{0.690}&\bol{0.487} &6\\
	\hline
	\rule{0pt}{3ex}
	&GARCH     &0.137&0.950&0.289&0.553&0.669&0.819&0 \\
	&SRN-GARCH* &0.123&0.824&0.275&0.524&0.642&0.763 &6\\
	\rule{0pt}{3ex}
	MedRV&GJR   &0.127&0.879&0.278&0.527&0.654&0.784 &0\\
	&SRN-GJR*   &0.123&0.817&0.276&0.525&0.641& 0.764&6\\
	\rule{0pt}{3ex}
	&EGARCH    &0.128&0.839&0.284&0.543&0.645&0.784 &0\\
	&SRN-EGARCH*&\bol{0.119}&\bol{0.793}&\bol{0.274}&\bol{0.519}&\bol{0.638}&\bol{0.758}&6\\
	\hline
	\rule{0pt}{3ex}
	&GARCH     &0.100&0.562&0.250&0.475&0.716&0.546 &0\\
	&SRN-GARCH* &0.089&\bol{0.490}&0.235&0.443&0.698&0.492&6\\
	\rule{0pt}{3ex}
	RV&GJR       &0.092&0.513&0.239&0.449&0.704&0.510 &0\\
	&SRN-GJR*   &0.090&0.492&0.236&0.445&0.697&0.494&6\\
	\rule{0pt}{3ex}
	&EGARCH    &0.094&0.507&0.244&0.462&0.704&0.513 &0\\
	&SRN-EGARCH*&\bol{0.088}&0.492&\bol{0.234}&\bol{0.441}&\bol{0.695}&\bol{0.486} &6\\
	\hline\hline
	\end{tabular}
	\caption{RUT data: Forecast performance of the RECH and benchmark GARCH-type models using different realized measures. For each pair of the RECH and GARCH-type models, the asterisks indicate the model with the higher count. In each panel, the bold numbers indicate the best predictive scores.}
	\label{tab:RUT_forecast_score_compare}
\end{table}

\begin{table}[h!]
	\centering
	\footnotesize
	\begin{tabular}{crccccccc}
		\hline\hline
		\rule{0pt}{3ex}
		Estimator&& $\text{MSE}_1$ &$\text{MSE}_2$ & $\text{MAE}_1$&$\text{MAE}_2$&$\text{QLIKE}$&$\text{R}^2\text{LOG}$&Count\\
		\hline
		\rule{0pt}{3ex}
		&GARCH     &0.083&0.572&0.216&0.404&0.618&0.449&0 \\
		&SRN-GARCH* &0.078&\bol{0.542}&0.211&0.399&0.611&0.419 &6\\
		\rule{0pt}{3ex}
		BV&GJR*       &\bol{0.075}&\bol{0.542}&\bol{0.205}&\bol{0.379}&0.609&0.417 &4\\
		&SRN-GJR   &0.078&0.545&0.210&0.399&0.606&0.412 &2\\
		\rule{0pt}{3ex}
		&EGARCH    &0.078&0.547&0.210&0.401&0.602&0.408 &1\\
		&SRN-EGARCH*&0.077&0.548&0.207&0.395&\bol{0.601}&\bol{0.401} &5\\
		\hline
		\rule{0pt}{3ex}
		&GARCH     &0.103&0.691&0.245&0.452&0.628&0.609 &0\\
		&SRN-GARCH* &0.098&0.672&0.240&0.448&0.621&0.567 &6\\
		\rule{0pt}{3ex}
		$\text{RKV}_1$&GJR* &\bol{0.094}&\bol{0.664}&\bol{0.234}&\bol{0.426}&0.619&0.570 &4\\
		&SRN-GJR   &0.098&0.676&0.239&0.448&0.616&0.560 &2\\
		\rule{0pt}{3ex}
		&EGARCH    &0.099&0.678&0.240&0.452&0.613&0.556 &1\\
		&SRN-EGARCH*&0.097&0.681&0.237&0.447&\bol{0.612}&\bol{0.549} &5\\
		\hline
		\rule{0pt}{3ex}
		&GARCH     &0.067&0.396&0.197&0.368&0.623&0.354&0 \\
		&SRN-GARCH* &0.064&0.375&0.192&0.363&0.621&0.333 &6\\
		\rule{0pt}{3ex}
		$\text{RKV}_2$&GJR* &\bol{0.060}&\bol{0.369}&\bol{0.185}&\bol{0.341}&0.614&0.324 &4\\
		&SRN-GJR   &0.064&0.379&0.191&0.363&0.615&0.326 &2\\
		\rule{0pt}{3ex}
		&EGARCH    &0.064&0.380&0.191&0.366&\bol{0.610}&0.320 &1\\
		&SRN-EGARCH*&0.063&0.382&0.188&0.359&\bol{0.610}&\bol{0.314} &4\\
		\hline
		\rule{0pt}{3ex}
		&GARCH     &0.068&0.410&0.198&0.370&0.623&0.362&0 \\
		&SRN-GARCH* &0.065&0.390&0.193&0.366&0.621&0.340 &6\\
		\rule{0pt}{3ex}
		$\text{RKV}_3$&GJR*  & \bol{0.061}&\bol{0.384}&\bol{0.186}&\bol{0.343}&\bol{0.614}&0.331 &6\\
		&SRN-GJR   &0.065&0.394&0.193&0.367&0.615& 0.333&0\\
		\rule{0pt}{3ex}
		&EGARCH    &0.065&0.395&0.193&0.369&\bol{0.610}&0.327 &1\\
		&SRN-EGARCH*&0.064&0.397&0.189&0.363&\bol{0.610}&\bol{0.321} &4\\
		\hline
		\rule{0pt}{3ex}
		&GARCH     &0.084&0.525&0.223&0.416&0.586&0.490&0 \\
		&SRN-GARCH* &0.077&0.460&0.215&0.402&0.574&0.454 &6\\
		\rule{0pt}{3ex}
		MedRV&GJR       &0.077&0.490&0.213&\bol{0.392}&0.578&0.462 &2\\
		&SRN-GJR*   &0.076&0.462&0.214&0.402&0.569&0.448 &4\\
		\rule{0pt}{3ex}
		&EGARCH    &0.076&0.458&0.213&0.403&\bol{0.562}&0.436 &1\\
		&SRN-EGARCH*&\bol{0.074}&\bol{0.450}&\bol{0.211}&0.396&0.563&\bol{0.433}&5\\
		\hline
		\rule{0pt}{3ex}
		&GARCH     &0.087&0.637&0.218&0.410&0.621&0.456&0\\
		&SRN-GARCH* &0.083&0.619&0.213&0.407&0.616&0.426 &6\\
		\rule{0pt}{3ex}
		RV&GJR*       &\bol{0.079}&\bol{0.613}&\bol{0.206}&\bol{0.384}&0.613&0.423 &4\\
		&SRN-GJR   &0.083&0.624&0.214&0.408&0.612&0.421&2\\
		\rule{0pt}{3ex}
		&EGARCH    &0.083&0.624&0.214&0.412&0.608&0.415 &1\\
		&SRN-EGARCH*&0.082&0.627&0.210&0.405&\bol{0.607}&\bol{0.409} &5\\
		\hline\hline
	\end{tabular}
	\caption{DAX data: Forecast performance of the RECH and benchmark GARCH-type models using different realized measures. For each pair of the RECH and GARCH-type models, the asterisks indicate the model with the higher count. In each panel, the bold numbers indicate the best predictive scores.}
	\label{tab:DAX_forecast_score_compare}
\end{table}

In particular, for the SP500, N225 and RUT datasets, the RECH models consistently perform the best across all panels. For example, the forecast results in Table \ref{tab:SP500_forecast_score_compare} show that the SRN-GARCH model has the highest numbers of lowest forecast errors in all panels, implying that the SRN-GARCH model forecasts volatility the best for the SP500 data.
For the N225 and RUT datasets, the SRN-EGARCH model clearly outperforms the other RECH and benchmark models in all realized measures. The superior predictive performance of the RECH models over the GARCH-type counterparts provides further evidence to support the conclusion that the RECH models do not overfit the index datasets.

As mentioned in Section \ref{sec:RECH_specifications}, we now discuss an useful feature of the RECH models; that is, the volatility estimates and volatility forecasts of the RECH specifications are less sensitive to the choice of the structure for the GARCH component than a single GARCH-type model. For example, for each dataset in Table \ref{tab:Real_params} and \ref{tab:Real_params2}, we compute the difference between the highest and lowest marginal likelihood estimates among the RECH specifications and calculate the same value for the GARCH-type benchmark models. Table \ref{tab:marllh_difference} shows that these discrepancies of in-sample performance among the RECH models are much smaller than those of the GARCH-type models, across all datasets. For each panel in Tables \ref{tab:SP500_forecast_score_compare} to \ref{tab:DAX_forecast_score_compare}, we compute the difference between the highest and lowest forecast scores among the RECH models and do the same for the benchmark models; Figure \ref{f:Application_forecast_discrepancy} plots the results.

\begin{table}[h!]
	\begin{center}
		\footnotesize
		\begin{tabular}{ccccc}
			\hline\hline
			\rule{0pt}{3ex}
			&SP500         &N225      &RUT         &DAX\\
			\hline
			\rule{0pt}{3ex}
            GARCH-type models &21.7& 30.3 &12.6 & 29.8 \\
            RECH models       &6.9 & 6.8  &5.5  & 3.3 \\
			\hline\hline
		\end{tabular}
	\end{center}
	\caption{Applications: The difference between the highest and lowest marginal likelihood estimates of the RECH and the benchmark models across all in-sample data. The numbers are in the natural log scale. }
	\label{tab:marllh_difference}
\end{table}
\begin{figure}[h!]
	\centering
	\includegraphics[width=1\columnwidth]{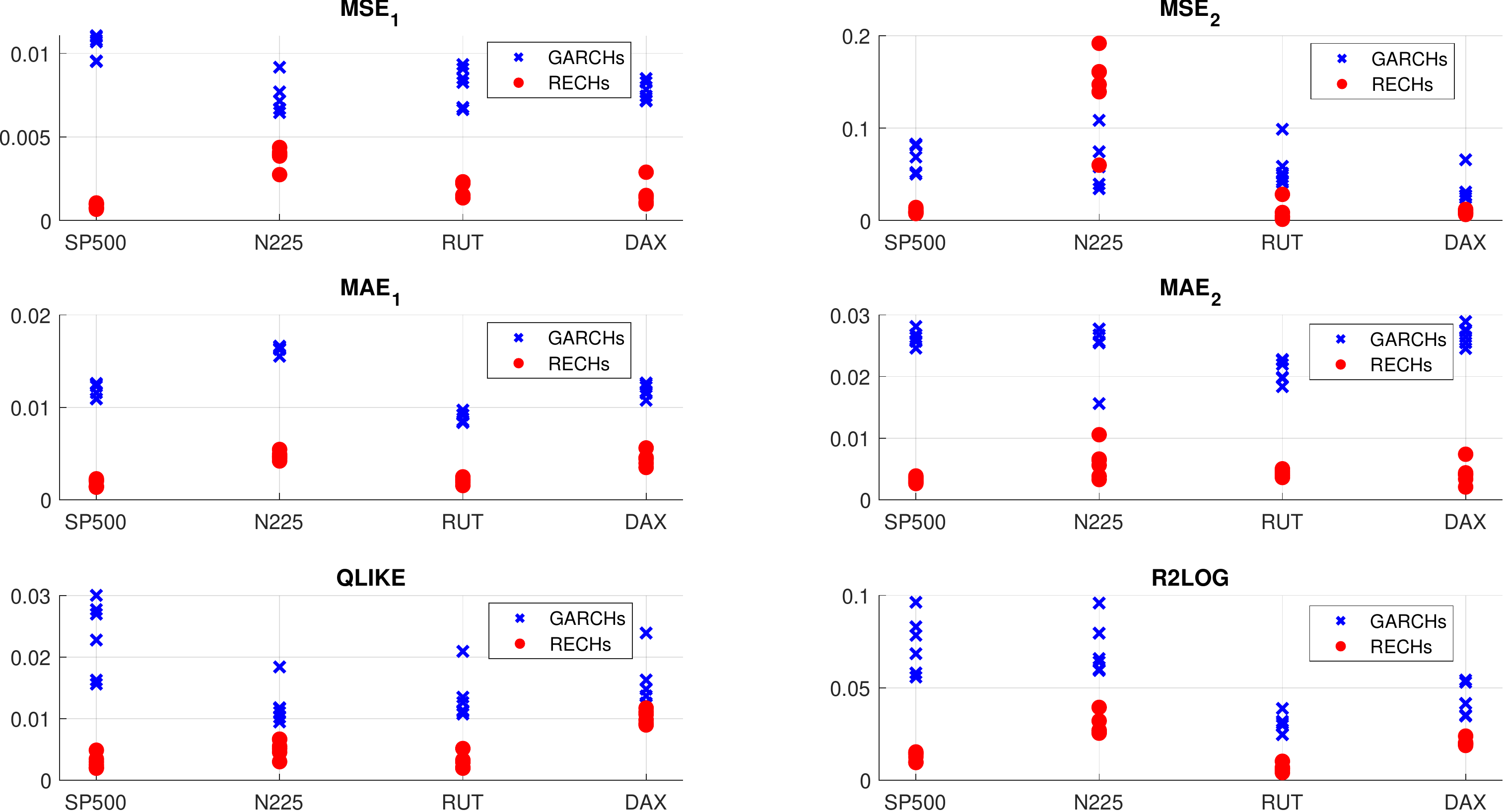}
	\caption{Applications: The difference between the highest and lowest forecast scores of the RECH and the benchmark models. In each panel, each column contains 6 values of the RECH models and 6 values of the benchmark models, corresponding to 6 realized measures. }
	\label{f:Application_forecast_discrepancy}
\end{figure}

The comparison results in Figure \ref{f:Application_forecast_discrepancy} indicate that the discrepancies of out-of-sample performance among the RECH models are consistently lower than those of the GARCH-type models, except for the $\text{MSE}_2$ score for the N225 data.
The result that model performance is less sensitive to the RECH specification is useful in practice as users do not need to worry about which specification should be used for their financial dataset.

\subsection{Application to exchange rate data}
This section reports on the application of the RECH model to analyse the USD/GBP daily exchange rates observed from 21/03/2001 to 01/03/2009\footnote{The dataset was also downloaded from the Realized Library of The Oxford-Man Institute}. We use the first $1000$ observations for model estimation and the last $1000$ observations for evaluating predictive performance.

The in-sample results (not shown) suggest that, unlike for the stock data, for exchange rate data GARCH performs the best compared to GJR and EGARCH.
This is consistent with the study of \cite{Hansen:2005} who show that there is no evidence that GARCH is outperformed by more sophisticated models on the DM/USD exchange rates.
The in-sample result also shows that SRN-GARCH outperforms GARCH, thus being the best model, in terms of marginal likelihood.
\begin{table}[h!]
	\begin{center}
		\footnotesize
		\begin{tabular}{rccccccccccc}
			\hline\hline
			\rule{0pt}{3ex}
			&PPS   &\#Vio&QS &\%Hit     &$\text{MSE}_1$   &$\text{MSE}_2$  &$\text{MAE}_1$   &$\text{MAE}_2$ &QLike& $\text{R}^2\text{Log}$\\
			\hline
			\rule{0pt}{3ex}
			GARCH       &\bol{0.827} &17 &0.018 &0.017  &0.029 &0.228 &0.099 &0.169 &-0.121 &0.166\\
			SRN-GARCH*  &\bol{0.827} &\bol{14} &\bol{0.017} &0.017  &\bol{0.024} &0.180 &\bol{0.094} &\bol{0.154} &\bol{-0.128} &\bol{0.159}\\
			\rule{0pt}{3ex}
			GJR         &0.829 &17 &0.018 &0.017 &0.029 &0.225 &0.100 &0.169 &-0.121&0.167\\
			SRN-GJR*    &0.829 &15 &\bol{0.017} &0.017 &\bol{0.024} &\bol{0.178} &\bol{0.094} &\bol{0.154} &-0.127&0.160\\
			\rule{0pt}{3ex}
			EGARCH      &0.836 &17 &0.019 &\bol{0.016} &0.039 &0.297 &0.112 &0.196 &-0.097 &0.198\\
			SRN-EGARCH* &0.843 &15 &\bol{0.017} &0.017 &0.028 &0.194 &0.107 &0.178 &-0.104 &0.202\\
			\hline\hline
		\end{tabular}
	\end{center}
	\caption{USD/GBP exchange rate: one-step-ahead forecast comparison. 
		The bold numbers denote the best scores. For each pair of the RECH and GARCH-type models, the asterisk indicates the models having better forecast performance.}
	\label{tab:exchange_forecast}
\end{table}

Table \ref{tab:exchange_forecast} summarizes the forecast performance measured by the predictive scores discussed in Section \ref{sec:simulation and applications},
which suggests that RECH models are able to improve on their counterpart GARCH-type models in terms of volatility forecasts.

\section{Conclusion}\label{sec:conclusion}
We propose a new class of conditional heteroskedastic models, which we call RECH models, by incorporating a RNN structure into the conditional variance of the GARCH-type models, and study in detail three RECH specifications: SRN-GARCH, SRN-GJR and SRN-EGARCH.
We use Sequential Monte Carlo with likelihood annealing and data annealing for in-sample Bayesian inference and out-of-sample forecasting. We also use the estimate of marginal likelihood as a by-product of the SMC for model choice. The extensive simulation and empirical studies suggest that the RECH models not only have both attractive in-sample performance and accurate out-of-sample forecasts, but can also explain the volatility movement.
In addition to the empirical study reported in Section \ref{sec:simulation and applications},
we tested the RECH models on all of the other datasets included in the Realized Library which contains 31 major stock markets around the world.
In all cases, with the adding of the RNN component, the RECH model is not worse than its GARCH counterpart,
whereas in some of these cases, significant predictive improvement is achieved by RECH.

An attractive feature of the proposed hybrid framework is that it is easy to use advances in both the deep learning and volatility modeling literatures to extend the current RECH models.
This opens up many interesting future applications and areas of research.
For example, one can use the Fourier Recurrent Unit (FRU) of \cite{zhang:2018} to construct the recurrent component of the RECH framework; FRU is currently considered as the state-of-the-art RNN architecture in deep learning. For the GARCH component, one can use the Bad Environment - Good Environment (BEGE) model of \cite{Bekaert:2015}, which can efficiently simulate the heavy tailed behavior of financial returns.
Another interesting research direction is extending the univariate RECH models to the multivariate case.
We conjecture that the recurrent neural network architectures will be more powerful for multivariate inputs as they can naturally capture the interaction between the inputs. This research is in progress.

\newpage
\section*{Appendix}
\subsection*{A1:  SMC with data annealing}
\label{sec:data_anneal}
\begin{algorithm}[h!]
	\caption{SMC with data annealing}
	\label{alg:data_annealing}
	1. Sample  $\theta^j_0 \sim p(\theta)$ and set  $W_0^j=1/M$ for $j=1...M$ \\
	2. \textbf{For} $t=1,...,T$,
	\begin{itemize}
		\item[] \textbf{Step 1, reweighting:} Compute the unnormalized weights
		\bea
		w_t^j = W^j_{t-1}p(y_t|y_{1:t-1},\theta^j_{t-1}), \; \; j=1,...,M,
		\eea
		and set the new normalized weights as
		\bea
		W^j_t = \frac{w^j_t}{\sum_{s=1}^{M}w^s_t}, \; \; j=1,...,M.
		\eea			
		\item[] \textbf{Step 2: } Compute the effective sample size (ESS)
		\bea
		\text{ESS} = \frac{1}{\sum_{j=1}^{M} \left(W_t^j\right)^2}.
		\eea
		\begin{itemize}
			\item[] \textbf{if} $\text{ESS} < c M$ for some $0<c<1$, \textbf{then}
			\begin{itemize}
				\item[(i)] {\bf Resampling}: Resample from $\{\theta_{t-1}^j,W_{t}^j\}^M_{j=1}$, and then set $W_t^j=1/M$ for $j=1...M$, to obtain the new equally-weighted particles $\{\theta_{t}^j,W_{t}^j\}^M_{j=1}$.
				\item[(ii)] {\bf Markov move}: for each $j=1,...,M$, move the sample $\theta_t^j$ according to $N_{\text{data}}$ random walk Metropolis-Hasting steps:
				\begin{itemize}
					\item[(a)] Generate a proposal $\theta_t^{j\prime}$ from multivariate normal distribution $\N(\theta_t^j,\Sigma_t)$ with $\Sigma_t$ the covariance matrix.
					\item[(b)] Set $\theta_t^j = \theta_t^{j \prime}$ with the probability
					\bea
					\text{min}\left(1,\frac{p(y_{1:t}|\theta_t^{j \prime})p(\theta_t^{j \prime})}{p(y_{1:t}|\theta_t^{j})p(\theta_t^{j})}\right)
					\eea
					otherwise keep $\theta_t^j$.
				\end{itemize}
			\end{itemize}
			\textbf{end}
		\end{itemize}
	\end{itemize}
\end{algorithm}

\newpage
\subsection*{A2: Additional results for Section \ref{sec:simulation and applications}}
\begin{figure}[h!]
	\centering
	\includegraphics[width=0.95\columnwidth]{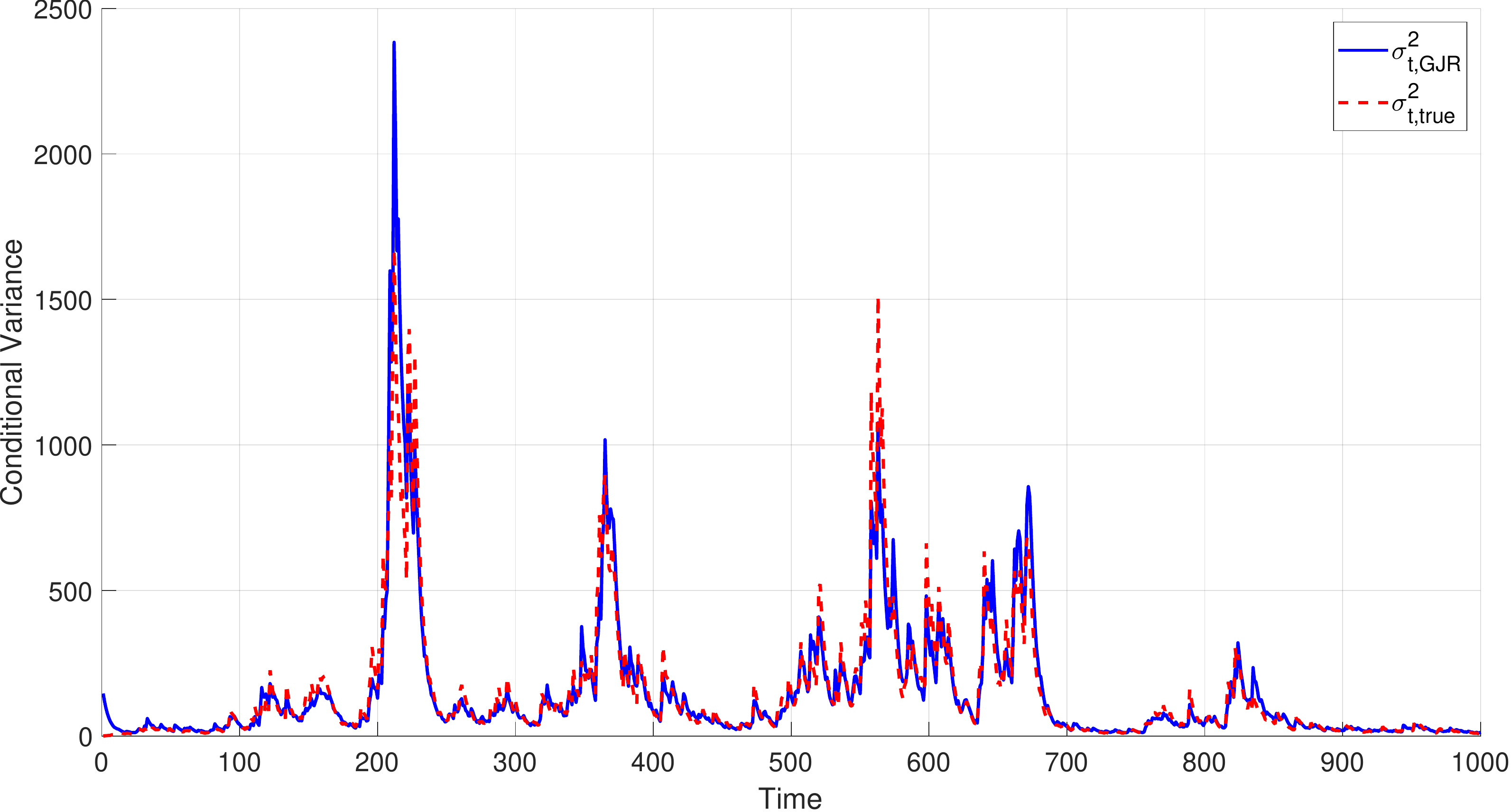}
	\caption{SIM II: The true conditional variance (dashed line) and estimated conditional variance (solid line) using the GJR model. (The figure is better viewed in colour).}
	\label{f:Sim2_volatility_density_2}
\end{figure}

\begin{figure}[h!]
	\centering
	\includegraphics[width=0.95\columnwidth]{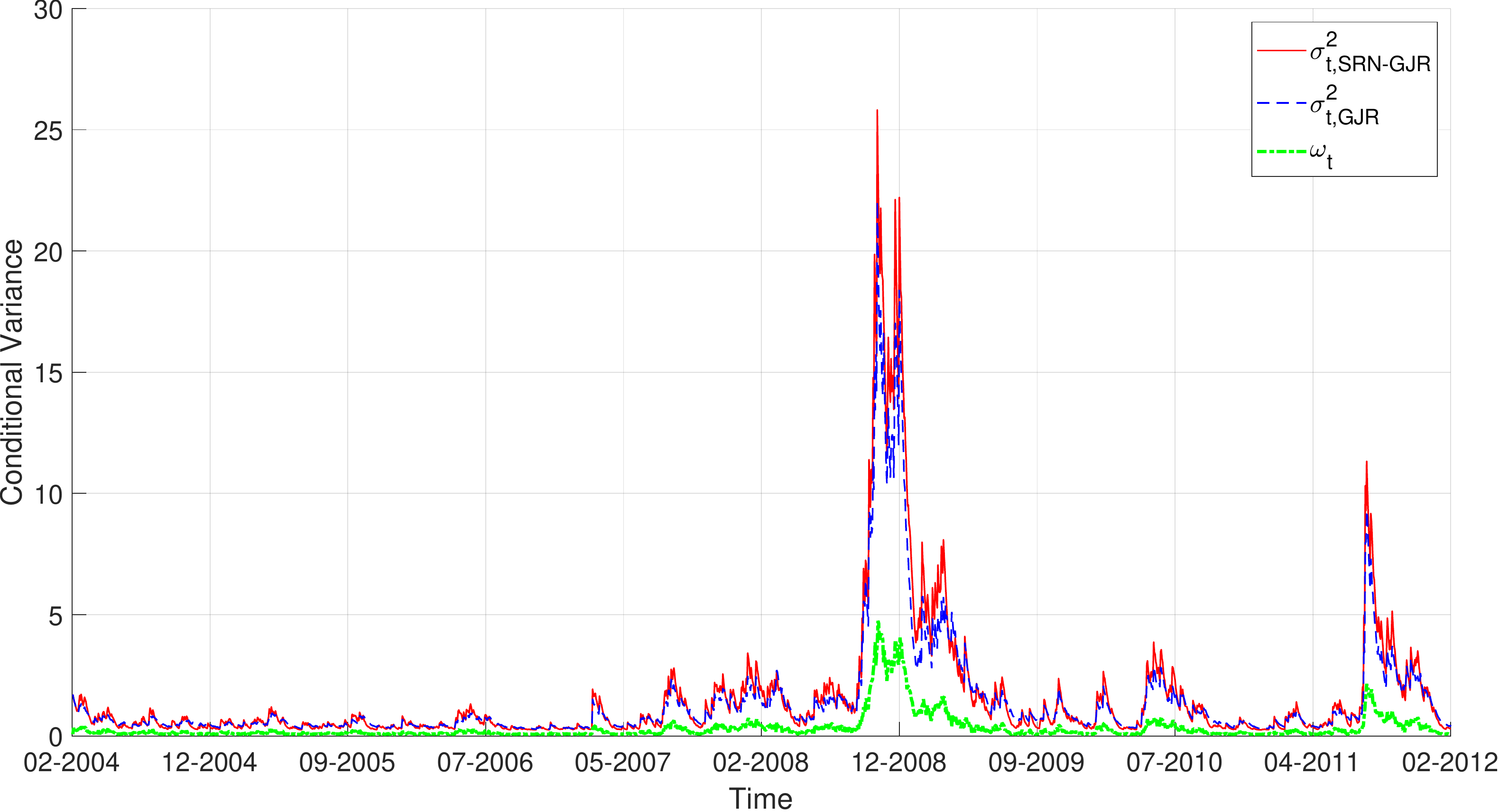}
	\caption{SP500: The in-sample conditional variance of the GJR (dashed line) and SRN-GJR (solid line) at all time steps. The bottom line shows the values of the recurrent component $\omega_t$ of the SRN-GJR specification. (The figure is better viewed in colour).}
	\label{f:SP500_volatility_density_SRN_GJR}
\end{figure}

\begin{figure}[h!]
	\centering
	\includegraphics[width=0.95\columnwidth]{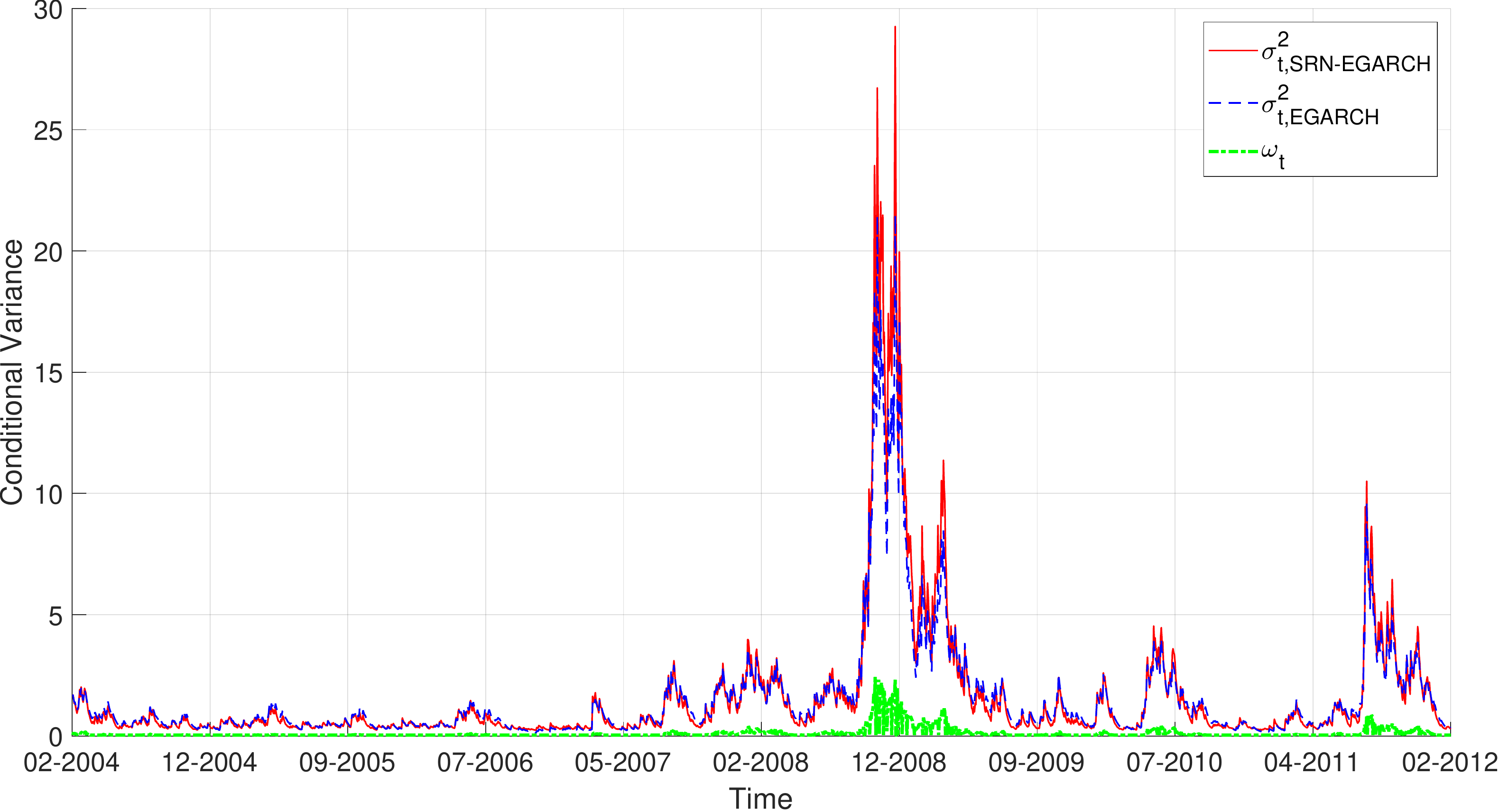}
	\caption{SP500: The in-sample conditional variance of the EGARCH (dashed line) and SRN-EGARCH (solid line) at all time steps. The bottom line shows the values of the recurrent component $\omega_t$ of the SRN-EGARCH specification. (The figure is better viewed in colour).}
	\label{f:SP500_volatility_density_SRN_EGARCH}
\end{figure}

\begin{figure}[h!]
	\begin{center}
		\includegraphics[width=0.95\columnwidth]{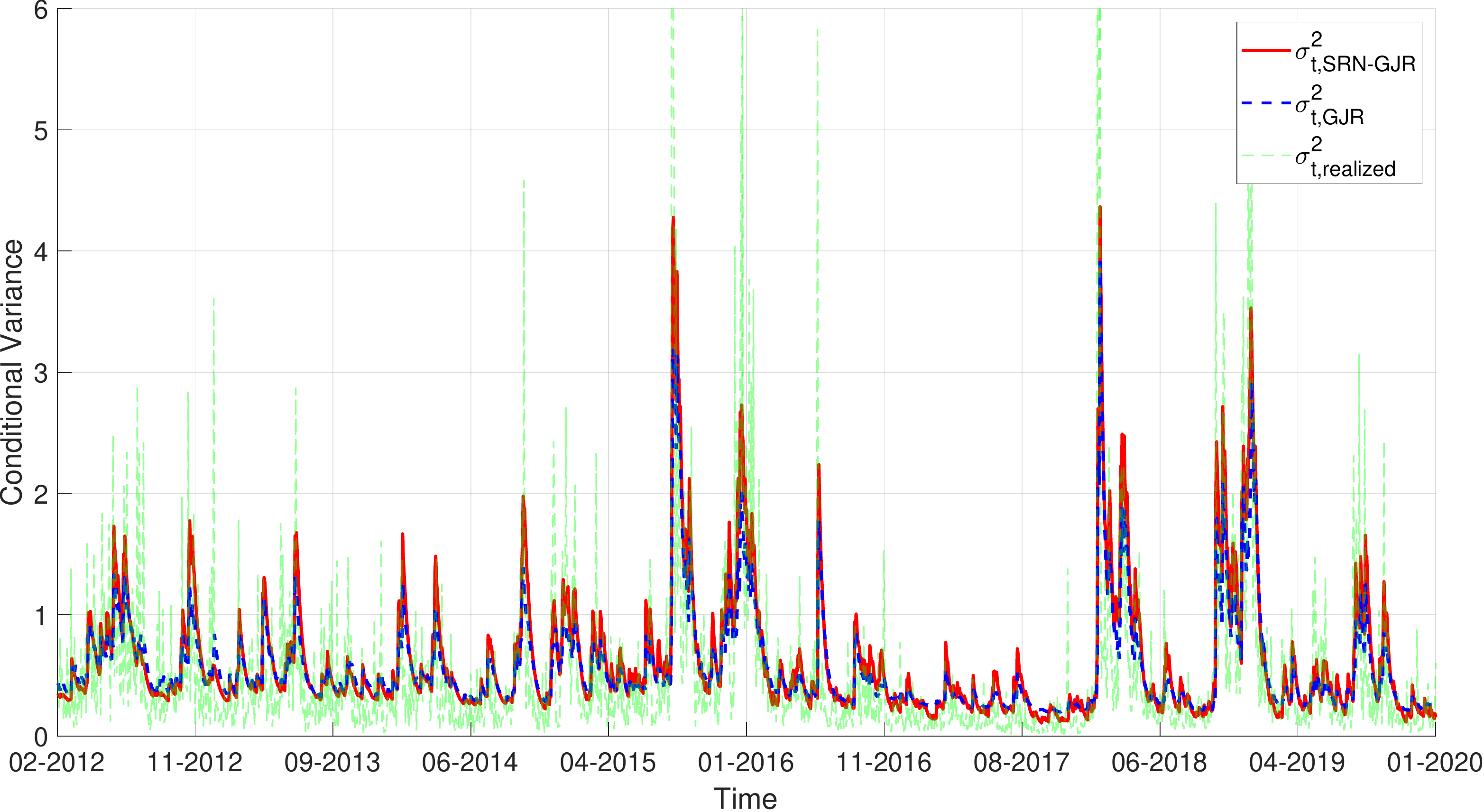}
		\caption{SP500: Forecast conditional variance by the GJR (dashed) and SRN-GJR (solid) models, together with the realized variance (dotted). (The figure is better viewed in colour).}
		\label{f:SP500_forecast_variance_GJR}
	\end{center}
\end{figure}

\begin{figure}[ht!]
	\begin{center}
		\includegraphics[width=0.95\columnwidth]{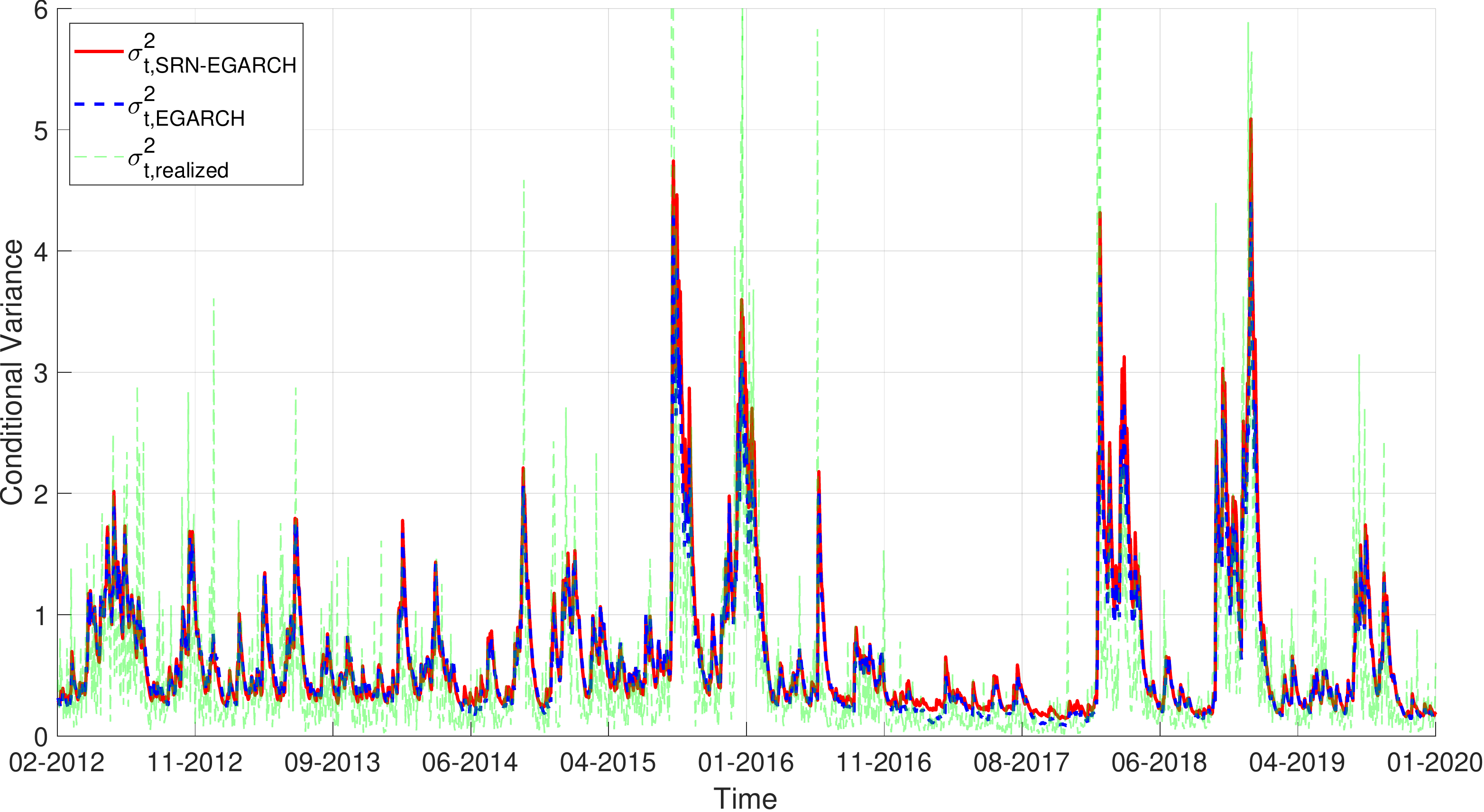}
		\caption{SP500: Forecast conditional variance by the EGARCH (dashed) and SRN-EGARCH (solid) models, together with the realized variance (dotted). (The figure is better viewed in colour).}
		\label{f:SP500_forecast_variance_EGARCH}
	\end{center}
\end{figure}

\newpage
\subsection*{A3:  Comparison to the LSTM-SV and GP-Vol models}

This section reports the comparison of the predictive performance between the RECH and other non-linear volatility models including the LSTM-SV model of \cite{Nguyen:2019} and the GP-Vol model of \cite{Lobato:2014}. The GP-Vol model is an engineering-oriented stochastic volatility model using a Gaussian process to capture the non-linearity in the volatility dynamics. As the GP-Vol model focuses mainly on prediction, we use it as another benchmark model, together with the LSTM-SV model, to assess the predictive performance of the RECH models. We use the software packages provided by \cite{Nguyen:2019} and \cite{Lobato:2014} to perform Bayesian inference and prediction for the LSTM-SV and GP-Vol models, respectively, with all settings at their default values. For all models, we use the posterior means estimated from in-sample data to perform one-step-ahead forecasting for out-of-sample data. We use $T_{\text{out}}=1000$ observations for the out-of-sample period.

Table \ref{tab:LSTM_SV_forecast} shows the out-of-sample performance of the models evaluated on the SP500 index, using the same realize measures discussed in Section \ref{sec:applications}.
Table \ref{tab:LSTM_SV_forecast} suggests that RECH models in general predict better than  the LSTM-SV and GP-Vol models. The superiority of the SRN-GARCH model, and RECH models in general, over the LSTM-SV model is expected as the SRN component in RECH models is able to capture the leverage effect, while this is not the case for the LSTM component in LSTM-SV. We note that the GP-Vol model uses a Gaussian process with its covariance matrix expanding over time and hence it becomes computationally expensive in applications with long time series. We observe similar results between the RECH, LSTM-SV and GP-Vol models for the other datasets in Section \ref{sec:applications}.

\begin{table}[h!]
	\centering
	\footnotesize
	\begin{tabular}{crcccccc}
		\hline\hline
		\rule{0pt}{3ex}
		Estimator& & $\text{MSE}_1$ &$\text{MSE}_2$ & $\text{MAE}_1$&$\text{MAE}_2$&$\text{QLIKE}$&$\text{R}^2\text{LOG}$\\
		\hline
		\rule{0pt}{3ex}
		\multirow{5}{*}{BV}&GP-Vol     &0.176  &1.586  &0.309  &0.555   &0.546  &1.097   \\
		\rule{0pt}{3ex}
		&LSTM-SV           &0.104  &1.326  &0.237  &0.397  &0.354  &0.717    \\
		\rule{0pt}{3ex}
		&SRN-GARCH*         &0.105  &1.210  &\bol{0.201}  &\bol{0.332}  &0.448  &\bol{0.572}  \\
		&SRN-GJR*           &\bol{0.096}  &\bol{1.210}  &0.220 &0.370  &\bol{0.288}  &0.637 \\
		&SRN-EGARCH        &0.100  &1.222  &0.228  &0.383  &0.297  &0.681 \\
		\hline
		\rule{0pt}{3ex}
		\multirow{5}{*}{$\text{RKV}_1$}&GP-Vol      &0.173  &1.039  &0.317  &0.549  &0.501  &1.260  \\
		\rule{0pt}{3ex}
		&LSTM-SV     &0.115  &\bol{0.829}  &0.263  &0.425  &\bol{0.366}  &0.978    \\
		\rule{0pt}{3ex}
		&SRN-GARCH*       &0.129  &1.059  &\bol{0.240}  &\bol{0.395}  &0.598  &\bol{0.772}  \\
		&SRN-GJR         &\bol{0.104}  &0.875  &\bol{0.240}  &0.400  &0.369  &0.785 \\
		&SRN-EGARCH      &0.107  &0.882  &0.246  &0.411  &0.373  &0.826 \\
		\hline
		\rule{0pt}{3ex}
		\multirow{5}{*}{$\text{RKV}_2$}&GP-Vol  &0.154  &0.851  &0.295  &0.522  &0.603  &0.955  \\
		\rule{0pt}{3ex}
		&LSTM-SV     &0.078  &0.524  &0.213  &0.349  &0.383  &0.567    \\
		\rule{0pt}{3ex}
		&SRN-GARCH*       &\bol{0.070}  &\bol{0.440}  &\bol{0.179}  &\bol{0.271}  &0.307  &\bol{0.471}  \\
		&SRN-GJR         &0.077  &\bol{0.440}  &0.209  &0.332  &\bol{0.224}  &0.587 \\
		&SRN-EGARCH      &0.083  &0.457  &0.220  &0.350  &0.236  &0.635 \\
		\hline
		\rule{0pt}{3ex}
		\multirow{5}{*}{$\text{RKV}_3$}&GP-Vol      &0.155  &0.893  &0.296  &0.523  &0.595  &0.963 \\
		\rule{0pt}{3ex}
		&LSTM-SV    &0.080  &0.574  &0.214  &0.352  &0.379  & 0.575   \\
		\rule{0pt}{3ex}
		&SRN-GARCH*       &\bol{0.074}  &0.499  &\bol{0.182}  &\bol{0.278}  &0.326  &\bol{0.479}  \\
		&SRN-GJR         &0.077  &\bol{0.480}  &0.208  &0.332  &\bol{0.232}  &0.579 \\
		&SRN-EGARCH      &0.082  &0.496  &0.219  &0.350  &0.244  &0.626 \\
		\hline
		\rule{0pt}{3ex}
		\multirow{5}{*}{MedRV} &GP-Vol      &0.194  &1.147  &0.333  &0.585  &0.559  &1.329 \\
		\rule{0pt}{3ex}
		&LSTM-SV      &0.117  &0.830  &0.266  &0.437  &0.336  &0.936    \\
		\rule{0pt}{3ex}
		&SRN-GARCH*       &\bol{0.103}  &0.734  &\bol{0.209}  &\bol{0.327}  &0.303  &\bol{0.703}  \\
		&SRN-GJR         &0.106  &\bol{0.639}  &0.253  &0.400  &\bol{0.196}  &0.927 \\
		&SRN-EGARCH      &0.112  &0.645  &0.263  &0.414  &0.210  &0.990 \\
		\hline
		\rule{0pt}{3ex}
		\multirow{5}{*}{RV} &GP-Vol      &0.186  &2.048  &0.318  &0.571  &0.488  &1.189\\
		\rule{0pt}{3ex}
		&LSTM-SV     &0.123  &\bol{1.865}  &0.255  &0.429  &0.342  &0.861    \\
		\rule{0pt}{3ex}
		&SRN-GARCH   &0.143  &2.366  &0.228  &\bol{0.404} &\bol{0.337}  &0.675 \\
		&SRN-GJR*     &\bol{0.114}  &2.132  &\bol{0.227}  &0.406 &0.381  &\bol{0.658} \\
		&SRN-EGARCH  &0.117  &2.144  &0.236  &0.420 &0.387  &0.696 \\
		\hline\hline
	\end{tabular}
	\caption{SP500: Out-of-sample performance of the GP-Vol, LSTM-SV and RECH models using different realized measures. In each panel, the bold numbers indicate the best predictive scores and the asterisk indicates the model with the best predictive performance. }
	\label{tab:LSTM_SV_forecast}
\end{table}

\subsection*{A4: Proofs}
\begin{proof}[Proof of Theorem \ref{the:Theorem}]
Recall the $\sigma$-fields $\mathcal F_t=\sigma(y_s,s\leq t)$, $t\geq 1$, and let us define $\mathcal F_0$ to be the $\sigma$-field generated by $\sigma_0^2$.
As the recurrent component is bounded,
\beq\label{eq:basic eq}
\E(y_t^2|\F_{t-1})=\sigma_t^2\leq M+\alpha\sigma_{t-1}^2+\beta y_{t-1}^2,\;\;t>1.
\eeq
We have that
\bean
\E(y_t^2|\F_{t-2})&=&\E\big(\E(y_t^2|\F_{t-1})|\F_{t-2}\big)\\
&\leq& M+\alpha\sigma_{t-1}^2+\beta \E(y_{t-1}^2|\F_{t-2})\\
&=& M+\alpha\sigma_{t-1}^2+\beta \sigma_{t-1}^2,
\eean
hence, by \eqref{eq:basic eq},
\beqn
\E(y_t^2|\F_{t-2})\leq\begin{cases}
M+(\alpha+\beta)\sigma_{1}^2<M+\sigma_{0}^2,&t=2\\
M+(\alpha+\beta)\big(M+\alpha\sigma_{t-2}^2+\beta y_{t-2}^2\big),&t>2.
\end{cases}
\eeqn
Similarly,
\bean
\E(y_t^2|\F_{t-3})&=&\E\big(\E(y_t^2|\F_{t-2})|\F_{t-3}\big)\\
&\leq& M+(\alpha+\beta)\big(M+\alpha\sigma_{t-2}^2+\beta \E(y_{t-2}^2|\F_{t-3})\big)\\
&=& M\big(1+(\alpha+\beta)\big)+(\alpha+\beta)^2\sigma_{t-2}^2,\;\;t\geq 3.
\eean
For $t=3$,
\[\E(y_t^2|\F_{t-3})=M\big(1+(\alpha+\beta)\big)+(\alpha+\beta)^2\sigma_{1}^2< M\big(1+(\alpha+\beta)\big)+\sigma_{0}^2,\]
and for $t>3$, by \eqref{eq:basic eq},
\bean
\E(y_t^2|\F_{t-3})&\leq& M\big(1+(\alpha+\beta)\big)+(\alpha+\beta)^2\big(M+\alpha\sigma_{t-3}^2+\beta y_{t-3}^2\big)\\
&\leq& M\big(1+(\alpha+\beta)+(\alpha+\beta)^2\big)+(\alpha+\beta)^2\big(\alpha\sigma_{t-3}^2+\beta y_{t-3}^2\big).
\eean
Hence
\beqn
\E(y_t^2|\F_{t-3})\leq\begin{cases}
M\big(1+(\alpha+\beta)\big)+\sigma_{0}^2,&t=3\\
M\big(1+(\alpha+\beta)+(\alpha+\beta)^2\big)+(\alpha+\beta)^2\big(\alpha\sigma_{t-3}^2+\beta y_{t-3}^2\big),&t>3.
\end{cases}
\eeqn
By deduction we have that, for $k=2,...,t-1$,
\beq
\E(y_t^2|\F_{t-k})\leq
\begin{cases}
M\big(1+\sum_{i=1}^{k-2}(\alpha+\beta)^i\big)+\sigma_0^2,&t=k\\
M\big(1+\sum_{i=1}^{k-1}(\alpha+\beta)^i\big)+(\alpha+\beta)^{k-1}\big(\alpha\sigma_{t-k}^2+\beta y_{t-k}^2\big),&t>k.
\end{cases}
\eeq
Therefore,
\bean
\V(y_t^2|\sigma_0^2)\leq\E(y_t^2|\F_0)&\leq&M\big(1+\sum_{i=1}^{k-2}(\alpha+\beta)^i\big)+\sigma_0^2\\
&<&M\big(1+\sum_{i=1}^{\infty}(\alpha+\beta)^i\big)+\sigma_0^2=\frac{M}{1-\alpha-\beta}+\sigma_0^2.
\eean
\end{proof}

\clearpage
\bibliographystyle{apalike}
\bibliography{references}
\end{document}